\pdfoutput=1
\ifx\synctex\undefined\else\synctex=1\fi
\documentclass[final]{llncs}

\usepackage[utf8x]{inputenc}

\usepackage{microtype}

\usepackage{amsmath} 
\usepackage{amsxtra} 
\usepackage{amssymb}
\usepackage{mathabx}
\usepackage{mathtools}
\usepackage{textcomp}
\usepackage{bm}

\usepackage{url}

\usepackage{commenting}

\usepackage{paralist}

\usepackage{graphicx}

\usepackage{algorithm}
\usepackage{algorithmicx}
\usepackage[noend]{algpseudocode}
\usepackage{caption}

\spnewtheorem{ass}{Assumption}{\bfseries}{\itshape}
\spnewtheorem{fact}{Fact}{\bfseries}{\itshape}

\renewcommand{\iff}{\Leftrightarrow}

\newcommand{\rbr}{{\bf ]\!\!]}}
\newcommand{\lbr}{{\bf [\!\![}}
\newcommand{\sem}[1]{\lbr #1 \rbr}

\renewcommand{\vec}[1]{{\bf {#1}}}

\newcommand{\arrow}[2]{\xrightarrow{{\scriptscriptstyle #1}}}
\newcommand{\Arrow}[1]{\xRightarrow{{\scriptscriptstyle #1}}}
\newcommand{\xArrow}[2]{\xRightarrow[{\scriptscriptstyle #2}]{{\scriptscriptstyle #1}}}

\newcommand{\nat}{{\bf \mathbb{N}}}

\def\age#1{\left[#1\right]}
\def\set#1{{\left\{ #1 \right\}}}

\def\tuple#1{{\langle #1 \rangle}}
\def\nats{{\mathbb{N}}}
\def\zed{\mathbb{Z}}
\def\card#1{{|\!|{#1}|\!|}}
\def\len#1{{\vert{#1}\vert}}
\def\prod{\Delta}
\def\pat{{\mathbf{b}}}
\def\patt{{\widetilde{\mathbf{b}}}}
\def\patg{{\Gamma_\pat}}
\def\pattg{{\Gamma_\patt}}
\def\df#1{\scriptscriptstyle\mathbf{df}(#1)}

\def\Vars{\ensuremath{\Xi}}
\def\rank#1{\langle\!\langle#1\rangle\!\rangle}

\def\Varsi{\ensuremath{\Xi_{\widecheck{\text{\tiny 1..d}}}}}
\def\Varse{\ensuremath{\Xi_{\widehat{\text{\tiny 1..d}}}}}

\def\Varsilj{\ensuremath{\Xi_{\widecheck{{\tiny j_1{.}{.}j_s}}}}}
\def\Varselj{\ensuremath{\Xi_{\widehat{{\tiny j_1{.}{.}j_s}}}}}

\newcommand{\cycles}[2]{\Omega_{{#1}}({#2})}
\newcommand{\cyclestar}[2]{(\Omega_{{#1}}({#2}))^*}

\def\bdwords{\Upsilon}

\renewcommand{\vec}[1]{{\mathbf {#1}}}

\def\proj{\mathbin{\downarrow}}

\def\foreach{\mathrm{REACH}_{\mathit{fo}}}

\declareauthor{ri}{Radu}{blue}
\declareauthor{pg}{Pierre}{red}
\declareauthor{er}{oldtext}{black!20!white}

\renewcommand{\proj}[2]{{#1}\mathclose{\downarrow}_{{#2}}}
\newcommand{\projpatt}[1]{{#1}\mathclose{\downdownarrows}_{{\patt}}}

\title{Interprocedural Reachability for Flat Integer Programs}
\author{Pierre Ganty\inst{1} \and Radu Iosif\inst{2}}
\institute{\({}^1\)IMDEA Software Institute \quad \({}^2\) CNRS/VERIMAG, Grenoble, France}

\begin{document}

\maketitle

\begin{abstract}
We study programs with integer data, procedure calls and arbitrary
call graphs. We show that, whenever the guards and updates are given
by octagonal relations, the reachability problem along control flow
paths within some language \(w_1^* \ldots w_d^*\) over program
statements is decidable in \textsc{Nexptime}. To achieve this upper
bound, we combine a program transformation into the same class of
programs but without procedures, with an \textsc{Np}-completeness
result for the reachability problem of procedure-less
programs. Besides the program, the expression \(w_1^* \ldots w_d^*\)
is also mapped onto an expression of a similar form but this time over
the transformed program statements. Several arguments involving
context-free grammars and their generative process enable us to give
tight bounds on the size of the resulting expression. The currently
existing gap between \textsc{Np}-hard and \textsc{Nexptime} can be
closed to \textsc{Np}-complete when a certain parameter of the
analysis is assumed to be constant.
\end{abstract}

\section{Introduction}

This paper studies the complexity of the reachability problem for a
class of programs featuring procedures and local/global variables
ranging over integers. In general, the reachability problem for this
class is undecidable \cite{minsky67}. Thus, we focus on a special case
of the reachability problem which restricts both the class of input
programs and the set of executions considered. The class of input
programs is restricted by considering that all updates to the integer
variables $\vec{x}$ are defined by \emph{octagonal constraints}, that
are conjunctions of atoms of the form $\pm x \pm y \leq c$, with
$x,y\in\vec{x}\cup\vec{x}'$, where $\vec{x}'$ denote the future values
of the program variables. The reachability problem is restricted by
limiting the search to program executions conforming to a regular
expression of the form \(w_1^*\ldots w_d^*\) where the \(w_i\)'s are
finite sequences of program statements.

We call this problem \emph{flat-octagonal reachability}
(fo-reachability, for short). Concretely, given:
\begin{inparaenum}[\upshape(\itshape i\upshape)]
\item a program \(\mathcal{P}\) with procedures and local/global variables, 
      whose statements are specified by octagonal constraints, and
\item a bounded expression \(\pat=w_1^* \ldots w_d^*\), where \(w_i\)'s 
are sequences of statements of \(\mathcal{P}\),
\end{inparaenum}
the fo-reachability problem \(\foreach(\mathcal{P},\pat)\) asks:
can \(\mathcal{P}\) run to completion by executing a sequence of program
statements \(w \in \pat\) ?
Studying the complexity of this problem provides the theoretical
foundations for implementing efficient decision procedures, of
practical interest in areas of software verification, such as
bug-finding \cite{EsparzaG11}, or counterexample-guided abstraction
refinement \cite{KroeningLW13,HojjatIKKR12}.

Our starting point is the decidability of the fo-reachability problem
in the absence of procedures. Recently, the precise complexity of this
problem was coined to \textsc{Np}-complete \cite{bik14}. However, this
result leaves open the problem of dealing with procedures and local
variables, let alone when the graph of procedure calls has
cycles, such as in the example of
Fig. \ref{fig:running-example}~(a). Pinning down the complexity of the
fo-reachability problem in presence of (possibly recursive)
procedures, with local variables ranging over integers, is the
challenge we address here.

The decision procedure we propose in this paper reduces \(\foreach(\mathcal{P},\pat)\), from a program \(\mathcal{P}\)
with arbitrary call graphs, to procedure-less programs as follows:
\begin{compactenum}
\item we apply a source-to-source transformation returning a procedure-less
      program \(\mathcal{Q}\), with statements also defined by
      octagonal relations, such that \(\foreach(\mathcal{P},\pat)\) is
      equivalent to the unrestricted reachability problem
      for \(\mathcal{Q}\), when no particular bounded expression is
      supplied.

\item we compute a bounded expression \(\patg\) over the statements
      of \(\mathcal{Q}\), such that \(\foreach(\mathcal{P},\pat)\) is
      equivalent to \(\foreach(\mathcal{Q},\patg)\).
\end{compactenum}

The above reduction allows us to conclude that the fo-reachability
problem for programs with arbitrary call graphs is decidable and
in \textsc{Nexptime}. Naturally, the \textsc{Np}-hard lower
bound \cite{bik14} for the fo-reachability problem of procedure-less
programs holds in our setting as well.
Despite our best efforts, we did not close the complexity gap yet.
However we pinned down a natural parameter, called \emph{index},
related to programs with arbitrary call graphs, such that, when
setting this parameter to a fixed constant (like \(3\) in \(3\)-SAT),
the complexity of the resulting fo-reachability problem for programs
with arbitrary call graphs becomes \textsc{Np}-complete. Indeed, when
the index is fixed, the aforementioned reduction computing
\(\foreach(\mathcal{Q},\patg)\) runs in polynomial time. Then the
\textsc{Np} decision procedure for the fo-reachability of procedure-less
programs \cite{bik14} shows the rest.

The index parameter is better understood in the context of formal
languages. The control flow of procedural programs is captured
precisely by the language of a context-free grammar. A \(k\)-index (\(k>0\))
underapproximation of this language is obtained by filtering out the derivations 
containing a sentential form with \(k+1\) occurrences of nonterminals.
The key to our results is a toolbox of language theoretic
constructions of independent interest that enables to reason about the
structure of context-free derivations generating words into \(\pat = w_1^* \ldots w_d^*\), that
is, words of the form \(w_1^{i_1} \ldots w_d^{i_d}\) for some integers \(i_1,\ldots, i_d \geq 0\).

\begin{figure}[hbt]
\centering
%
\begin{picture}(0,0)%
\includegraphics{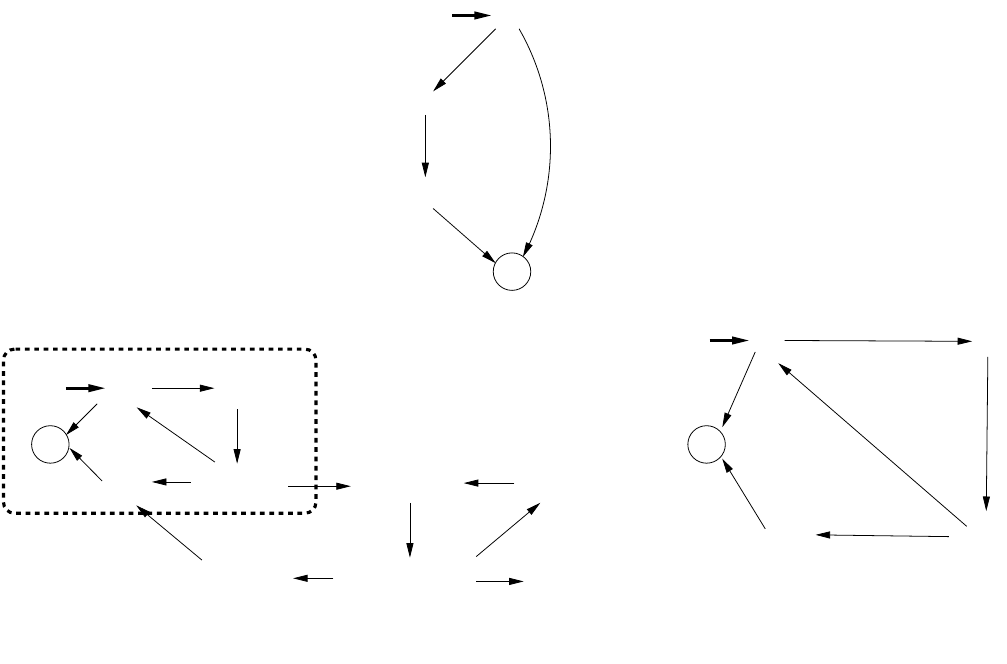}%
\end{picture}%
\setlength{\unitlength}{1973sp}%
\begingroup\makeatletter\ifx\SetFigFont\undefined%
\gdef\SetFigFont#1#2#3#4#5{%
  \reset@font\fontsize{#1}{#2pt}%
  \fontfamily{#3}\fontseries{#4}\fontshape{#5}%
  \selectfont}%
\fi\endgroup%
\begin{picture}(9648,6348)(-32,-5851)
\put(7662,-5144){\makebox(0,0)[lb]{\smash{{\SetFigFont{6}{7.2}{\rmdefault}{\mddefault}{\updefault}{\color[rgb]{0,0,0}$x_J=z_K=0$}%
}}}}
\put(3901,-91){\makebox(0,0)[b]{\smash{{\SetFigFont{6}{7.2}{\rmdefault}{\mddefault}{\updefault}{\color[rgb]{0,0,0}$x'=x$}%
}}}}
\put(9526,-3376){\makebox(0,0)[lb]{\smash{{\SetFigFont{6}{7.2}{\rmdefault}{\mddefault}{\updefault}{\color[rgb]{0,0,0}$x_J=x_I-1$}%
}}}}
\put(9526,-3601){\makebox(0,0)[lb]{\smash{{\SetFigFont{6}{7.2}{\rmdefault}{\mddefault}{\updefault}{\color[rgb]{0,0,0}$z_L=z_K$}%
}}}}
\put(9526,-3826){\makebox(0,0)[lb]{\smash{{\SetFigFont{6}{7.2}{\rmdefault}{\mddefault}{\updefault}{\color[rgb]{0,0,0}$x_L=x_I$}%
}}}}
\put(9526,-4051){\makebox(0,0)[lb]{\smash{{\SetFigFont{6}{7.2}{\rmdefault}{\mddefault}{\updefault}{\color[rgb]{0,0,0}$\mathbf{havoc}()$}%
}}}}
\put(6128,-3392){\makebox(0,0)[lb]{\smash{{\SetFigFont{6}{7.2}{\rmdefault}{\mddefault}{\updefault}{\color[rgb]{0,0,0}$z_O=0$}%
}}}}
\put(6203,-3167){\makebox(0,0)[lb]{\smash{{\SetFigFont{6}{7.2}{\rmdefault}{\mddefault}{\updefault}{\color[rgb]{0,0,0}$x_I=0$}%
}}}}
\put(4126,-886){\makebox(0,0)[lb]{\smash{{\SetFigFont{6}{7.2}{\rmdefault}{\mddefault}{\updefault}{\color[rgb]{0,0,0}$\mathbf{t_2}$}%
}}}}
\put(4351,-1636){\makebox(0,0)[lb]{\smash{{\SetFigFont{6}{7.2}{\rmdefault}{\mddefault}{\updefault}{\color[rgb]{0,0,0}$\mathbf{t_3}$}%
}}}}
\put(4426,-211){\makebox(0,0)[lb]{\smash{{\SetFigFont{6}{7.2}{\rmdefault}{\mddefault}{\updefault}{\color[rgb]{0,0,0}$\mathbf{t_1}$}%
}}}}
\put(2401,-4186){\makebox(0,0)[b]{\smash{{\SetFigFont{5}{6.0}{\rmdefault}{\mddefault}{\updefault}{\color[rgb]{0,0,0}$X_1^{\tuple{0}}X_3^{\tuple{0}}$}%
}}}}
\put(2401,-5086){\makebox(0,0)[b]{\smash{{\SetFigFont{5}{6.0}{\rmdefault}{\mddefault}{\updefault}{\color[rgb]{0,0,0}$X_3^{\tuple{0}}X_3^{\tuple{1}}$}%
}}}}
\put(4726,-3961){\makebox(0,0)[b]{\smash{{\SetFigFont{6}{7.2}{\rmdefault}{\mddefault}{\updefault}{\color[rgb]{0,0,0}$\mathbf{p_1}$}%
}}}}
\put(5101,-4636){\makebox(0,0)[b]{\smash{{\SetFigFont{6}{7.2}{\rmdefault}{\mddefault}{\updefault}{\color[rgb]{0,0,0}$\mathbf{p_3}$}%
}}}}
\put(3901,-5086){\makebox(0,0)[b]{\smash{{\SetFigFont{5}{6.0}{\rmdefault}{\mddefault}{\updefault}{\color[rgb]{0,0,0}$X_3^{\tuple{0}}X_1^{\tuple{1}}X_3^{\tuple{1}}$}%
}}}}
\put(5101,-5161){\makebox(0,0)[lb]{\smash{{\SetFigFont{5}{6.0}{\rmdefault}{\mddefault}{\updefault}{\color[rgb]{0,0,0}$\cdots$}%
}}}}
\put(1801,-3061){\makebox(0,0)[b]{\smash{{\SetFigFont{6}{7.2}{\rmdefault}{\mddefault}{\updefault}{\color[rgb]{0,0,0}$\mathbf{p_1}$}%
}}}}
\put(444,-3811){\makebox(0,0)[b]{\smash{{\SetFigFont{6}{7.2}{\rmdefault}{\mddefault}{\updefault}{\color[rgb]{0,0,0}$\varepsilon$}%
}}}}
\put(1201,-4186){\makebox(0,0)[b]{\smash{{\SetFigFont{5}{6.0}{\rmdefault}{\mddefault}{\updefault}{\color[rgb]{0,0,0}$X_3^{\tuple{0}}$}%
}}}}
\put(699,-4118){\makebox(0,0)[b]{\smash{{\SetFigFont{6}{7.2}{\rmdefault}{\mddefault}{\updefault}{\color[rgb]{0,0,0}$\mathbf{p_3}$}%
}}}}
\put(7876,-1936){\makebox(0,0)[b]{\smash{{\SetFigFont{6}{7.2}{\rmdefault}{\mddefault}{\updefault}{\color[rgb]{0,0,0}(c)}%
}}}}
\put(9601,-2836){\makebox(0,0)[b]{\smash{{\SetFigFont{6}{7.2}{\rmdefault}{\mddefault}{\updefault}{\color[rgb]{0,0,0}$X_2^{\tuple{0}}$}%
}}}}
\put(9601,-4711){\makebox(0,0)[b]{\smash{{\SetFigFont{6}{7.2}{\rmdefault}{\mddefault}{\updefault}{\color[rgb]{0,0,0}$X_1^{\tuple{0}}X_3^{\tuple{0}}$}%
}}}}
\put(7501,-4711){\makebox(0,0)[b]{\smash{{\SetFigFont{6}{7.2}{\rmdefault}{\mddefault}{\updefault}{\color[rgb]{0,0,0}$X_3^{\tuple{0}}$}%
}}}}
\put(6744,-3811){\makebox(0,0)[b]{\smash{{\SetFigFont{6}{7.2}{\rmdefault}{\mddefault}{\updefault}{\color[rgb]{0,0,0}$\epsilon$}%
}}}}
\put(4126,-1411){\makebox(0,0)[b]{\smash{{\SetFigFont{6}{7.2}{\rmdefault}{\mddefault}{\updefault}{\color[rgb]{0,0,0}$X_3$}%
}}}}
\put(4126,-511){\makebox(0,0)[b]{\smash{{\SetFigFont{6}{7.2}{\rmdefault}{\mddefault}{\updefault}{\color[rgb]{0,0,0}$X_2$}%
}}}}
\put(4876,-2161){\makebox(0,0)[b]{\smash{{\SetFigFont{6}{7.2}{\rmdefault}{\mddefault}{\updefault}{\color[rgb]{0,0,0}$\epsilon$}%
}}}}
\put(4876,314){\makebox(0,0)[b]{\smash{{\SetFigFont{6}{7.2}{\rmdefault}{\mddefault}{\updefault}{\color[rgb]{0,0,0}$X_1$}%
}}}}
\put(5626,-661){\makebox(0,0)[b]{\smash{{\SetFigFont{6}{7.2}{\rmdefault}{\mddefault}{\updefault}{\color[rgb]{0,0,0}$x=0$}%
}}}}
\put(5626,-916){\makebox(0,0)[b]{\smash{{\SetFigFont{6}{7.2}{\rmdefault}{\mddefault}{\updefault}{\color[rgb]{0,0,0}$z'=0$}%
}}}}
\put(3826,-1786){\makebox(0,0)[b]{\smash{{\SetFigFont{6}{7.2}{\rmdefault}{\mddefault}{\updefault}{\color[rgb]{0,0,0}$x'=x$}%
}}}}
\put(3976,-2011){\makebox(0,0)[b]{\smash{{\SetFigFont{6}{7.2}{\rmdefault}{\mddefault}{\updefault}{\color[rgb]{0,0,0}$z'=z+2$}%
}}}}
\put(3301,-886){\makebox(0,0)[b]{\smash{{\SetFigFont{6}{7.2}{\rmdefault}{\mddefault}{\updefault}{\color[rgb]{0,0,0}$z'=\mathtt{P}(x-1)$}%
}}}}
\put(3001,-1111){\makebox(0,0)[b]{\smash{{\SetFigFont{6}{7.2}{\rmdefault}{\mddefault}{\updefault}{\color[rgb]{0,0,0}$x'=x$}%
}}}}
\put(3901,-4186){\makebox(0,0)[b]{\smash{{\SetFigFont{5}{6.0}{\rmdefault}{\mddefault}{\updefault}{\color[rgb]{0,0,0}$X_3^{\tuple{0}}X_2^{\tuple{1}}$}%
}}}}
\put(5401,-4186){\makebox(0,0)[b]{\smash{{\SetFigFont{5}{6.0}{\rmdefault}{\mddefault}{\updefault}{\color[rgb]{0,0,0}$X_3^{\tuple{0}}X_1^{\tuple{1}}$}%
}}}}
\put(8026,-3436){\makebox(0,0)[lb]{\smash{{\SetFigFont{6}{7.2}{\rmdefault}{\mddefault}{\updefault}{\color[rgb]{0,0,0}$z_O=z_L+2$}%
}}}}
\put(8026,-3211){\makebox(0,0)[lb]{\smash{{\SetFigFont{6}{7.2}{\rmdefault}{\mddefault}{\updefault}{\color[rgb]{0,0,0}$x_O=x_L$}%
}}}}
\put(7201,-586){\makebox(0,0)[lb]{\smash{{\SetFigFont{6}{7.2}{\rmdefault}{\mddefault}{\updefault}{\color[rgb]{0,0,0}$\mathbf{p_1} ~:~ X_1 ~\rightarrow~ \mathbf{t_1} X_2$ }%
}}}}
\put(7201,-886){\makebox(0,0)[lb]{\smash{{\SetFigFont{6}{7.2}{\rmdefault}{\mddefault}{\updefault}{\color[rgb]{0,0,0}$\mathbf{p_2} ~:~ X_2 ~\rightarrow~ \mathbf{\langle\! t_2} ~X_1~ \mathbf{t_2 \!\rangle} ~X_3$}%
}}}}
\put(7201,-1186){\makebox(0,0)[lb]{\smash{{\SetFigFont{6}{7.2}{\rmdefault}{\mddefault}{\updefault}{\color[rgb]{0,0,0}$\mathbf{p_3} ~:~ X_3 ~\rightarrow~ \mathbf{t_3}$}%
}}}}
\put(7201,-1486){\makebox(0,0)[lb]{\smash{{\SetFigFont{6}{7.2}{\rmdefault}{\mddefault}{\updefault}{\color[rgb]{0,0,0}$\mathbf{p_4} ~:~ X_1 ~\rightarrow~ \mathbf{t_4}$}%
}}}}
\put(7403,-2829){\makebox(0,0)[b]{\smash{{\SetFigFont{6}{7.2}{\rmdefault}{\mddefault}{\updefault}{\color[rgb]{0,0,0}$X_1^{\tuple{0}}$}%
}}}}
\put(4501,-2536){\makebox(0,0)[b]{\smash{{\SetFigFont{6}{7.2}{\rmdefault}{\mddefault}{\updefault}{\color[rgb]{0,0,0}(b)}%
}}}}
\put(1051,-2536){\makebox(0,0)[lb]{\smash{{\SetFigFont{8}{9.6}{\rmdefault}{\mddefault}{\updefault}{\color[rgb]{0,0,0}The program $\mathcal{P}$}%
}}}}
\put(901,-2506){\makebox(0,0)[b]{\smash{{\SetFigFont{6}{7.2}{\rmdefault}{\mddefault}{\updefault}{\color[rgb]{0,0,0}(a)}%
}}}}
\put(601,239){\makebox(0,0)[lb]{\smash{{\SetFigFont{6}{7.2}{\rmdefault}{\mddefault}{\updefault}{\color[rgb]{0,0,0}int {\tt P}(int x) \{}%
}}}}
\put(601,-16){\makebox(0,0)[lb]{\smash{{\SetFigFont{6}{7.2}{\rmdefault}{\mddefault}{\updefault}{\color[rgb]{0,0,0}int z;}%
}}}}
\put(601,-361){\makebox(0,0)[lb]{\smash{{\SetFigFont{6}{7.2}{\rmdefault}{\mddefault}{\updefault}{\color[rgb]{0,0,0}1: assume($x \geq 0$);}%
}}}}
\put(601,-661){\makebox(0,0)[lb]{\smash{{\SetFigFont{6}{7.2}{\rmdefault}{\mddefault}{\updefault}{\color[rgb]{0,0,0}2: if ($x > 0$)}%
}}}}
\put(601,-961){\makebox(0,0)[lb]{\smash{{\SetFigFont{6}{7.2}{\rmdefault}{\mddefault}{\updefault}{\color[rgb]{0,0,0}3:\ \ \ \ z := P(x-1);}%
}}}}
\put(601,-1261){\makebox(0,0)[lb]{\smash{{\SetFigFont{6}{7.2}{\rmdefault}{\mddefault}{\updefault}{\color[rgb]{0,0,0}4:\ \ \ \ z := z+2;}%
}}}}
\put(601,-1861){\makebox(0,0)[lb]{\smash{{\SetFigFont{6}{7.2}{\rmdefault}{\mddefault}{\updefault}{\color[rgb]{0,0,0}6:\ \ \ \ z := 0}%
}}}}
\put(601,-2161){\makebox(0,0)[lb]{\smash{{\SetFigFont{6}{7.2}{\rmdefault}{\mddefault}{\updefault}{\color[rgb]{0,0,0}7: return z; \}}%
}}}}
\put(601,-1561){\makebox(0,0)[lb]{\smash{{\SetFigFont{6}{7.2}{\rmdefault}{\mddefault}{\updefault}{\color[rgb]{0,0,0}5: else}%
}}}}
\put(4951,-1186){\makebox(0,0)[lb]{\smash{{\SetFigFont{6}{7.2}{\rmdefault}{\mddefault}{\updefault}{\color[rgb]{0,0,0}$\mathbf{t_4}$}%
}}}}
\put(1726,-4261){\makebox(0,0)[b]{\smash{{\SetFigFont{6}{7.2}{\rmdefault}{\mddefault}{\updefault}{\color[rgb]{0,0,0}$\mathbf{p_4}$}%
}}}}
\put(3226,-3961){\makebox(0,0)[b]{\smash{{\SetFigFont{6}{7.2}{\rmdefault}{\mddefault}{\updefault}{\color[rgb]{0,0,0}$\mathbf{p_1}$}%
}}}}
\put(1201,-3286){\makebox(0,0)[b]{\smash{{\SetFigFont{5}{6.0}{\rmdefault}{\mddefault}{\updefault}{\color[rgb]{0,0,0}$X_1^{\tuple{0}}$}%
}}}}
\put(2401,-3286){\makebox(0,0)[b]{\smash{{\SetFigFont{5}{6.0}{\rmdefault}{\mddefault}{\updefault}{\color[rgb]{0,0,0}$X_2^{\tuple{0}}$}%
}}}}
\put(966,-3526){\makebox(0,0)[b]{\smash{{\SetFigFont{6}{7.2}{\rmdefault}{\mddefault}{\updefault}{\color[rgb]{0,0,0}$\mathbf{p_4}$}%
}}}}
\put(1785,-3586){\makebox(0,0)[b]{\smash{{\SetFigFont{6}{7.2}{\rmdefault}{\mddefault}{\updefault}{\color[rgb]{0,0,0}$\mathbf{p_3}$}%
}}}}
\put(2451,-3575){\makebox(0,0)[b]{\smash{{\SetFigFont{6}{7.2}{\rmdefault}{\mddefault}{\updefault}{\color[rgb]{0,0,0}$\mathbf{p_2}$}%
}}}}
\put(3066,-4896){\makebox(0,0)[b]{\smash{{\SetFigFont{6}{7.2}{\rmdefault}{\mddefault}{\updefault}{\color[rgb]{0,0,0}$\mathbf{p_4}$}%
}}}}
\put(4954,-4917){\makebox(0,0)[b]{\smash{{\SetFigFont{6}{7.2}{\rmdefault}{\mddefault}{\updefault}{\color[rgb]{0,0,0}$\mathbf{p_1}$}%
}}}}
\put(4100,-4546){\makebox(0,0)[b]{\smash{{\SetFigFont{6}{7.2}{\rmdefault}{\mddefault}{\updefault}{\color[rgb]{0,0,0}$\mathbf{p_2}$}%
}}}}
\put(1371,-4726){\makebox(0,0)[b]{\smash{{\SetFigFont{6}{7.2}{\rmdefault}{\mddefault}{\updefault}{\color[rgb]{0,0,0}$\mathbf{p_3}$}%
}}}}
\put(6477,-4486){\makebox(0,0)[b]{\smash{{\SetFigFont{6}{7.2}{\rmdefault}{\mddefault}{\updefault}{\color[rgb]{0,0,0}$z_0=z_I+2$}%
}}}}
\put(6422,-4261){\makebox(0,0)[b]{\smash{{\SetFigFont{6}{7.2}{\rmdefault}{\mddefault}{\updefault}{\color[rgb]{0,0,0}$x_I=x_O$}%
}}}}
\put(7314,-3821){\makebox(0,0)[lb]{\smash{{\SetFigFont{6}{7.2}{\rmdefault}{\mddefault}{\updefault}{\color[rgb]{0,0,0}$\vec{x}'_I=\vec{x}_J$}%
}}}}
\put(7283,-4036){\makebox(0,0)[lb]{\smash{{\SetFigFont{6}{7.2}{\rmdefault}{\mddefault}{\updefault}{\color[rgb]{0,0,0}$\vec{x}'_O=\vec{x}_K$}%
}}}}
\put(7454,-2487){\makebox(0,0)[lb]{\smash{{\SetFigFont{6}{7.2}{\rmdefault}{\mddefault}{\updefault}{\color[rgb]{0,0,0}$\mathbf{havoc}(\vec{x}_L, \vec{x}_J, \vec{x}_K)$}%
}}}}
\put(1830,-5720){\makebox(0,0)[b]{\smash{{\SetFigFont{6}{7.2}{\rmdefault}{\mddefault}{\updefault}{\color[rgb]{0,0,0}(d)}%
}}}}
\put(7265,-5734){\makebox(0,0)[b]{\smash{{\SetFigFont{6}{7.2}{\rmdefault}{\mddefault}{\updefault}{\color[rgb]{0,0,0}(e)}%
}}}}
\put(7460,-5755){\makebox(0,0)[lb]{\smash{{\SetFigFont{8}{9.6}{\rmdefault}{\mddefault}{\updefault}{\color[rgb]{0,0,0}The program $\mathcal{Q}$}%
}}}}
\put(8010,-2660){\makebox(0,0)[lb]{\smash{{\SetFigFont{6}{7.2}{\rmdefault}{\mddefault}{\updefault}{\color[rgb]{0,0,0}$x_I>0$}%
}}}}
\put(7573,-5374){\makebox(0,0)[lb]{\smash{{\SetFigFont{6}{7.2}{\rmdefault}{\mddefault}{\updefault}{\color[rgb]{0,0,0}$\mathbf{havoc}(\vec{x}_I)$}%
}}}}
\put(7681,-4941){\makebox(0,0)[lb]{\smash{{\SetFigFont{6}{7.2}{\rmdefault}{\mddefault}{\updefault}{\color[rgb]{0,0,0}$\vec{x}'_I=\vec{x}_L$}%
}}}}
\put(3901,164){\makebox(0,0)[b]{\smash{{\SetFigFont{6}{7.2}{\rmdefault}{\mddefault}{\updefault}{\color[rgb]{0,0,0}$x>0$}%
}}}}
\end{picture}%
 \caption{$\vec{x}_I = \set{x_I,z_I}$ ($\vec{x}_O = \set{x_O,z_O}$) are for
the input (output) values of $x$ and $z$, respectively. 
$\vec{x}_{J,K,L}$ provide extra copies.
$\mathbf{havoc}(\vec{y})$ stands for
$\bigwedge_{x \in \vec{x}_{I,O,J,K,L}
\setminus \vec{y}} x'=x$, and $\vec{x}'_\alpha = \vec{x}_\beta$ for
$\bigwedge_{x \in \vec{x}} x'_\alpha=x_\beta$.  }
\label{fig:running-example}
\end{figure}
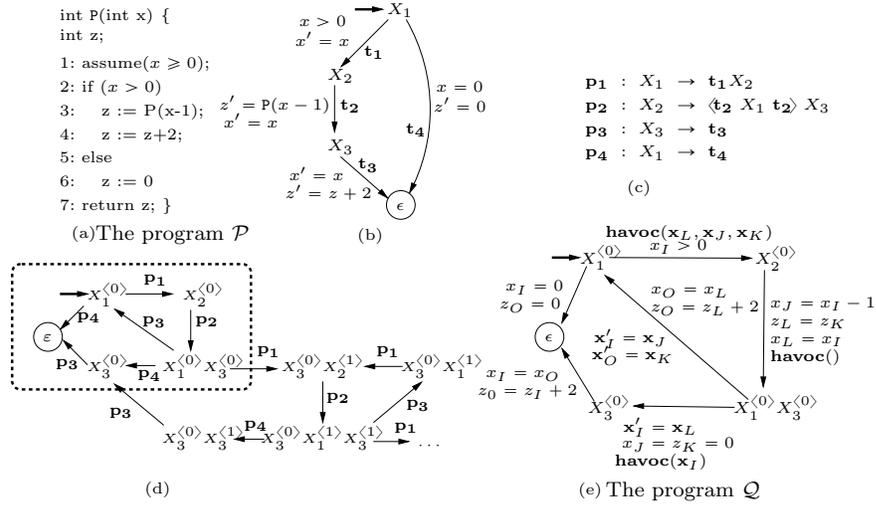

To properly introduce the reader to our result, we briefly recall the
important features of our source-to-source transformation through an
illustrative example. We apply first our program
transformation \cite{gik13} to the program \(\mathcal{P}\) shown in
Fig. \ref{fig:running-example} (a). The call graph of this program
consists of a single state \(\mathtt{P}\) with a self-loop. The output
program \(\mathcal{Q}\) given Fig. \ref{fig:running-example} (e), has
no procedures and it can thus be analyzed using any existing
intra-procedural tool \cite{bik10,BFLP03}. The relation between the
variables \(x\) and \(z\) of the input program can be inferred from
the analysis of the output program. For instance, the input-output
relation of the program \(\mathcal{P}\) is defined by $z'=2x$, which
matches the precondition $z_O = 2x_I$ of the
program \(\mathcal{Q}\). Consequently, any assertion such as
``\emph{there exists a value \(n>0\) such that \(\mathtt{P}(n)<n\)}''
can be phrased as: ``\emph{there exist values \(n < m\) such
that \(\mathcal{Q}(n,m)\) reaches its final state}''. While the former
can be encoded by a reachability problem on \(\mathcal{P}\), by adding
an extra conditional statement, the latter is an equivalent
reachability problem for \(\mathcal{Q}\).

For the sake of clarity, we give several representations of the input program
\(\mathcal{P}\) that we assume the reader is familiar with including the text of the
program in Fig.  \ref{fig:running-example} (a) and the corresponding control flow
graph in Fig.  \ref{fig:running-example} (b).

In this paper, the formal model we use for programs is based on context-free grammars.
The grammar for \(\mathcal{P}\) is given at  Fig. \ref{fig:running-example} (c).
The r\^ole of the grammar is to define the set
of \emph{interprocedurally valid} paths in the control-flow graph of
the program \(\mathcal{P}\). Every edge in the control-flow graph matches one or two
symbols from the finite alphabet
$\set{\mathbf{t_1}, \langle\!\mathbf{t_2},
\mathbf{t_2}\!\rangle, \mathbf{t_3}, \mathbf{t_4}}$,
where $\langle\!\mathbf{t_2}$ and $\mathbf{t_2}\!\rangle$ denote the
call and return, respectively. The set of nonterminals is
$\set{X_1,X_2,X_3,X_4}$. Each edge
in the graph translates to a production rule in the grammar, labeled
$\mathbf{p_1}$ to $\mathbf{p_4}$.  For instance, the call edge
$X_2 \arrow{\mathbf{t_2}}{} X_3$ becomes
$X_2 \rightarrow \langle\!\mathbf{t_2} X_1 \mathbf{t_2}\!\rangle
X_3$. The language of the grammar of Fig.~\ref{fig:running-example} (c) (with axiom $X_1$) is
the set $L=\set{\left(\mathbf{t_1}\langle\!\mathbf{t_2}\right)^n \mathbf{t_4}
\left(\mathbf{t_2}\!\rangle\mathbf{t_3}\right)^n \mid n \in \nat}$ of 
interprocedurally valid paths in the control-flow graph. Observe that
\(L\) is included in the language of the regular expression \(\pat
= \left(\mathbf{t_1}\langle\!\mathbf{t_2}\right)^* \mathbf{t_4}^*
\left(\mathbf{t_2}\!\rangle\mathbf{t_3}\right)^*\).

Our program transformation is based on the observation that the semantics of
\(\mathcal{P}\) can be precisely defined on the set of \emph{derivations} of
the associated grammar. In principle, one can always represent this set of
derivations as a possibly infinite automaton (Fig. \ref{fig:running-example}
(d)), whose states are sequences of nonterminals annotated with priorities
(called ranks)\footnote{The precise definition and use of ranks will be
explained in Section~\ref{sec:dfkderiv}.}, and whose transitions are labeled
with production rules. Each finite path in this automaton, starting from
$X_1^\tuple{0}$, defines a valid prefix of a derivation. Since \(L\subseteq
\pat\), Luker \cite{Luker78} shows that it is sufficient to keep a finite
sub-automaton, enclosed with a dashed box in Fig.~\ref{fig:running-example} (d),
in which each state consists of a finite number of ranked nonterminals (in our
case at most \(2\)).

Finally, we label the edges of this finite automaton with octagonal
constraints that capture the semantics of the relations labeling the
control-flow graph from Fig. \ref{fig:running-example} (b). We give
here a brief explanation for the labeling of the finite automaton in
Fig. \ref{fig:running-example} (e), in other words, the output
program \(\mathcal{Q}\) (see \cite{gik13} for more details). The idea
is to compute, for each production rule $\mathbf{p_i}$, a relation
$\rho_i(\vec{x}_I, \vec{x}_O)$, based on the constraints associated
with the symbols occurring in \(\mathbf{p_i}\) (labels from
Fig. \ref{fig:running-example} (b)). For instance, in the transition
$X_2^{\tuple{0}} \arrow{\mathbf{p_2}}{}
X_1^{\tuple{0}}X_3^{\tuple{0}}$, the auxiliary variables store
intermediate results of the computation of \(\mathbf{p}_2\) as
follows: \( [\vec{x}_I]\; \langle\!\mathbf{t_2} \;[\vec{x}_J]\;
X_1 \;[\vec{x}_K]\; \mathbf{t_2}\!\rangle \;[\vec{x}_L]\;
X_3 \;[\vec{x}_O]\). The guard of the transition can be understood by
noticing that $\langle\!\mathbf{t_2}$ gives rise to the constraint
$x_J=x_I-1$, $\mathbf{t_2}\!\rangle$ to $z_L=z_K$, $x_I=x_L$
corresponds to the frame condition of the call, and $\mathbf{havoc}()$
copies all current values of $\vec{x}_{I,J,K,L,O}$ to the future ones.
It is worth pointing out that the constraints labeling the transitions
of the program \(\mathcal{Q}\) are necessarily octagonal if the
statements of \(\mathcal{P}\) are defined by octagonal constraints.

An intra-procedural analysis of the program \(\mathcal{Q}\) in
Fig. \ref{fig:running-example} (e) infers the precondition
$x_I\geq0 \wedge z_O=2x_I$ which coincides with the input/output
relation of the recursive program \(\mathcal{P}\) in
Fig. \ref{fig:running-example} (a), i.e.\ $x \geq 0 \wedge z'=2x$. The
original query \(\exists n > 0 \colon \mathtt{P}(n) < n\) translates
thus into the satisfiability of the formula \(x_I > 0 \wedge
z_O=2x_I \wedge x_I<z_O\), which is clearly false.

The paper is organised as follows: basic definitions are given Section \ref{sec:prelim},
Section~\ref{sec:foreachpb} defines the fo-reachability problem, 
Section~\ref{sec:dfkderiv} presents an alternative program semantics based on derivations
and introduces subsets of derivations which are sufficient to decide reachability,
Section \ref{sec:bounded-control-sets} starts with on overview of our decision
procedure and our main complexity results and continues with the key steps of our algorithms. 
The appendix contains all the missing details.
\comment[pg]{Turn into ``A companion technical report \cite{} contains the missing details.''}

\section{Preliminaries}\label{sec:prelim}

Let $\Sigma$ be a finite nonempty set of symbols, called
an \emph{alphabet}. We denote by $\Sigma^*$ the set of finite words
over $\Sigma$ which includes $\varepsilon$, the empty word. The concatenation
of two words $u,v\in\Sigma^*$ is denoted by $u\cdot v$ or $u\, v$.
Given a word \(w\in\Sigma^*\), let \(\len{w}\) denote its length and
let \((w)_i\) with \(1\leq i\leq \len{w}\) be the \(i\)th symbol
of \(w\). Given \(w\in\Sigma^*\) and \(\Theta\subseteq\Sigma\), we
write \(\proj{w}{\Theta}\) for the word obtained by deleting
from \(w\) all symbols not in \(\Theta\), and sometimes we
write \(\proj{w}{a}\) for \(\proj{w}{\set{a}}\).
A {\em bounded expression} \(\pat\) over alphabet \(\Sigma\) is a regular
expression of the form $w_1^* \ldots w_d^*$, where $w_1,\ldots,w_d \in
\Sigma^*$ are nonempty words and its size is given by \(\len{\pat}=\sum_{i=1}^d
\len{w_i}\).  We use $\pat$ to denote both the bounded expression and its
language. We call a language \(L\) \emph{bounded} when \(L\subseteq\pat\) for some bounded expression \(\pat\).

A \emph{grammar} is a tuple \(G=\tuple{\Vars,\Sigma,\prod}\)
where \(\Vars\) is a finite nonempty set
of \emph{nonterminals}, \(\Sigma\) is an alphabet of \emph{terminals},
such that \(\Vars \cap \Sigma = \emptyset\), and
\(\prod \subseteq \Vars\times (\Sigma\cup \Vars)^*\) is a finite set
of \emph{productions}. For a production $(X,w)
\in \prod$, often conveniently noted \(X \rightarrow w\), we define its {\em
size} as $\len{(X,w)} = \len{w}+1$, and $\len{G}
= \sum_{p\in\prod} \len{p}$ defines the size of $G$.

Given two words $u,v \in (\Sigma \cup \Vars)^*$, a production \(
(X,w)\in\prod\) and a position \(1\leq j\leq \len{u}\), we define
a \emph{step} $u \Arrow{(X,w)/j}_G v$ if and only if \( (u)_j = X\)
and \( v = (u)_1 \cdots (u)_{j-1} \, w\, (u)_{j+1}\cdots
(u)_{\len{u}}\). We omit \( (X,w)\) or \(j\) above the arrow when
clear from the context.
A \emph{control word} is a finite word \(\gamma\in\prod^*\) over the
alphabet of productions. A \emph{step sequence} \( u \Arrow{\gamma}_G
v\) is a sequence \(u = w_{0} \Arrow{(\gamma)_1}_G
w_1 \ldots w_{n-1}\Arrow{(\gamma)_n}_G w_n = v\) where \(n = \len{\gamma}\).
If \(u\in\Vars\) is a nonterminal and \(v\in\Sigma^*\) is a word
without nonterminals, we call the step sequence \(u \Arrow{\gamma}_G
v\) a \emph{derivation}. When the control word \(\gamma\) is not
important, we write \(u \Rightarrow^*_G v\) instead
of \(u \Arrow{\gamma}_G v\), and we chose to omit the grammar \(G\)
when clear from the context.

Given a nonterminal $X \in \Vars$ and $Y \in \Vars \cup
\{\varepsilon\}$, i.e.\ $Y$ is either a nonterminal or the empty word,
we define the set \(L_{X,Y}(G)=\set{u \, v \in \Sigma^* \mid
  X \Rightarrow^* u \, Y \, v}\).  The set $L_{X,\varepsilon}(G)$ is
  called the {\em language} of $G$ produced by $X$, and is denoted
  $L_X(G)$ in the following. For a set $\Gamma \subseteq \prod^*$ of
  control words (also called a \emph{control set}), we denote by
  $\hat{L}_{X,Y}(\Gamma, G) = \{ u\, v\in \Sigma^* \mid \exists
\gamma \in \Gamma \colon X \xRightarrow{\gamma} u\, Y\,
v\}$ the language generated by $G$ using only control words from
$\Gamma$. We also write \(\hat{L}_{X}(\Gamma, G)\) for 
\(\hat{L}_{X,\varepsilon}(\Gamma, G)\).

Let $\vec{x}$ denote a nonempty finite set of integer variables,
and $\vec{x}' = \{x' \mid x \in \vec{x}\}$. A \emph{valuation} of
$\vec{x}$ is a function $\smash{\nu : \vec{x} \arrow{}{} \zed}$. The
set of all such valuations is denoted by $\zed^{\vec{x}}$. A formula
$\phi(\vec{x},\vec{x}')$ is evaluated with respect to two valuations
$\nu,\nu'\in\zed^\vec{x}$, by replacing each occurrence of
$x\in\vec{x}$ with $\nu(x)$ and each occurrence of $x'\in\vec{x'}$
with $\nu'(x)$. We write $(\nu,\nu') \models \phi$ when the formula
obtained from these replacements is valid. A formula
$\phi_R(\vec{x},\vec{x}')$ \emph{defines} a relation
$R \subseteq \zed^{\vec{x}}
\times \zed^{\vec{x}}$ whenever for all $\nu,\nu' \in \zed^{\vec{x}}$, we have
$(\nu,\nu')\in R$ if{}f $(\nu,\nu') \models \phi_R$. The composition of
two relations $R_1, R_2 \subseteq \zed^\vec{x} \times
\zed^\vec{x}$ defined by formulae $\varphi_1(\vec{x},\vec{x'})$ and
$\varphi_2(\vec{x},\vec{x'})$, respectively, is the relation $R_1
\circ R_2 \subseteq \zed^{\vec{x}} \times \zed^{\vec{x}}$, defined by 
$\exists \vec{y} ~.~ \varphi_1(\vec{x},\vec{y}) \wedge
\varphi_2(\vec{y}, \vec{x'})$.
For a finite set \(S\), we denote its cardinality by \(\card{S}\).

\section{Interprocedural Flat Octogonal Reachability}\label{sec:foreachpb}

In this section we define formally the class of programs and
reachability problems considered. An \emph{octagonal relation}
\(R\subseteq\zed^{\vec{x}}\times\zed^{\vec{x}}\) is a relation defined 
by a finite conjunction of constraints of the form
\(\pm x \pm y \leq c\), where \(x,y \in \vec{x} \cup \vec{x}'\) and 
\(c\in\zed\). The set of octagonal relations over the variables in 
\(\vec{x}\) and \(\vec{x}'\) is denoted as \(\mathrm{Oct}(\vec{x}, \vec{x}')\). 
The \emph{size} of an octagonal relation \(R\), denoted \(\len{R}\) is
the size of the binary encoding of the smallest octagonal constraint
defining \(R\).

An \emph{octagonal program} is a tuple
\(\mathcal{P}=\tuple{G,I,\sem{.}}\), where \(G\) is a grammar
\(G=\tuple{\Vars,\Sigma,\prod}\), \(I\in\Vars\) is an \emph{initial} 
location, and \(\sem{.} :
L_I(G) \rightarrow \mathrm{Oct}(\vec{x},\vec{x}')\) is a mapping of
the words produced by the grammar \(G\), starting with the initial
location \(I\), to octagonal relations.  The alphabet \(\Sigma\)
contains a symbol \(t\) for each \emph{internal} program statement
(that is not a call to a procedure) and two symbols \(\langle\!t,
t\!\rangle\) for each \emph{call} statement \(t\). The grammar \(G\)
has three kinds of productions: \begin{inparaenum}[\upshape(\itshape
i\upshape)]
\item \((X,t)\) if \(t\) is a statement leading from \(X\) to a return location, 
\item \((X,t\, Y)\) if \(t\) leads from \(X\) to \(Y\), and 
\item \((X,\langle\!t \, Y \, t\!\rangle \, Z)\) if \(t\) is a call statement, 
\(Y\) is the initial location of the callee, and \(Z\) is the continuation of the call.
\end{inparaenum}
Through several program transformations, we may generate another
grammar with other kinds of productions. \comment[ri]{this sounds strange; I would remove it entirely}
The only property we need for our results is that every grammar \(G\)
with we deal with has each of its productions \( (X,w) \) satisfying:
\(\len{\proj{w}{\Sigma}}\leq 2\) and \(\len{\proj{w}{\Vars}}\leq 2\)
where \(\Sigma\) and \(\Vars\) are the terminals and nonterminals of \(G\), respectively.
Each edge \(t\) that is not a call has an associated octagonal
relation \(\rho_t\in\mathrm{Oct}(\vec{x},\vec{x}')\) and each matching
pair \(\langle\!t, t\!\rangle\) has an associated \emph{frame
condition} \(\phi_t\in\mathrm{Oct}(\vec{x},\vec{x}')\), which equates
the values of the local variables, that are not updated by the call,
to their future values. The size of an octagonal program \(\mathcal{P}
= \tuple{G,I,\sem{.}}\), with \(G=\tuple{\Vars,\Sigma,\prod}\), is the
sum of the sizes of all octagonal relations labeling the productions
of \(G\), formally \(\len{\mathcal{P}}
= \sum_{(X,\, t) \in \prod} \len{\rho_t}
+ \sum_{(X,\, tY) \in \prod} \len{\rho_t} 
+ \sum_{(X,\, \, \langle\!t\, Y\, t\!\rangle\, Z) \in \prod} 
(\len{\rho_{\langle\!t}} + \len{\rho_{t\!\rangle}} + \len{\phi_t})\).

For example, the program in Fig. \ref{fig:running-example} (a,b) is
represented by the grammar in Fig. \ref{fig:running-example} (c).  The
terminals are mapped to octagonal relations
as: \(\rho_{\mathbf{t_1}} \equiv x>0 \wedge
x'=x\), \(\rho_{\langle\!\mathbf{t_2}} \equiv
x'=x-1\), \(\rho_{\mathbf{t_2}\!\rangle} \equiv
z'=z\), \(\rho_{\mathbf{t_3}} \equiv x'=x \wedge z'=z+2\)
and \(\rho_{\mathbf{t_4}} \equiv x=0 \wedge z'=0\). The frame
condition is \(\phi_{\mathbf{t_2}} \equiv x'=x\), as only \(z\) is
updated by the call \(z'=\mathtt{P}(x-1)\).

\smallskip
\noindent {\bf Word-based semantics.} %
For each word \(w\in L_I(G)\), each occurrence of a
terminal \(\langle\!t\) in \(w\) is matched by an occurrence
of \(t\!\rangle\), and the matching positions are nested\footnote{A
relation \(\mathord{\leadsto} \subseteq \set{1,\ldots,\len{w}} \times \set{1,\ldots,\len{w}}\)
is said to be nested \cite{alur-jacm09} when no two pairs \(i \leadsto
j\) and \(i' \leadsto j'\) cross each other, as in \(i < i' \leq j <
j'\). }. The semantics of the word \(\sem{w}\) is an octagonal
relation defined inductively\footnote{Octagonal relations are closed under intersections and compositions \cite{mine}.} on the
structure of \(w\):
\begin{inparaenum}[\upshape(\itshape i\upshape)]
\item \(\sem{t} = \rho_t\), 
\item \(\sem{t\cdot v} = \rho_t \circ \sem{v}\), and
\item \(\sem{\langle\!t \cdot u \cdot t\!\rangle \cdot v} = 
\left(\left(\rho_{\langle\!t} \circ \sem{u} \circ \rho_{t\!\rangle}\right) \cap \phi_t\right) \circ \sem{v}\),
\end{inparaenum}
for all \(t,\langle\!t,t\!\rangle \in \Sigma\) such that \(\langle\!t\) and \(t\!\rangle\) match. For instance, the
semantics of the word
\(w=\mathbf{t_1} \langle\!\mathbf{t_2} \mathbf{t_4} \mathbf{t_2}\!\rangle \mathbf{t_3} \in
L_{X_1}(G)\), for the grammar \(G\) given in Fig. \ref{fig:running-example} (c), is 
\(\sem{w} \equiv x=1 \wedge z'=2\). Observe that this word defines the effect of 
an execution of the program in Fig. \ref{fig:running-example} (a)
where the function \(\mathtt{P}\) is called twice---the first call is
a top-level call, and the second is a recursive call (line 3).

\smallskip
\noindent {\bf Reachability problem.}
The semantics of a program \(\mathcal{P} = \tuple{G,I,\sem{.}}\) is
defined as \(\sem{\mathcal{P}} = \bigcup_{w \in
L_I(G)} \sem{w}\). Consider, in addition, a bounded
expression \(\pat\), we
define \(\sem{\mathcal{P}}_\pat = \bigcup_{w \in
L_I(G) \cap \pat} \sem{w}\).  The problem asking
whether \(\sem{\mathcal{P}}_\pat \neq \emptyset\)
for a pair \(\mathcal{P},\pat\) is called
the \emph{flat-octagonal reachability problem}.
We use  \(\foreach(\mathcal{P}, \pat)\) to denote
a particular instance.

\section{Index-bounded depth-first derivations}\label{sec:dfkderiv}

In this section, we give an alternate but equivalent program semantics
based on derivations.  Although simple, the word semantics is defined
using a nesting relation that pairs the positions of a word labeled
with matching symbols \(\langle\!t\) and \(t\!\rangle\). In contrast,
the derivation-based semantics just needs the control word.

To define our derivation based semantics, we first define structured
subsets of derivations namely the depth-first and bounded-index
derivations. The reason is two-fold:
\begin{inparaenum}[\upshape(\itshape a\upshape)]
\item the correctness proof of our program transformation \cite{gik13} 
      returning the procedure-less program \(\mathcal{Q}\) depends on
	bounded-index depth-first derivations, and
\item in the reduction of the \(\foreach(\mathcal{P},\pat)\) problem to
	that of \(\foreach(\mathcal{Q},\patg)\), the computation
	of \(\patg\) depends on the fact that the control structure
	of \(\mathcal{Q}\) stems from a finite automaton recognizing
	bounded-index depth-first derivations.
\end{inparaenum}
Key results for our decision procedure are those of
Luker \cite{Luker78,Luker80} who, intuitively, shows that
if \(L_X(G) \subseteq \pat\) then it is sufficient to consider
depth-first derivations in which no step contains more than \(k\)
simultaneous occurrences of nonterminals, for some \(k>0\)
(Theorem \ref{thm:luker}). 

\noindent {\bf Depth-first derivations.}
It is well-known that a derivation can be associated a unique parse
tree.  A derivation is said to be \emph{depth-first} if it corresponds
to a depth-first traversal of the corresponding parse tree.  More
precisely, given a step sequence \(w_0 \Arrow{(X_0,v_0)/j_0}
w_1 \ldots w_{n{-}1} \Arrow{(X_{n{-}1},v_{n{-}1})/j_{n{-}1}} w_{n}\), and two
integers \(m\) and
\(i\) such that \(0 \leq m < n\) and \(1\leq i\leq \len{w_m}\) define
\(f_m(i)\) to be the index \(\ell\) of the first word \(w_{\ell}\) of the step
sequence in which the particular occurrence of \( (w_m)_i \) appears.
A step sequence is \emph{depth-first} \cite{Luker80} if{}f for all \(m\), \(0\leq m < n\):
{\setlength\abovedisplayskip{4pt}
\setlength\belowdisplayskip{4pt}
\[ f_m( j_m ) = \max \{ f_m( i )  \mid  1\leq i \leq \len{w_m}\text{ and } (w_m)_i \in \Vars\}\enspace .\]}
For example, \(X \Arrow{(X,YY)/1} YY \Arrow{(Y,Z)/2}
YZ \Arrow{(Z,a)/2} Ya\) is depth-first, whereas \(X \Arrow{(X,YY)/1}
YY \Arrow{(Y,Z)/2} YZ \Arrow{(Y,Z)/1} ZZ\) is not. We
have \(f_2(1)=1\) because \( (w_2)_1 = Y\) first appeared
at \(w_1\), \(f_2(2)=2\) because \( (w_2)_2 = Z \) first appeared
at \(w_2\), \(j_2=1\) and \( f_2(2) \nleq f_2(j_2) \) since \(2 \nleq
1\). We denote by $u \xArrow{\gamma}{\mathbf{df}} w$ a depth-first
step sequence and call it depth-first derivation when \(u\in\Vars\) and \(w\in\Sigma^*\).

\noindent {\bf Depth-first derivation-based semantics.}
In previous work \cite{gik13}, we defined the semantics of a
procedural program based on the control word of the derivation instead of the produced words. We briefly
recall this definition here. Given a depth-first
derivation \(X \xArrow{\gamma}{\mathbf{df}} w\), the
relation \(\sem{\gamma} \subseteq \zed^{\vec{x}} \times \zed^{\vec{x}}\)
is defined inductively on \(\gamma\) as follows:
\begin{inparaenum}[\upshape(\itshape i\upshape)]
\item \(\sem{(X,t)} = \rho_t\),
\item \(\sem{(X,t\, Y) \cdot \gamma'} = \rho_t \circ \sem{\gamma'}\) where
       \(Y \xArrow{\gamma'}{\mathbf{df}} w'\), and
\item \(\sem{(X,\langle\!t \, Y \, t\!\rangle \, Z) \cdot \gamma' \cdot \gamma''} = 
      \sem{(X,\langle\!t \, Y \, t\!\rangle \, Z) \cdot \gamma'' \cdot \gamma'} = 
      \left(\left(\rho_{\langle\! t} \circ \sem{\gamma'} \circ \rho_{t\!\rangle}\right) 
      \cap \phi_t\right) \circ \sem{\gamma''}\) where \(Y \xArrow{\gamma'}{\mathbf{df}} w'\)
      and \(Z \xArrow{\gamma''}{\mathbf{df}} w''\).
\end{inparaenum}
We showed \cite[Lemma 2]{gik13}
that, whenever \(X \xArrow{\gamma}{\mathbf{df}} w\), we have \(\sem{w}\neq\emptyset\) if{}f \(\sem{\gamma}\neq\emptyset\).

\noindent {\bf Index-bounded derivations.}
A step \(u \Rightarrow v\) is
said to be \(k\)-index (\(k>0\)) if{}f 
neither \(u\) nor \(v\) contains \(k+1\) occurrences of nonterminals, i.e.
\( \len{\proj{u}{\Vars}} \leq k\) and \( \len{\proj{v}{\Vars}} \leq k\).
We denote by \(u \xArrow{\gamma}{(k)} v\)
a \(k\)-index step sequence and by
\(u \xArrow{\gamma}{\df{k}} v\) a step
sequence which is both depth-first and \(k\)-index.
For $X \in \Vars$, $Y \in \Vars \cup \{\varepsilon\}$ and \(k>0\), we
define the \(k\)-\emph{index language} \(L^{(k)}_{X,Y}(G) = \{u\,
v\in\Sigma^* \mid \exists \gamma\in\prod^* \colon
X \xArrow{\gamma}{(k)} u \, Y \, v\}\), the \(k\)-\emph{index
depth-first control set} \(\Gamma_{X,Y}^{\df{k}}(G)
= \{\gamma\in\prod^* \mid \exists u,v \in \Sigma^* \colon
X \xArrow{\gamma}{\df{k}} u\, Y\, v\}\). We write \(L^{(k)}_X(G)\)
and \(\Gamma_X^{\df{k}}(G)\) when \(Y=\varepsilon\), and drop \(G\) from
the previous notations, when the grammar is clear from the
context. For instance, for the grammar in
Fig. \ref{fig:running-example} (c), we have $L_{X_1}^{(2)}(G) =
\set{\left(\mathbf{t_1}\langle\!\mathbf{t_2}\right)^n \mathbf{t_4}
\left(\mathbf{t_2}\!\rangle\mathbf{t_3}\right)^n \mid n \in \nat} = L_{X_1}(G)$ 
and $\Gamma_{X_1}^{\df{2}}
= \left(\mathbf{p_1}\mathbf{p_2}\mathbf{p_3}\right)^*
\left(\mathbf{p_4} \cup \mathbf{p_1}\mathbf{p_2}\mathbf{p_4}\mathbf{p_3}\right)$.
\smallskip

\begin{theorem}[Lemma~2 \cite{Luker80}, Theorem~1 \cite{Luker78}]\label{thm:luker}
Given a grammar \(G = \tuple{\Vars, \Sigma, \prod} \) and \(X\in\Vars\):
\begin{compactitem}
\item for all \(w\in\Sigma^*\),  \(X \xRightarrow[(k)]{}^* w \) if and only if
	\(X \xRightarrow[\df{k}]{}^* w\);
\item if \(L_X(G) \subseteq \pat\) for a bounded expression \(\pat\) over \(\Sigma\) then
	\(L_X(G) = L_X^{(K)}(G) \) where \(K = O(\len{G}) \).
\end{compactitem}
\end{theorem}

The introduction of the notion of index naturally calls for an index
dependent semantics and an index dependent reachability problem. As we
will see later, we have tight complexity results when it comes to the
index dependent reachability problem. Given \(k>0\),
let \(\sem{\mathcal{P}}^{(k)} = \bigcup_{w \in L_I^{(k)}(G)} \sem{w}\)
and let \(\sem{\mathcal{P}}^{(k)}_{\pat} = \bigcup_{w \in
L_I^{(k)}(G)\cap\pat} \sem{w}\). Thus we define, for a constant \(k\)
not part of the input, the
problem \(\foreach^{(k)}(\mathcal{P},\pat)\), which asks
whether \(\sem{\mathcal{P}}^{(k)}_\pat \neq \emptyset\).

\noindent {\bf Finite representations of bounded-index depth-first control sets.}
It is known that the set of \(k\)-index depth-first derivations of a
grammar \(G\) is recognizable by a finite
automaton \cite[Lemma~5]{Luker80}. Below we give a formal definition of this
automaton, that will be used to produce bounded control sets for
covering the language of \(G\). Moreover, we provide an upper bound on
its size, which will be used to prove an upper bound for the time
to compute this set (Section \ref{sec:bounded-control-sets}).

Given \(k>0\) and a grammar \(G=\tuple{\Vars, \Sigma, \prod}\), we define a
labeled graph \(A^{\df{k}}_G\) such that its paths defines the set of
\(k\)-index depth-first step sequences of \(G\).  To define the vertices and
edges of this graph, we introduce the notion of ranked words, where the rank
plays the same r\^ole as the value \(f_m(i)\) defined previously. The advantage
of ranks is that only \(k\) of them are needed for \(k\)-index depth-first
derivations whereas the set of \(f_m(i)\) values grows with the length of
derivations. Since we restrict ourselves to \(k\)-index
depth-first derivations, we thus only need \(k\) ranks, from \(0\) to \(k-1\).  
The rank based definition of depth-first derivations can be found in
Appendix~\ref{app:fsa-dfk}.
\comment[pg]{Turn into ``The rank based definition of depth-first derivations can be found in the technical report \cite{}.''}

For a \(d\)-dimensional vector \(\vec{v}\in\nat^d\), we
write \((\vec{v})_i\) for its \(i\)th element (\(1 \leq i \leq d\)). A vector \(\boldsymbol{v}\in\nats^d\) is
said to be \emph{contiguous} if
\(\set{(\boldsymbol{v})_1,\ldots,(\boldsymbol{v})_d} = \set{0,\ldots,k} \), 
for some \(k\geq 0\). 
Given an alphabet \(\Sigma\) define the ranked alphabet \(\Sigma^{\nats}\) 
to be the set \( \{ \sigma^{\tuple{i}} \mid \sigma\in\Sigma, i\in\nats \} \).
A ranked word is a word over a ranked alphabet. 
Given a word \(w\) of length \(n\) and an \(n\)-dimensional
vector \(\boldsymbol{\alpha}\in\nat^n\), the \emph{ranked
word} \(w^{\boldsymbol{\alpha}}\) is the sequence
\({(w)_1}^{\tuple{(\boldsymbol{\alpha})_1}}\ldots {(w)_n}^{\tuple{(\boldsymbol{\alpha})_n}}\), 
in which the \(i\)th element of \(\boldsymbol{\alpha}\) annotates the \(i\)th symbol of \(w\). We also denote \(w^{ \rank{c} } =
{(w)_1}^{\tuple{c}} \ldots {(w)_{\len{w}}}^{\tuple{c}}\) as a
shorthand. Let \(A_G^{\df{k}} = \tuple{Q, \prod, \rightarrow}\) be the
following labeled graph, where: 
{\setlength\abovedisplayskip{4pt}
\setlength\belowdisplayskip{4pt}
\[Q = \{ w^{\boldsymbol{\alpha}} \mid w \in \Vars^{*}, \len{w}\leq k, \boldsymbol{\alpha}\in\nat^{\len{w}} 
\text{ is contiguous}, (\boldsymbol{\alpha})_1 \leq \cdots \leq (\boldsymbol{\alpha})_{\len{w}} \}\]}
is the set of vertices, the edges are labeled by the set \(\prod\) of
productions of \(G\), and the edge relation is defined next. For all
vertices \(q,q' \in Q\) and labels \( (X,w)\in \Delta\), we
have \(q \xrightarrow{(X,w)} q'\) if and only if
\begin{compactitem}
\item \( q = u\, X^{\tuple{i}}\, v\) for some \(u,v\), where \(i\) is the maximum rank in \(q\), and 
\item \( q'= u\, v\, (\proj{w}{\Vars})^{\rank{i'}} \), 
where \(\len{u\, v\, (\proj{w}{\Vars})^{\rank{i'}}}\leq k\) and \( i' = \begin{cases}
  0 & \text{if } u\, v = \varepsilon\\
	i & \text{else if } \proj{(u\, v)}{\Vars^{\tuple{i}}} = \varepsilon\\
	i+1 & \text{else}
\end{cases}\) 
\end{compactitem}
We denote
by \(\len{A^{\df{k}}_G} = \card{Q}\) the size (number of vertices)
of \(A^{\df{k}}_G\). In the following, we omit the subscript $G$
from \(A^{\df{k}}_G\), when the grammar is clear from the context. For
example, the graph \(A^{\df{2}}\) for the grammar from
Fig. \ref{fig:running-example} (c), is the subgraph of
Fig. \ref{fig:running-example} (d) enclosed in a dashed line.

\begin{lemma}\label{fsa-dfk}
Given $G = \tuple{ \Vars, \Sigma, \prod }$, and $k > 0$, for each 
$X \in \Vars$, $Y \in \Vars \cup \{\varepsilon\}$ and $\gamma \in \prod^*$, we have
\(\gamma \in \Gamma_{X,Y}^{\df{k}}(G)\) if and only if \(X^{\tuple{0}} \arrow{\gamma}{} Y^{\tuple{0}}\)
is a path in \(A^{\df{k}}_G\). Moreover, we have \(\len{A^{\df{k}}_G} = \len{G}^{\mathcal{O}(k)}\).
\end{lemma}

\section{A Decision Procedure for \(\foreach(\mathcal{P},\pat)\)}
\label{sec:bounded-control-sets}

In this section we describe a decision procedure for the
problem \(\foreach(\mathcal{P},\pat)\) where \(\mathcal{P} = \tuple{G,
I, \sem{.}}\) is an octagonal program, whose underlying grammar
is \(G=\tuple{\Vars,\Sigma,\prod}\), and \(\pat=w_1^*\ldots w_d^*\) is
a bounded expression over \(\Sigma\). The procedure follows the
roadmap described next. 

First, we compute, in time polynomial in the sizes of
\(\mathcal{P}\) and \(\pat\), a set of programs
\(\{ \mathcal{P}_i=\tuple{G^\cap,X_i,\sem{.}} \}_{i=1}^\ell\), 
such that \(L_I(G) \cap \pat = \bigcup_{i=1}^{\ell} L_{X_i}(G^\cap)\),
which implies \(\sem{\mathcal{P}}_\pat
= \bigcup_{i=1}^\ell \sem{\mathcal{P}_i}\).  The grammar \(G^\cap\) is
an automata-theoretic product between the grammar \(G\) and the bounded expression \(\pat\). For space
reasons, the formal definition of \(G^\cap\) is deferred to
Appendix \ref{sec:easy}, and we refer the reader to
Example \ref{ex:bounded-expression1}. 
\comment[pg]{Turn into ``Details are in the technical report \cite{} but an example of \(G^{\cap}\) is provided below.''}
Deciding \(\foreach(\mathcal{P},\pat)\)
reduces thus to deciding several
instances \(\set{\foreach(\mathcal{P}_i,\pat)}_{i=1}^\ell\) of the
fo-reachability problem.

\begin{example}\label{ex:bounded-expression1}
  Let us consider the bounded expression $\pat = (ac)^* \, (ab)^* \, (db)^*$. 
  Consider the grammar $G^\pat$ with the following productions: 
  \(\textsc{q}_1^{(1)} \rightarrow a\, \textsc{q}_2^{(1)} \mid \varepsilon\),  
  \(\textsc{q}_1^{(2)} \rightarrow a\, \textsc{q}_2^{(2)} \mid \varepsilon\),
  \(\textsc{q}_1^{(3)} \rightarrow d\, \textsc{q}_2^{(3)} \mid \varepsilon\),
  \(\textsc{q}_2^{(1)} \rightarrow c\, \textsc{q}_1^{(1)} \mid c\, \textsc{q}_1^{(2)} \mid c\, \textsc{q}_1^{(3)}\),
  \(\textsc{q}_2^{(2)} \rightarrow b\, \textsc{q}_1^{(2)} \mid b\, \textsc{q}_1^{(3)}\),
  \(\textsc{q}_2^{(3)} \rightarrow b\, \textsc{q}_1^{(3)}\). 
  It is easy to check that $\pat = \bigcup_{i=1}^3L_{\textsc{q}_1^{(i)}}(G^\pat)$.  
  Let $G = \tuple{\set{X,Y,Z,T}, \set{a,b,c,d}, \prod}$ where 
  $\prod = \set{X \rightarrow aY,\; Y \rightarrow Zb,\; Z \rightarrow cT,\; Z \rightarrow \varepsilon,\; T \rightarrow Xd}$, 
  i.e.\ we have $L_X(G)=\set{(ac)^n \, ab \, (db)^n \mid n \in \nats}$.  The following productions define a grammar $G^\cap$: 
  {\small\[\begin{array}{rcllcrcllcrcll}
\lbrack\textsc{q}_1^{(j)} X \textsc{q}_1^{(3)}\rbrack & \stackrel{p_1}{\rightarrow} & a \,    \lbrack\textsc{q}_2^{(j)} Y \textsc{q}_1^{(3)}\rbrack & & \hspace*{.5cm} &
\lbrack\textsc{q}_2^{(1)} Y \textsc{q}_1^{(3)}\rbrack & \stackrel{p_2}{\rightarrow} & \lbrack\textsc{q}_2^{(1)} Z \textsc{q}_2^{(3)}\rbrack \,    b \\
\lbrack\textsc{q}_2^{(1)} Z \textsc{q}_2^{(3)}\rbrack & \stackrel{p_3}{\rightarrow} & c \,    \lbrack\textsc{q}_1^{(j)} T \textsc{q}_2^{(3)}\rbrack & & \hspace*{.5cm} &
\lbrack\textsc{q}_2^{(2)} Z \textsc{q}_2^{(2)}\rbrack & \stackrel{p_4}{\rightarrow} & \varepsilon\\
\lbrack\textsc{q}_1^{(j)} T \textsc{q}_2^{(3)}\rbrack & \stackrel{p_5}{\rightarrow} & \lbrack\textsc{q}_1^{(j)} X \textsc{q}_1^{(3)}\rbrack \,    d & \mbox{, for $j=1,2$} & &
\lbrack\textsc{q}_1^{(2)} X \textsc{q}_1^{(3)}\rbrack & \stackrel{p_6}{\rightarrow} & a \,    \lbrack\textsc{q}_2^{(2)} Y \textsc{q}_1^{(3)}\rbrack \\
\lbrack\textsc{q}_2^{(2)} Y \textsc{q}_1^{(3)}\rbrack & \stackrel{p_7}{\rightarrow} & \lbrack\textsc{q}_2^{(2)} Z \textsc{q}_2^{(2)}\rbrack \,    b 
\end{array}\]}
One can check $L_X(G) =  L_X(G) \cap \pat = L_{[\textsc{q}_1^{(1)} X \textsc{q}_1^{(3)}]}(G^\cap) \cup L_{[\textsc{q}_1^{(2)} X
      \textsc{q}_1^{(3)}]}(G^\cap)$. \(\blacksquare\)
\end{example}

A bounded expression \(\pat=w_1^* \ldots w_d^*\) over
alphabet \(\Sigma\) is said to be {\em $d$-letter-bounded} (or simply
letter-bounded, when $d$ is not important) when $\len{w_i}=1$, for all
$i=1,\ldots,d$. A letter-bounded expression \(\patt\) is \emph{strict}
if all its symbols are distinct. A language $L \subseteq \Sigma^*$ is
(strict, letter-)~bounded if{}f $L \subseteq \pat$, for some (strict,
letter-)~bounded expression $\pat$.

Second, we reduce the problem from \(\pat=w_1^* \ldots w_d^*\) to the
strict letter-bounded case \(\patt=a_1^* \ldots a_d^*\), by building a
grammar \(G^\bowtie\), with the same nonterminals as \(G^\cap\), such
that, for each \(i=1,\ldots,\ell\)
\begin{inparaenum}[\upshape(\itshape i\upshape)]
\item \(L_{X_i}(G^\bowtie) \subseteq \patt\),
\item \(w_1^{i_1}\ldots w_d^{i_d}\in L_{X_i}^{(k)}(G^{\cap})\) if{}f 
	\(a_1^{i_1}\ldots a_d^{i_d}\in L_{X_i}^{(k)}(G^{\bowtie}) \), for all \(k>0\)
\item\label{it:transposition} from each control set \(\Gamma\) that covers the language 
\(L^{(k)}_{X_i}(G^\bowtie) \subseteq \hat{L}_{X_i}(\Gamma,G^\bowtie)\) for some \(k>0\), 
one can compute, in polynomial time, a control
set \(\widetilde{\Gamma}\) that covers the
language \(L^{(k)}_{X_i}(G^\cap) \subseteq \hat{L}_{X_i}(\widetilde{\Gamma},G^\cap)\).
\end{inparaenum}

\begin{example}[contd. from Example \ref{ex:bounded-expression1}]\label{ex:bounded-expression2}
  Let $\mathcal{A} = \set{a_1, a_2, a_3}$, \(\patt=a_1^*a_2^*a_3^*\) and 
  $h \colon \mathcal{A} \rightarrow \Sigma^*$ be the homomorphism given 
  by $h(a_1) = ac, h(a_2) = ab$ and $h(a_3) = db$. The grammar $G^\bowtie$ 
  results from deleting \(a\)'s and \(d\)'s in \(G^{\cap}\) and replacing 
  \(b\) in \(p_2\) by \(a_3\), \(b\) in \(p_7\) by \(a_2\) and \(c\) by \(a_1\).
  Then, it is easy to check that $h^{-1}( L_X(G) ) \cap \patt = 
  L_{[\textsc{q}_1^{(1)} X \textsc{q}_1^{(3)}]}(G^\bowtie) \cup 
  L_{[\textsc{q}_1^{(2)} X \textsc{q}_1^{(3)}]}(G^\bowtie) = 
  \set{a_1^n \, a_2 \, a_3^n \mid n \in \nats}$. \hfill \(\blacksquare\)
\end{example}

Third, for the strict letter-bounded grammar \(G^\bowtie\), we compute
a control set \(\Gamma \subseteq (\prod^{\bowtie})^*\) using the
result of Theorem \ref{thm:letter-bounded-control-set}, which yields
a set of bounded expressions \(\mathcal{S}_\patt = \set{\Gamma_{i,
1}, \ldots, \Gamma_{i, m_i}}\), such
that \(L_{X_i}^{(k)}(G^\bowtie) \subseteq \bigcup_{j=1}^{m_i}
\hat{L}_{X_i}(\Gamma_{i, j} \cap \Gamma_{X_i}^{\df{k+1}}, G^\bowtie)\). 
By applying the aforementioned transformation
({\itshape\ref{it:transposition}}) from \(\Gamma\)
to \(\widetilde{\Gamma}\), we obtain
that \(L^{(k)}_{X_i}(G^\cap) \subseteq \bigcup_{j=1}^{m_i}
\hat{L}_{X_i}(\widetilde{\Gamma}_{i, j} \cap \Gamma_{X_i}^{\df{k+1}}, G^\cap)\).
Theorem \ref{thm:luker} allows to effectively compute value \(K>0\)
such that \(L_{X_i}(G^\cap) = L_{X_i}^{(K)}(G^\cap)\), for
all \(i=1,\ldots,\ell\). Thus we obtain\footnote{Because
\(L_{X_i}(G^\cap) \subseteq L_{X_i}^{(K)}(G^\cap) \subseteq \bigcup_{j=1}^{m_i} 
\hat{L}_{X_i}(\widetilde{\Gamma}_{i, j} \cap \Gamma_{X_i}^{\df{k+1}}, G^\cap) \subseteq L_{X_i}(G^\cap)\) .}
\(L_{X_i}(G^\cap) = \bigcup_{j=1}^{m_i}  \hat{L}_{X_i}(\widetilde{\Gamma}_{i, j} 
\cap \Gamma_{X_i}^{\df{K+1}}, G^\cap)\), for all \(i=1,\ldots,\ell\).

The final step consists in building a finite
automaton \(A^{\df{K+1}}\) that recognizes the control
set \(\Gamma_{X_i}^{\df{K+1}}\) (Lemma \ref{fsa-dfk}). This yields a
procedure-less program \(\mathcal{Q}\), whose control structure is
given by \(A^{\df{K+1}}\), and whose labels are given by the semantics
of control words. We recall that, for every word \(w \in
L_{X_i}(G^\cap)\) there exists a control
word \(\gamma \in \Gamma_{X_i}^{\df{K+1}}\) such
that \(\sem{w}\neq\emptyset\) if{}f \(\sem{\gamma}\neq\emptyset\). We
have thus reduced each of the
instances \(\set{\foreach(\mathcal{P}_i,\pat)}_{i=1}^\ell\) of the
fo-reachability problem to a set of
instances \(\{\foreach(\mathcal{Q},\widetilde{\Gamma}_{i, j}) \mid 1\leq i\leq \ell,\, 1\leq j\leq m_i\}\). The
latter problem, for procedure-less programs, is decidable in
\textsc{Nptime} \cite{bik14}. Next is our main result whose proof is in Appendix~\ref{app:fo-reachability}.
\comment[pg]{Turn into ``A detailed proof of the main result, stated next, is given in the technical report \cite{}''.}

\begin{theorem}\label{thm:fo-reachability}
Let \(\mathcal{P}=\tuple{G,I,\sem{.}}\) be an octagonal program,
where \(G=\tuple{\Vars,\Sigma,\prod}\) is a grammar, and \(\pat\) is a
bounded expression over \(\Sigma\). Then the
problem \(\foreach(\mathcal{P}, \pat)\) is decidable
in \textsc{Nexptime}, with a \textsc{Np}-hard lower bound.  If,
moreover, \(k\) is a constant, \(\foreach^{(k)}(\mathcal{P},\pat)\) is \textsc{Np}-complete.
\end{theorem}

The rest of this section describes the construction of the control
sets \(\mathcal{S}_\patt\) and gives upper bounds on the time needed
for this computation. We use the following ingredients:
\begin{inparaenum}[\upshape(\itshape i\upshape)] 
\item Algorithm \ref{alg:constant-control-set} for building 
bounded control sets for $s$-letter bounded languages, where $s\geq0$ is a constant 
(in our case, at most $2$) (Section \ref{sec:constant-control-set}), and
\item a decomposition of $k$-index depth-first derivations, that 
distinguishes between a prefix producing a word from the \(2\)-letter
bounded expression $a_1^*a_d^*$, and a suffix producing two words
included in bounded expressions strictly smaller than \(\patt\)
(Section \ref{sec:bounded}). 
\end{inparaenum}
The decomposition enables the generalization from \(s\)-letter bounded
languages where \(s\) is a constant to arbitrary letter bounded
languages. In particular, the required set of bounded
expressions \(\mathcal{S}_\patt\) is built inductively over the
structure of this decomposition, applying at each step
Algorithm \ref{alg:constant-control-set} which computes bounded
control sets for 2-letter bounded languages. The main algorithm
(Algorithm \ref{alg:bounded-control-set}) returns a finite
set \(\mathcal{S}_\patt\) of bounded expressions
$\{\Gamma_1, \ldots, \Gamma_m\}$. Below we abuse notation and
write \(\bigcup\mathcal{S}_\patt\)
for \(\bigcup_{i=1}^m \Gamma_i \). The time needed to build each
bounded expression $\Gamma_i \in \mathcal{S}_\patt$ is
$\len{G}^{\mathcal{O}(k)}$ and does not depend of $\len{\patt}=d$,
whereas the time needed to build the entire set \(\mathcal{S}_\patt\)
is \(\len{G}^{\mathcal{O}(k)+d}\). These arguments come in handy when
deriving an upper bound on the (non-deterministic) time complexity of
the fo-reachability problem for programs with arbitrary call graphs.
A non-deterministic version of Algorithm \ref{alg:bounded-control-set}
that choses one set \(\Gamma_i\in\mathcal{S}_\patt\), instead of
building the whole set \(\mathcal{S}_\patt\), is used to establish the
upper bounds for the \(\foreach(\mathcal{P}, \pat)\)
and \(\foreach^{(k)}(\mathcal{P}, \pat)\) problems in the proof of
Theorem \ref{thm:fo-reachability}.

\subsection{Constant $s$-Letter Bounded Languages}\label{sec:constant-control-set}

Here we define an algorithm for building
bounded control sets that are sufficient for covering a $s$-letter
bounded language \(L_X(G) \subseteq a_1^* \ldots a_s^*\), when
$s\geq0$ is a constant\footnote{In our case $s=0,1,2$, but the
construction can be generalized to any constant $s \geq 0$. }, i.e.\
not part of the input of the algorithm. In the following, we consider
the labeled graph \(A^{\df{k}} = \tuple{ Q, \prod, \rightarrow}\),
whose paths correspond to the $k$-index depth-first step sequences of
$G$ (Lemma~\ref{fsa-dfk}). Recall that the number of vertices in this
graph is $\len{A^{\df{k}}} \leq \len{G}^{2k}$. 

Given $q,q'\in Q$, we denote by $\Pi(q,q')$ the set of paths with
source $q$ and destination $q'$. For a path \(\pi\), we denote
by \(\omega(\pi) \in \prod^*\) the sequence of edge labels on \(\pi\).
A path \(\pi\) is a \emph{cycle} if its endpoints coincide.
Furthermore, the path is said to be an \emph{elementary cycle} if it
contains no other cycle than itself.  Finally, \(\pi\) is acyclic if
it contains no cycle.  The word \emph{induced} by a path
in \(A^{\df{k}}\) is the sequence of terminal symbols generated by the
productions fired along that path. Observe that,
since \(L_X(G) \subseteq a_1^* \ldots a_s^*\), any word induced by a
subpath of some path \(\pi\in\Pi(X^{\tuple{0}},\varepsilon)\) is
necessarily of the form \(a_1^{i_1} \ldots a_s^{i_s}\), for
some \(i_1,\ldots,i_s\geq0\).

Algorithm~\ref{alg:constant-control-set} describes the effective
construction of a bounded expression \(\Gamma\) over the productions
of $G$ using the sets of elementary cycles of $A^{\df{k}}$. The crux
is to find, for each vertex $q$ of $A^{\df{k}}$, a subset $C_q$ of
elementary cycles having $q$ at the endpoints, such that the set of
words induced by \(C_q\) is that of the entire set of elementary
cycles having $q$ at endpoints. Since the only vertex occurring more
than once in an elementary cycle \(\rho\) is the endpoint \(q\), we
have that \(\len{\rho}\) is at most the number of
vertices \(\len{A^{\df{k}}}\), and each production rule generates at
most $2$ terminal symbols, hence no word induced by a elementary cycle
is longer than \(2\len{A^{\df{k}}} \leq 2\len{G}^{2k}\). The number of
words \( a_1^{i_1} \ldots a_s^{i_s} \) induced by elementary cycles
with endpoints $q$ is thus bounded by the number of nonnegative
solutions of the inequality $x_1 + \cdots + x_s \leq 2\len{G}^{2k}$,
which, in turn, is of the order of $\len{G}^{\mathcal{O}(k)}$. So for
each vector \(\boldsymbol{v}\in\nat^s\) such that $(\boldsymbol{v})_1
+ \cdots + (\boldsymbol{v})_s \leq 2\len{G}^{2k}$, it suffices to
include in $C_q$ only one elementary cycle inducing the
word \(a_1^{(\boldsymbol{v})_1} \ldots a_s^{(\boldsymbol{v})_s}\).
Thus it is sufficient to consider sets $C_q$ of cardinality
$\card{C_q}=\len{G}^{\mathcal{O}(k)}$, for all $q \in Q$.

Lines (\ref{ln:h-start}--\ref{ln:h-end}) of Algorithm
\ref{alg:constant-control-set} build a graph $\mathcal{H}$
with vertices $\tuple{q,a_1^{i_1} \ldots a_s^{i_s}}$, where $q \in Q$
is a vertex of $A^{\df{k}}$ and \(i_1,\ldots,i_s\) a solution to the
above inequality (line~\ref{ln:h-start}), hence \(\mathcal{H}\) is a
finite and computable graph. There is an edge between two vertices
$\tuple{q,a_1^{i_1} \ldots a_s^{i_s}}$ and $\tuple{q',a_1^{j_1} \ldots
a_s^{j_s}}$ in $\mathcal{H}$ if and only if $q \arrow{p}{} q'$ in
$A^{\df{k}}$ and $a_{\ell}^{j_{\ell}} = a_{\ell}^{i_{\ell}} \cdot
(\proj{p}{a_{\ell}})$ for every \(\ell\), that is \(j_{\ell}\) is the
sum of \(i_{\ell}\) and the number of occurrences of \(a_{\ell}\)
produced by \(p\) (which is precisely captured by the
word \(\proj{p}{a_{\ell}}\)) (line \ref{ln:h-trans}). The sets $C_q$
are computed by applying the Dijkstra's single source shortest path
algorithm\footnote{We consider all edges to be of weight \(1\).}  to
the graph $\mathcal{H}$ (line \ref{ln:dijkstra}) and retrieving
in \(C_q\) the paths \(\tuple{q,\varepsilon} \rightarrow^* \tuple{q,a_1^{i_1} \ldots
a_s^{i_s}}\), such that \(i_1 + \cdots + i_s \leq 2\len{G}^{2k} \)
(line \ref{ln:cq}).

For a finite set of words \(S=\set{u_1,\ldots,u_h}\), the function
\(\Call{Concat}{S}\) returns the bounded expression \(u_1^* \ldots u_h^*\).
Algorithm \ref{alg:constant-control-set} uses this function to build a
bounded expression $\Gamma$ that covers all words induced by paths
from \(\Pi(X^\tuple{0},\varepsilon)\). This construction relies on the
following argument: for each $\pi\in\Pi(X^\tuple{0},\varepsilon)$, there
exists another path $\pi'\in\Pi(X^\tuple{0},\varepsilon)$, such that
their induced words coincide, and, moreover, $\pi'$ can be
factorized%
as
$\varsigma_1 \cdot \theta_1 \cdots \varsigma_\ell \cdot \theta_\ell \cdot \varsigma_{\ell+1}$,
where \(\varsigma_1 \in \Pi(X^{\tuple{0}},q_1)\), \(\varsigma_{\ell+1} \in \Pi(q_\ell,
\varepsilon)\) and \(\varsigma_{j} \in \Pi(q_{j-1},q_{j})\) for each \(1<
j \leq \ell\) are acyclic paths, $\theta_1, \ldots, \theta_\ell$ are
elementary cycles with endpoints $q_{1}, \ldots, q_{\ell}$, respectively,
and $\ell \leq \len{A^{\df{k}}}$. Thus we can cover each segment
$\varsigma_i$ by a bounded expression $C=\Call{Concat}{\prod}^{\len{G}^{2k}-1}$
(line \ref{ln:concat-prod}), and each segment $\theta_j$ by the
bounded expression \(B_0=\Call{Concat}{\set{\omega(\pi) \mid \pi \in
C_{q_j}}}\) (line \ref{ln:concat-cq}), yielding the required
expression $\Gamma$. The following lemma proves the correctness of
Algorithm \ref{alg:constant-control-set} and gives an upper bound on its runtime.

\begin{lemma}\label{lem:ginsbook-d}
  Let $G = \tuple{ \Vars, \mathcal{A}, \prod }$ be a grammar
  and \(a_1^*\ldots a_s^*\) is a strict \(s\)-letter-bounded
  expression over \(\mathcal{A}\), where \(s\geq0\) is a constant.
  Then, for each $k > 0$ there exists a bounded expression \(\Gamma\)
  over \(\prod\) such that, for all $X \in \Vars$ and
  $Y \in \Vars \cup \{\varepsilon\}$, we have $L^{(k)}_{X,Y}(G)
  = \hat{L}_{X,Y}(\Gamma\cap\Gamma_{X,Y}^{\df{k}},G)$, provided
  that \(L_{X,Y}(G)\subseteq a_1^*\ldots a_s^*\). Moreover, \(\Gamma\)
  is computable in time $\len{G}^{\mathcal{O}(k)}$.
\end{lemma}

\begin{algorithm}[thb]
{\scriptsize
\begin{algorithmic}[0]
\State {\bf input} A grammar $G = \tuple{\Vars, \mathcal{A}, \prod}$, 

       \State \hspace*{7mm} a strict \(s\)-letter-bounded expression \(a_{1}^* \ldots a_{s}^*\) 
       over \(\mathcal{A}\), where \(s\geq0\) is a fixed constant,

       \State \hspace*{7mm} and $k > 0$

  \State {\bf output} a bounded expression $\Gamma$ over \(\prod\) such that
  $L^{(k)}_{X,Y}(G) = \hat{L}_{X,Y}(\Gamma \cap \Gamma_{X,Y}^{\df{k}}, G)$
	for all $X \in \Vars$ and $Y \in \Vars \cup \{\varepsilon\}$, such that
	\(L_{X,Y}(G) \subseteq a_1^*\ldots a_s^*\)
\end{algorithmic}
\begin{algorithmic}[1]
  \Function{ConstantBoundedControlSet}{$G,a_{1}^* \ldots a_{s}^*,k$}

	\State \(\mathit{Val} \leftarrow \{a_1^{k_1} \cdots a_s^{k_s} \mid 
	\sum_{j=1}^s k_{j} \leq 2\len{G}^{2k}\}\) \label{ln:h-start}

        \State $V \leftarrow Q \times \mathit{Val}$
				\Comment{\(Q\) are the vertices of \(A^{\df{k}}\), considering $\card{Q} \leq \len{G}^{2k}$ suffices}

	\State $\delta \leftarrow \{\tuple{q,a_1^{i_1}\ldots a_s^{i_s}} \arrow{p}{}
	\tuple{q', a_1^{j_1}\ldots a_s^{j_s}} \mid q \arrow{p}{} q'
	~\mbox{in $A^{\df{k}}$, }~\forall \ell \in \set{1,\ldots,s} \ldotp
        a^{j_{\ell}}_{\ell} {=} a^{i_\ell}_{\ell}\cdot \bigl(\proj{p}{a_{\ell}}\bigr)\}$
        \label{ln:h-trans}

        \State $\mathcal{H} \leftarrow \tuple{V, \prod, \delta}$\label{ln:h-end}

        \State $B_0 \leftarrow \varepsilon$

        \State \(\Call{DijkstraShortestPaths}{\mathcal{H}}\)
        \label{ln:dijkstra}

  \For{$q \in Q$}\label{ln:b0-start}

	\State \(C_q \leftarrow \bigcup_{w\in\mathit{Val}}
        \Call{GetShortestPath}{\mathcal{H},\tuple{q,\varepsilon}, \tuple{q,w}} \)
        \label{ln:cq}

	\State $B_0 \leftarrow B_0 \cdot \Call{Concat}{\{\omega(\pi) \mid \pi \in C_q\} }$
        \label{ln:concat-cq}

  \EndFor\label{ln:b0-end}

  \State \(C \leftarrow \varepsilon\) 
  \label{ln:c-start}

  \For{\(i = 1 \ldots \len{G}^{2k}-1\)}
           \State \(C \leftarrow C \cdot \Call{Concat}{\prod}\)
           \label{ln:concat-prod}
  \EndFor\label{ln:c-end}

  \State $\Gamma \leftarrow \varepsilon$ 

  \For{$i = 1 \ldots \len{G}^{2k}$}\label{ln:pat-start}
  \State $\Gamma \leftarrow \Gamma \cdot C \cdot B_0$
  \EndFor\label{ln:pat-end}

  \State $\Gamma \leftarrow \Gamma \cdot C \cdot B_0 \cdot C$
  \State {\bf return} $\Gamma$
  \EndFunction
\end{algorithmic}
\caption{Control Sets for the Case of Constant Size Bounded Expressions}\label{alg:constant-control-set}
}
\end{algorithm}

\subsection{The General Case}\label{sec:bounded}

The key to the general case is a lemma decomposing derivations. 

\noindent {\bf Decomposition Lemma.} %
Our construction of a bounded control set that covers a strict
letter-bounded context-free language $L_X(G) \subseteq a_1^* \ldots
a_d^*$ is by induction on $d \geq 1$, and is inspired by a
decomposition of the derivations in $G$, given by
Ginsburg \cite[Chapter 5.3, Lemma 5.3.3]{ginsburg}. Because
his decomposition is oblivious to the index or the depth-first policy,
it is too weak for our needs. Therefore, we give first a stronger
decomposition result for \(k\)-index depth-first derivations.

Without loss of generality, the decomposition lemma assumes the bounded
expression covering \(L_X(G)\) to be \emph{minimal}: a strict letter-bounded
expression \(\patt\) is \emph{minimal} for a language \(L\) if{}f \(L \subseteq
\patt\) and for every subexpression \(\mathbf{b}'\), resulting from deleting
some \(a_i^*\) from \(\patt\), we have \(L \nsubseteq \mathbf{b}'\).  Clearly,
each strict letter-bounded language has a unique minimal expression.

Basically, for every $k$-index depth-first derivation with control
word $\gamma$, its productions can be rearranged into a $(k+1)$-index
depth-first derivation, consisting of a prefix $\gamma^\sharp$
producing a word in \(a_1^*\, a_d^*\), then a production \(
(X_i,w)\) followed by two control words $\gamma'$ and $\gamma''$ that
produce words contained within two bounded expressions
$a_\ell^* \ldots a_m^*$ and $a_m^* \ldots a_r^*$, respectively, where
$\max(m-\ell,r-m) < d-1$ (Lemma~\ref{lem:ginsbook-surgery}).
Let us first define the partition \((\Varse,\Varsi)\) of \(\Vars\), as
follows: {\setlength\abovedisplayskip{4pt}
\setlength\belowdisplayskip{4pt}
\[
\begin{array}{rclcl}
  Y \in \Varse & \iff &  L_{Y}(G)\cap (a_1 \cdot \mathcal{A}^*) \neq \emptyset & \text{ {\bf and} } &
  L_{Y}(G) \cap (\mathcal{A}^* \cdot a_d) \neq \emptyset\enspace.
    \end{array} \]}
Naturally, define \(\Varsi = \Vars \setminus \Varse\).
Since the bounded expression $a_1^* \ldots a_d^*$ is, by assumption, 
minimal for $L_X(G)$, then \(a_1\) occurs in some word of \(L_X(G)\)
and \(a_d\) occurs in some word of \(L_X(G)\). Thus it is always the
case that \(\Varse \neq \emptyset\), since \(X \in \Varse\). The
partition of nonterminals into $\Varse$ and $\Varsi$ induces a
decomposition of the grammar $G$. First, let $G^\sharp
= \tuple{\Vars,\mathcal{A},\prod^\sharp}$, where:
{\setlength\abovedisplayskip{4pt}
\setlength\belowdisplayskip{4pt}
\[
\prod^\sharp = \set{(X_j, w) \in \prod
\mid X_j \in \Varsi} \cup \set{(X_j, u \, X_r \, v) \in \prod \mid
X_j, X_r \in \Varse}\enspace.\]}
Then, for each production $(X_i,w)
\in \prod$ such that $X_i \in \Varse$ and $w \in (\Varsi \cup
\mathcal{A})^*$, we define the grammar $G_{i,w} = \tuple{ \Vars,
  \mathcal{A}, \prod_{i,w}}$, where:
{\setlength\abovedisplayskip{4pt}
\setlength\belowdisplayskip{4pt}
\[
\prod_{i,w} = \set{(X_j,v) \in \prod \mid X_j\in\Varsi} \cup \set{(X_i,w)}\enspace.
\]}
The decomposition of derivations is formalized by the following lemma: 

\begin{lemma}\label{lem:ginsbook-surgery}
  Given a grammar $G =\tuple{\Vars,\mathcal{A},\prod}$, a nonterminal
  $X\in\Vars$ such that \(L_X(G) \subseteq a_1^* \ldots a_d^*\) for
  some \(d \geq 3\), and \(k>0\), for every derivation
  $X \xRightarrow[\df{k}]{\gamma}_{\raisebox{1.5ex}{{\scriptsize G}}}
  w$, there exists a production $\mathrm{p}=(X_i, a\,y\,b\,z) \in \prod$ with
  $X_i \in \Varse$, $a,b \in \mathcal{A}\cup\set{\varepsilon}$ and
  $y,z \in \Varsi \cup \set{\varepsilon}$, and
  control words $\gamma^\sharp \in (\Delta^\sharp)^*$,
  $\gamma_{y},\gamma_{z} \in (\Delta_{i,aybz})^*$, such
  that \(\gamma^\sharp \; \mathrm{p}\; \gamma_y \; \gamma_z\) is a
  permutation of \(\gamma\)
  and: \begin{compactenum} 
  
  \item\label{item1:ginsbook-surgery} \(X \xRightarrow[\df{k+1}]{\gamma^\sharp}_{\raisebox{1.5ex}{{\scriptsize
  $G^{\sharp}$}}} u\, X_i\, v\) is a step sequence in \(G^\sharp\)
  with \(u, v\in \mathcal{A}^*\);

  \item\label{item2:ginsbook-surgery} \(y
		\xRightarrow[\df{k_y}]{\gamma_{y}}_{\raisebox{1.5ex}{{\scriptsize $G_{i,aybz}$}}} u_y\) and \(z
		\xRightarrow[\df{k_z}]{\gamma_{z}}_{\raisebox{1.5ex}{{\scriptsize $G_{i,aybz}$}}} u_z\) are (possibly empty) derivations 
                in \(G_{i,aybz}\) (\(u_y, u_z \in \mathcal{A}^* \)), for some integers \(k_y,k_z>0\), 
                such that \(\max(k_y,k_z)\leq k\) and \(\min(k_y,k_z)\leq k-1\);

  \item\label{item3:ginsbook-surgery} 
  \(X \xRightarrow[\df{k+1}]{\gamma^\sharp \,
	\mathrm{p} \, \gamma_y \, \gamma_z}_{\raisebox{1.5ex}{{\scriptsize G}}} w\)
	if \(y \xRightarrow[\df{k-1}]{\gamma_y}_{\raisebox{1.5ex}{{\scriptsize $G_{i,aybz}$}}} u_y\), and
    \(X \xRightarrow[\df{k+1}]{\gamma^\sharp \,
    \mathrm{p} \, \gamma_z \, \gamma_y}_{\raisebox{1.5ex}{{\scriptsize G}}} w\)
		if \(z \xRightarrow[\df{k-1}]{\gamma_z}_{\raisebox{1.5ex}{{\scriptsize $G_{i,aybz}$}}} u_z\);        

  \item\label{item4:ginsbook-surgery} \(L_{X,X_i}(G^\sharp) \subseteq
    a_1^* a_d^* \);
  \item\label{item5:ginsbook-surgery} \(L_{y}(G_{i,aybz}) \subseteq
    a_{\ell}^* \ldots a_{m}^*\) if \(y \in \Varsi\), and \(L_{z}(G_{i,aybz})
    \subseteq a_{m}^* \ldots a_{r}^*\) if \(z \in \Varsi\), for some integers
    \(1\leq \ell \leq m \leq r \leq d\), such that \( \max(m-\ell,r-m) < d - 1\).
  \end{compactenum}
\end{lemma}

Let us now turn to the general case, in which the size of the strict
letter-bounded expression \(\patt = a_1^* \ldots a_d^*\) is not constant,
i.e. \(d\) is part of the input of the algorithm. The output of
Algorithm~\ref{alg:bounded-control-set} is a finite set of bounded
expressions $\mathcal{S}_\patt$ such that
\(L^{(k)}_X(G) \subseteq \hat{L}_X(\bigcup \mathcal{S}_\patt 
\cap\Gamma_X^{\df{k+1}}, G)\). 
The construction of the set $\mathcal{S}_\patt$ by
Algorithm~\ref{alg:bounded-control-set} (function
\Call{LetterBoundedControlSet}{}) follows the structure 
of the decomposition of control words given by
Lemma \ref{lem:ginsbook-surgery}. For every $k$-index depth-first derivation with 
control word $\gamma$, its productions can be rearranged into a $(k+1)$-index 
depth-first derivation, consisting of
\begin{inparaenum}[\upshape(\itshape i\upshape)]
\item a prefix $\gamma^\sharp$ producing a word in \(a_1^*\, a_d^*\), then
\item a \emph{pivot} production \( (X_i,w)\) followed by two words $\gamma'$ and $\gamma''$ such that:
\item $\gamma'$ and $\gamma''$ produce words included in two
bounded expressions $a_\ell^* \ldots a_m^*$ and $a_m^* \ldots
a_r^*$, respectively, where $\max(m-\ell,r-m) < d-1$.
\end{inparaenum}
The algorithm follows this decomposition and builds bounded
expressions \(\Gamma^{\sharp}\), \( (X_i,w)^*\), and the sets
\(\mathcal{S}'\) and \(\mathcal{S}''\) with the goal of capturing
\(\gamma^{\sharp}\), \( (X_i,w)\), \(\gamma\) and \(\gamma''\), respectively,
for all the control words such as \(\gamma\).
Because \(\gamma^{\sharp}\) produces a word from \(a_1^*\, a_d^*\),
the bounded expression \(\Gamma^{\sharp}\) is built calling
\Call{ConstantBoundedControlSet}{} (line~\ref{ln:gamma-sharp}). Since
\(\gamma'\) and \(\gamma''\) produce words within two sub-expressions of \(a_1^*
\ldots a_d^*\) with as many as \(d-2\) letters, these cases are handled by
two recursive calls to \Call{LetterBoundedControlSet}{}
(lines~\ref{line:rec-y} and \ref{line:rec-z}).

\begin{algorithm}[p]
{\scriptsize
\begin{algorithmic}[0]
  \State {\bf input} A grammar $G = \tuple{\Vars, \mathcal{A}, \prod}$, a nonterminal $X \in \Vars$,
  \State \hspace*{7mm} a strict \(d\)-letter-bounded expression \(\patt\) over \(\mathcal{A}\), such that \(L_{X}(G) \subseteq \patt\), and $k > 0$

  \State {\bf output} a set $\mathcal{S}_\patt$ of bounded
  expressions over \(\prod\), such that 
	\(L^{(k)}_X(G) \subseteq \hat{L}_X(\bigcup \mathcal{S}_\patt \cap \Gamma_X^{\df{k+1}},G)\)

\end{algorithmic}
\begin{algorithmic}[1]
  \Function{LetterBoundedControlSet}{$G_0, X_0, a_{i_1}^* \ldots a_{i_d}^*, k$}
  \State {\bf match} $G_0$ {\bf with} $\tuple{\Vars,\mathcal{A},\prod_0}$

		\State $a_{j_1}^* \cdots a_{j_s}^* \leftarrow \Call{minimizeExpression}{G_0, X_0, a_{i_1}^* \cdots a_{i_d}^*}$
        \label{ln:minexpr}
	\Comment{\(\{j_1,\ldots,j_s\}\subseteq\{i_1,\ldots,i_d\}\)}
  
	\If{$|a_{j_1}^* \ldots a_{j_s}^*| \leq 2$}
	\State {\bf return} $\{\Call{ConstantBoundedControlSet}{G_0, a_{j_1}^* \ldots a_{j_s}^*,k}\}$
  \EndIf

	\State $(\Varselj,\Varsilj) \leftarrow \Call{partitionNonterminals}{G_0, a_{j_1}^*\, a_{j_s}^*}$
        \label{ln:partition}

  \State $\prod^\sharp \leftarrow \{(X_j, w) \in \prod_0 \mid X_j
    \in \Varsilj\} \cup \{(X_j, u \, X_r \, v) \in \prod_0 \mid X_j, X_r
    \in \Varselj\}$
  
  \State $G^\sharp \leftarrow \tuple{\Vars, \mathcal{A}, \prod^\sharp}$

	\State $\Gamma^\sharp \leftarrow \Call{ConstantBoundedControlSet}{G^\sharp, a_{j_1}^*\, a_{j_s}^*, k+1}$ 
  \label{ln:gamma-sharp}

	\State $\mathcal{S}_{\patt} \leftarrow \emptyset$

  \For{$(X_i, aybz) \in \prod_0$ {\bf such that} $X_i \in \Varselj$, 
  \(a,b \in \mathcal{A} \cup \set{\varepsilon}\) and
  \(y,z \in \Varsilj\cup\set{\varepsilon}\)} 
  \label{line:for-begin}
  
  \If{$L_{X_0,X_i}(G^\sharp) \subseteq a_{j_1}^* \, a_{j_s}^*$}
  \label{line:inclusion}

  \State $\prod_{i,aybz} \leftarrow \{(X_j, v) \in \prod \mid X_j
    \in \Varsilj\} \cup \set{(X_i, a\,y\,b\,z)}$

  \State $G_{i,aybz} \leftarrow \tuple{\Vars, \mathcal{A}, \prod_{i,aybz}}$

  \If{$y \in \Vars$}
	\State $\mathcal{S}' \leftarrow \Call{LetterBoundedControlSet}{G_{i,aybz}, y, a_{j_1}^* \cdots a_{j_s}^*, k}$
  \label{line:rec-y}
  \Else~$\mathcal{S}' \leftarrow \emptyset$  
  \Comment{$y = \varepsilon$ in this case}
  
  \EndIf

  \If{$z \in \Vars$}
	\State $\mathcal{S}'' \leftarrow \Call{LetterBoundedControlSet}{G_{i,aybz}, z, a_{j_1}^* \cdots a_{j_s}^*, k}$
  \label{line:rec-z}
  \Else~$\mathcal{S}'' \leftarrow \emptyset$  
  \Comment{$z = \varepsilon$ in this case}
  \EndIf

	\State \(\mathcal{S}_{\patt} \leftarrow \mathcal{S} \cup 
  \bigcup_{\Gamma \in \mathcal{S}' \cup \mathcal{S}''} \Gamma^\sharp \cdot (X_i, a\,y\,b\,z)^* \cdot \Gamma\)
  \label{line:concatpattg}
  
  \EndIf \EndFor \label{line:for-end}

	\State {\bf return} $\mathcal{S}_{\patt}$
\EndFunction
\end{algorithmic}
\begin{algorithmic}[1]
  \Function{minimizeExpression}{$G,X,a_{i_1}^* \ldots a_{i_d}^*$}
  \State $\mathtt{expr} \leftarrow \varepsilon$
  \For{$\ell = 1, \ldots, d$}
  \If{$L_{X}(G) \cap (\mathcal{A}^* \cdot a_{i_\ell} \cdot
    \mathcal{A}^*) \neq \emptyset$}
		\State $\mathtt{expr} \leftarrow \mathtt{expr} \cdot a_{i_{\ell}}^*$
  \EndIf 
  \EndFor
  \State {\bf return} $\mathtt{expr}$
  \EndFunction
\end{algorithmic}
\begin{algorithmic}[1]
	\Function{partitionNonterminals}{$G,a_{j_1}^*\, a_{j_s}^*$}
  \State {\bf match} $G$ {\bf with} $\tuple{\Vars, \mathcal{A}, \prod}$ 
  \State $\mathtt{vars} \leftarrow \emptyset$
  \For{$Y \in \Vars$}
	\If{$L_Y(G) {\cap} a_{j_1} \, \mathcal{A}^* {\neq} \emptyset
	\land L_Y(G) {\cap} \mathcal{A}^* \, a_{j_s} {\neq} \emptyset$}
  \State $\mathtt{vars} \leftarrow \mathtt{vars} \cup \set{Y}$
  \EndIf
  \EndFor
  \State {\bf return} $(\mathtt{vars},\Vars\setminus\mathtt{vars})$
  \EndFunction
\end{algorithmic}
\caption{Control Sets for Letter-Bounded Grammars}\label{alg:bounded-control-set}
}
\end{algorithm}

\begin{theorem}\label{thm:letter-bounded-control-set}
Given a grammar \(G = \tuple{\Vars, \mathcal{A}, \prod}\),
and \(X \in \Vars\), such that \(L_X(G) \subseteq \patt\),
where \(\patt\) is the minimal strict \(d\)-letter bounded expression
for \(L_X(G)\), for each \(k > 0\), there exists a finite set 
of bounded expressions \(\mathcal{S}_\patt\) over \(\prod\) such that 
\(L_X^{(k)}(G) \subseteq \hat{L}_X(\bigcup\mathcal{S}_\patt \cap \Gamma_X^{\df{k+1}}, G)\).
Moreover, \(\mathcal{S}_\patt\) can be constructed in
time \(\len{G}^{\mathcal{O}(k)+d}\) and each
$\Gamma\in\mathcal{S}_\patt$ can be constructed in
time \(\len{G}^{\mathcal{O}(k)}\).
\end{theorem}
The next lemma shows that the worst-case exponential blowup in the value \(k\)
is unavoidable. 

\begin{lemma}\label{lem:optimality}
For every \(k>0\) there exists a
grammar \(G=\tuple{\Vars,\Sigma,\prod}\) and
\(X\in\Vars\) such that \(\len{G}=\mathcal{O}(k)\) and
every bounded expression \(\Gamma\), such that
\(L_X(G)=\hat{L}_X(\Gamma \cap \Gamma_X^{\df{k+1}},G)\) 
has length \(\len{\Gamma} \geq 2^{k-1}\).
\end{lemma}

\section{Related Work} 
The programs we have studied feature unbounded control (the
call stack) and unbounded data (the integer variables). The
decidability and complexity of the reachability problem for such
programs pose challenging research questions. A long standing and still
open one is the decidability of the reachability problem for programs
where variables behave like Petri net counters and control paths are
taken in a context-free language.  A lower bound
exists \cite{Lazic__2012} but decidability remains
open. Atig and Ganty \cite{AG11} showed decidability when the
context-free language is of bounded index.  The complexity of
reachability was settled for branching VASS by Lazic and
Schmitz \cite{Lazi__2014}. When variables
updates/guards are given by gap-order constraints, reachability
is decidable \cite{Abdulla_2013,Revesz93}. It is in
PSPACE when the set of control paths is regular \cite{Bozzelli_2014}.  
More general updates and guard (like
octagons) immediately leads to undecidability. This explains the
restriction to bounded control sets. Demri \textit{et
al.} \cite{Demri_2012} studied the case of updates/guards of the
form \(\sum_{i=1}^n a_i \cdot x_i + b
\leq 0 \land \vec{x}' = \vec{x} + c\). They show that LTL is \textsc{Np}-complete  on 
for bounded regular control sets, hence reachability is in \textsc{Np}.
Godoy and Tiwari \cite{GT09} studied the invariant checking problem for a class of procedural programs where all executions
conform to a bounded expression, among other restrictions.

\bibliographystyle{abbrv}

\clearpage
\appendix

\section*{Appendix}

The appendix is divided in two parts. Appendix~\ref{sec:easy} contains
easy results about context-free languages and have been included for the sake of being
self-contained. They are variations of classical constructions so as to
take into account index and depth-first
policy. To keep proofs concise, we assume that the grammars are in \(2\)-normal form (\(2\)NF for short).
A grammar is in \(2\)NF if all its productions \( (X,w)\) satisfy \(\len{w}\leq 2\).
Any grammar \(G\) can be converted into an equivalent \(2\)NF grammar \(H\), such that \(\len{H} = O(\len{G})\),
in time \(O(\len{G}^2)\) \cite{LL10}. Note that \(2\)NF is a special case of the general form we assumed
where each production \( (X,w) \) is such that \(w\) contains at most \(2\) terminals and \(2\) nonterminals.
Appendix~\ref{sec:hard} contains the rest of the proofs about
the combinatorial properties of derivations.

\section{From Bounded to Letter-bounded Languages}\label{sec:easy}

It is well-known that the intersection between a context-free and a regular
language is context-free. Below we define the grammar that generates the
intersection between the language of a given grammar $G = \tuple{ \Vars,
\Sigma, \prod }$ and a regular language given by a bounded expression
$\pat=w_1^* \dots w_d^*$ over \(\Sigma\) where \(\ell_i\) denotes the
length of each \(w_i\). Let $G^{\pat} = \tuple{\Vars^{\pat},\Sigma,\prod^{\pat}}$ be the
grammar generating the regular language of $\pat$, where:
\[\begin{array}{rcl}
\Vars^{\pat} & = &  \set{\textsc{q}^{(s)}_{r} \mid 1\leq s\leq d \, \land \,  1\leq r\leq \ell_s} \\[0.3cm]
\prod^{\pat} & = & \set{\textsc{q}^{(s)}_{i}\rightarrow (w_s)_{i}\, \textsc{q}^{(s)}_{i+1}\mid 1 \leq  s \leq  d \, \land \, 1 \leq  i < \ell_s} \; \cup\\[0.2cm]
& & \set{\textsc{q}^{(s)}_{\ell_s}\rightarrow (w_s)_{\ell_s}\, \textsc{q}^{(s')}_{1}\mid 1 \leq  s \leq  s' \leq  d} \; \cup\\[0.2cm]
& & \set{\textsc{q}_1^{(s)}\rightarrow \varepsilon \mid 1 \leq  s \leq  d}\enspace .
\end{array}\]
It is routine to check that \(\{w\mid \textsc{q}_1^{(i)}\Rightarrow^*
w \text{ for some } 1\leq i\leq d \}=\pat\). Moreover, notice that
the number of nonterminals in $G^\pat$ equals the size of $\pat$,
i.e.\ $\card{\Vars^\pat} = \len{\pat}$.

\begin{remark}\label{rem:letter-bounded-complement}
Note that when \(\pat\) is letter-bounded (\(\pat = a_1^* \ldots a_d^*\)), the grammar \(G^{\pat}_1 =
(\Vars^{\pat}_1,\Sigma,\prod^{\pat}_1) \) generating is given by:
\[\begin{array}{rcl}
\Vars_1^{\pat} & = & \set{\textsc{q}^{(s)} \mid 1\leq s\leq d} \cup \set{\textsc{q}_{\mathit{sink}}} \\[0.3cm]
\prod_1^{\pat} & = & \set{\textsc{q}^{(s)} \rightarrow a_{s'} \, \textsc{q}^{(s')} \mid 1 \leq  s \leq s' \leq  d } \; \cup\\[0.2cm]
& & \set{\textsc{q}^{(s)} \rightarrow b\, \textsc{q}_{\mathit{sink}} \mid b \in \Sigma \setminus \{a_s, a_{s+1}, \ldots, a_d\} }\; \cup\\[0.2cm]
& & \set{\textsc{q}^{(s)}\rightarrow \varepsilon \mid 1 \leq  s \leq  d}\; \cup\\[0.2cm]
& & \set{\textsc{q}_{\mathit{sink}} \rightarrow b\, \textsc{q}_{\mathit{sink}} \mid b \in \Sigma }
\end{array}\]
is such that \(L_{\textsc{q}^{(1)}}(G^{\pat}_1) = \pat\).
Furthermore, \(G^{\pat}_1\) is complete---all terminals can be
produced from all nonterminals---and it is deterministic when \(\pat\) is
strict. Then a grammar \(\overline{G^{\pat}_1}\), such that
\(L_{\textsc{q}^{(1)}}(\overline{G^{\pat}_1}) = \Sigma^* \setminus
\pat\), can be computed in time \(\mathcal{O}(|G_1^{\pat}|)\), by
replacing each production \( \textsc{q}^{(s)}\rightarrow \varepsilon,\,
1 \leq s \leq d\), with \(\textsc{q}_{\mathit{sink}} \rightarrow
\varepsilon\).
\end{remark}

Given $G^{\pat}$, and a grammar $G=(\Vars,\Sigma,\prod)$ in 2NF and \(X\in\Vars\), our goal is to define a
grammar \(G^{\cap} = \tuple {\Vars^{\cap},\Sigma,\prod^{\cap}}\) that
produces the language $L_X(G) \cap L(\pat)$, for some $X \in
\Vars$. The definition of $G^\cap = \tuple {\Vars^{\cap}, \Sigma,
  \prod^{\cap}}$ follows:
\begin{compactitem}
\item $\Vars^{\cap}=\set{ [\textsc{q}^{(r)}_{s} X \textsc{q}^{(u)}_{v}]
  \mid X \in \Vars \land \textsc{q}^{(r)}_{s}\in\Vars^{\pat}\land
  \textsc{q}^{(u)}_{v} \in\Vars^{\pat}\land r\leq u}$
\item $\prod^{\cap}$ is defined as follows: 
  \begin{compactitem}
  \item for every production $X\rightarrow w\in\prod$ where
    \(w\in\Sigma^*\), \(\prod^{\cap}\) has a production
    \begin{align}
      [\textsc{q}^{(r)}_{s} X \textsc{q}^{(u)}_{v}]&\rightarrow w & \text{if } \textsc{q}^{(r)}_{s}\Rightarrow^* w \, \textsc{q}^{(u)}_{v}\enspace ;\label{eq:terminalsonly}
    \end{align}
  \item for every production $X\rightarrow Y\in\prod$, where $Y \in \Vars$, \(\prod^{\cap}\) has a production
    \begin{align}
      [\textsc{q}^{(r)}_{s} X \textsc{q}^{(u)}_v]&\rightarrow[\textsc{q}^{(r)}_{s} Y \textsc{q}^{(u)}_{v}] \enspace ;\label{eq:1var}
    \end{align}
  \item for every production $X\rightarrow a\, Y\in\prod$, where $a \in \Sigma$ and $Y \in \Vars$, \(\prod^{\cap}\) has a production
    \begin{align}
      [\textsc{q}^{(r)}_{s} X \textsc{q}^{(u)}_v]&\rightarrow a \, [\textsc{q}^{(x)}_{y} Y \textsc{q}^{(u)}_v] & \text{if } \textsc{q}^{(r)}_{s}\rightarrow a\,\textsc{q}^{(x)}_{y} \in \prod^{\pat}\enspace ;\label{eq:gammavar} 
    \end{align}
  \item for every production $X\rightarrow Y\, a\in\prod$, where $Y \in \Vars$ and $a \in \Sigma$, \(\prod^{\cap}\) has a production
    \begin{align}
      [\textsc{q}^{(r)}_{s} X \textsc{q}^{(u)}_v]&\rightarrow [\textsc{q}^{(r)}_{s} Y \textsc{q}^{(x)}_{y}] \, a & \text{if } \textsc{q}^{(x)}_{y}\rightarrow a\,\textsc{q}^{(u)}_{v} \in \prod^{\pat}\enspace ;\label{eq:vargamma}
    \end{align}
  \item for every production $X\rightarrow Y\, Z\in\prod$, \(\prod^{\cap}\) has a production
    \begin{align}
			[\textsc{q}^{(r)}_{s} X \textsc{q}^{(u)}_{v}] &\rightarrow [\textsc{q}^{(r)}_{s} Y \textsc{q}^{(x)}_y]\,[\textsc{q}^{(x)}_y Z \textsc{q}^{(u)}_v] \enspace ;\label{eq:2var}
    \end{align}
  \item \(\prod^{\cap}\) has no other production.
  \end{compactitem}
\end{compactitem}

Let $\zeta\colon \Vars^\cap \rightarrow \Vars$ be the function that
``strips'' every nonterminal $[\textsc{q}_r^{(s)} X \textsc{q}_{v}^{(u)}]
\in \Vars^\cap$ of the nonterminals from $\Vars^\pat$,
i.e.\ $\zeta([\textsc{q}_r^{(s)} X \textsc{q}_{v}^{(u)}]) = X$. In the
following, we abuse notation and extend the $\zeta$ function to
symbols from $\Sigma \cup \Vars^\cap$, by defining $\zeta(a)=a$, for
each $a \in \Sigma$, and further to words $w\in
(\Sigma \cup \Vars^\cap)^*$ as $\zeta(w) = \zeta( (w)_1
) \cdots \zeta( (w)_{\len{w}} )$. Finally, for a production $p =
(X,w) \in \prod^\cap$, we define $\zeta(p) = (\zeta(X), \zeta(w))$,
and for a control word $\gamma \in (\prod^\cap)^*$, we write
$\zeta(\gamma)$ for $\zeta( (\gamma)_1 ) \cdots \zeta(
(\gamma)_{\len{\gamma}} )$.

\begin{lemma}\label{lem:intersection}
Given a grammar $G = \tuple{\Vars,\Sigma,\prod}$ and a grammar $G^{\pat} =
\tuple{\Vars^{\pat},\Sigma,\prod^{\pat}}$ generating \(\pat\), for every
\(X\in\Vars\), \(\textsc{q}_{s}^{(r)},\textsc{q}_{v}^{(u)} \in \Vars^{\pat}\),
\(w\in\Sigma^*\), and every $k > 0$, we have:
  \begin{compactenum}[\upshape(\itshape i\upshape)] 

\item\label{item:inter-only-if} for every
  $\gamma \in (\prod^\cap)^*$, $[\textsc{q}_{s}^{(r)}
  X \textsc{q}_{v}^{(u)}] \xRightarrow[\df{k}]{\gamma}^* w$ only if
  $X \xRightarrow[\df{k}]{\zeta(\gamma)} w$ and
  $\textsc{q}_{s}^{(r)}\Rightarrow_{G^{\pat}}^*
  w\, \textsc{q}_{v}^{(u)}$

 \item\label{item:inter-if} for every $\delta \in \prod^*$,
    $X \xRightarrow[\df{k}]{\delta} w$ and
    $\textsc{q}_{s}^{(r)}\xRightarrow{}_{G^{\pat}}^*
    w\, \textsc{q}_{v}^{(u)}$ only if $[\textsc{q}_{s}^{(r)}
    X \textsc{q}_{v}^{(u)}] \xRightarrow[\df{k}]{\gamma}^* w$, for
    some $\gamma \in \zeta^{-1}(\delta)$.  \end{compactenum}
    Consequently, we have $\bigcup_{1\leq s\leq x\leq d}
    L_{[\textsc{q}^{(s)}_1
    X \textsc{q}^{(x)}_{1}]}(G^{\cap})=L_X(G) \cap \pat \enspace$.
\end{lemma}
\begin{proof}
(\ref{item:inter-only-if}) By induction on $\len{\gamma} > 0$. For the base
  case $\len{\gamma}=1$---$\gamma$ is the production
  $([\textsc{q}_{s}^{(r)} X \textsc{q}_{v}^{(u)}] \arrow{}{} w) \in
  \Delta^\cap$ with \(w\in\Sigma^*\)---by case (\ref{eq:terminalsonly}) of the definition
  of $\Delta^\cap$, we have
  $\textsc{q}_{s}^{(r)}\Rightarrow_{G^{\pat}}^* w\, \textsc{q}_{v}^{(u)}$
  and there exists a production $X \arrow{}{} w \in \Delta$. Since,
  moreover, $\zeta([\textsc{q}_{s}^{(r)} X \textsc{q}_{v}^{(u)}]
  \arrow{}{} w) = (X \arrow{}{} w)$, we have that $X
  \xRightarrow[\df{1}]{\zeta(\gamma)} w$ in $G$.

  For the induction step $\len{\gamma} > 1$, we have $\gamma =
  ([\textsc{q}_{s}^{(r)} X \textsc{q}_{v}^{(u)}] \arrow{}{} \tau) \cdot
  \gamma'$, for some production $[\textsc{q}_{s}^{(r)} X
    \textsc{q}_{v}^{(u)}] \arrow{}{} \tau \in \Delta^\cap$, and a word
  $\tau \in (\Sigma \cup \Vars^\cap)^*$ of length $\len{\tau} \leq
  2$. We distinguish four cases, based on the structure of $\tau$:
  \begin{compactenum}
  \item if $\tau = [\textsc{q}_{s}^{(r)} Y \textsc{q}_{v}^{(u)}]$ then $\tau
    \xRightarrow[\df{k}]{\gamma'} w$ is a derivation of $G^\cap$. By
    the induction hypothesis, we obtain that
    $\textsc{q}_{s}^{(r)}\Rightarrow_{G^{\pat}}^* w\,\textsc{q}_{v}^{(u)}$
    and $Y \xRightarrow[\df{k}]{\zeta(\gamma')} w$ is a derivation of
    $G$. But \(X \arrow{}{} Y\in\prod\)---case
    (\ref{eq:1var}) of the definition of $\Delta^\cap$---hence 
    $\zeta(\gamma) = (X \arrow{}{} Y) \cdot \zeta(\gamma')$ and
    $X \xRightarrow[\df{k}]{\zeta(\gamma)} w$ is a derivation of $G$.

  \item if $\tau = a \, [\textsc{q}^{(x)}_{y} Y
    \textsc{q}^{(u)}_v]$ then $w = a \cdot w'$ and $G^\cap$ has
    derivation $[\textsc{q}^{(x)}_{y} Y \textsc{q}^{(u)}_v]
    \xRightarrow[\df{k}]{\gamma'} w'$. By the induction hypothesis, we
    obtain $\textsc{q}^{(x)}_{y} \xRightarrow{}_{G^\pat}^* w'
    \textsc{q}^{(u)}_{v}$ and $G$ has a derivation $Y
    \xRightarrow[\df{k}]{\zeta(\gamma')} w'$. By the case
    (\ref{eq:gammavar}) of the definition of $\Delta^\cap$, we have
    $\textsc{q}^{(r)}_{s}\rightarrow a \, \textsc{q}^{(x)}_{y}
    \in \prod^{\pat}$ and $\zeta([\textsc{q}_{s}^{(r)} X
      \textsc{q}_{v}^{(u)}] \arrow{}{} \tau) = (X \arrow{}{} a Y)
    \in \Delta$. Thus $\textsc{q}^{(r)}_{s} \xRightarrow{}_{G^\pat}^*
    w\, \textsc{q}^{(u)}_{v}$ and $X \xRightarrow[\df{k}]{\zeta(\gamma)}
    w$, since $\zeta(\gamma) = (X \arrow{}{} a Y) \cdot
    \zeta(\gamma')$.

    \item the case $\tau = [\textsc{q}^{(r)}_{s}
      Y \textsc{q}^{(x)}_{y}] \, a$ is symmetric, using the case
      (\ref{eq:vargamma}) of the definition of $\Delta^\cap$.  \item
      if $\tau = [\textsc{q}^{(r)}_{s} Y \textsc{q}^{(x)}_y]\,
      [\textsc{q}^{(x)}_y Z \textsc{q}^{(u)}_v]$ then, by
      Lemma~\ref{lem:leibniz}, there exist words $w_1,
      w_2 \in \Sigma^*$ such that $w = w_1w_2$ and either one of the
      following applies:
			\begin{compactenum} 
      \item $[\textsc{q}^{(r)}_{s}
      Y \textsc{q}^{(x)}_y] \xRightarrow[\df{k-1}]{\gamma_1} w_1$, 
      $[\textsc{q}^{(x)}_y Z \textsc{q}^{(u)}_v] \xRightarrow[\df{k}]{\gamma_2} w_2$ and
      $\gamma'=\gamma_1\,\gamma_2$, or

      \item $[\textsc{q}^{(r)}_{s} Y \textsc{q}^{(x)}_y] \xRightarrow[\df{k}]{\gamma_1} w_1$, 
      $[\textsc{q}^{(x)}_y Z \textsc{q}^{(u)}_v] \xRightarrow[\df{k-1}]{\gamma_2} w_2$
      and $\gamma' = \gamma_2\, \gamma_1$.
      \end{compactenum} 
			We consider the first case only, the second
      being symmetric.  Since $\len{\gamma_1} < \len{\gamma}$ and
      $\len{\gamma_2} < \len{\gamma}$, we apply the induction
      hypothesis and find out that
      $\textsc{q}^{(r)}_{s} \xRightarrow{}_{G^\pat}^* w_1\, \textsc{q}^{(x)}_y$,
      $\textsc{q}^{(x)}_y \xRightarrow{}_{G^\pat}^* w_2\, \textsc{q}^{(u)}_v$, and $G$ has derivations
      $Y \xRightarrow[\df{k-1}]{\zeta(\gamma_1)} w_1$ and
      $Z \xRightarrow[\df{k}]{\zeta(\gamma_2)} w_2$. Then
			$\textsc{q}^{(r)}_{s} \xRightarrow{}_{G^\pat}^* w_1\, w_2 \,\textsc{q}^{(u)}_v$ where \(w_1\, w_2 = w \). By
      case (\ref{eq:2var}) of the definition of $\Delta^\cap$,
      $\Delta$ has a production $(X \arrow{}{} YZ)
      = \zeta([\textsc{q}^{(r)}_{s}
      X \textsc{q}^{(u)}_v] \arrow{}{} \tau)$. Since $\gamma'
      = \gamma_1 \, \gamma_2$, then $\zeta(\gamma) = (X \arrow{}{}
      YZ) \cdot \zeta(\gamma_1) \, \zeta(\gamma_2)$, and $G$ has a
      $k$-index depth-first derivation
      $X \xRightarrow[\df{k}]{\zeta(\gamma)} w$. \end{compactenum}
 
\noindent(\ref{item:inter-if}) By induction on $\len{\delta} > 0$. For the
base case $\len{\delta} = 1$, we have $\delta = (X \arrow{}{} w) \in
\Delta$. By the case (\ref{eq:terminalsonly}) from the definition of
$\Delta^\cap$, $G^\cap$ has a rule $[\textsc{q}^{(s)}_r X
  \textsc{q}_{v}^{(u)}] \arrow{}{} w$ and, since, moreover,
$\zeta([\textsc{q}^{(s)}_r X \textsc{q}_{v}^{(u)}] \arrow{}{} w) =
\delta$, we have $\gamma = ([\textsc{q}^{(s)}_r X
  \textsc{q}_{v}^{(u)}] \arrow{}{} w)$.

For the induction step $\len{\delta} > 1$, we have $\delta = (X
\arrow{}{} \tau) \cdot \delta'$. We distinguish four cases, based on
the structure of $\tau$:
\begin{compactenum}
\item if $\tau = Y$, for some $Y \in \Vars$, by the induction
  hypothesis, $G^\cap$ has a derivation $[\textsc{q}^{(s)}_r Y
    \textsc{q}_{v}^{(u)}] \xRightarrow[\df{k}]{\gamma'} w$, for some
  $\gamma' \in \zeta^{-1}(\delta')$. Since
  $\textsc{q}_{s}^{(r)}\Rightarrow_{G^{\pat}}^* w\,\textsc{q}_{v}^{(u)}$  ---by case (\ref{eq:1var}) of the definition of $\Delta^\cap$---  $G^\cap$ has a production $p = ([\textsc{q}^{(s)}_r X
    \textsc{q}_{v}^{(u)}] \arrow{}{} [\textsc{q}^{(s)}_r Y
    \textsc{q}_{v}^{(u)}])$. We define $\gamma = p \cdot \gamma'$. It
  is immediate to check that $\zeta(\gamma)=\delta$.

\item if $\tau = a\, Y$, for some $a \in \Sigma$ and $Y \in
  \Vars$, then $w = a \cdot w'$. Hence $\textsc{q}_{s}^{(r)}
  \xRightarrow{}_{G^{\pat}} a\, \textsc{q}_{y}^{(x)}$,
  $\textsc{q}_{y}^{(x)} \xRightarrow{}_{G^{\pat}}^* w' \,\textsc{q}_{v}^{(u)}$ and $G$ has a derivation $Y
  \xRightarrow[\df{k}]{\delta'} w'$. By the induction hypothesis,
  $G^\cap$ has a derivation $[\textsc{q}_{y}^{(x)} Y
    \textsc{q}_{v}^{(u)}] \xRightarrow[\df{k}]{\gamma'} w'$, for some
		$\gamma' \in \zeta^{-1}(\delta')$. By the case (\ref{eq:gammavar}) of the
  definition of $\Delta^\cap$, there exists a production $p =
  ([\textsc{q}_{r}^{(s)} X \textsc{q}_{v}^{(u)}] \arrow{}{} a Y)
  \in \Delta^\cap$. We define $\gamma = p \cdot \gamma'$. It is
  immediate to check that $\zeta(\gamma) = \delta$, hence
  $[\textsc{q}_{r}^{(s)} X \textsc{q}_{v}^{(u)}]
  \xRightarrow[\df{k}]{\gamma} w$. 

\item the case $\tau = Y\, a$, for some $Y \in \Vars$ and $a \in
  \Sigma$, is symmetrical.

\item if $\tau = Y\, Z$, for some $Y,Z \in \Vars$, then, by Lemma~\ref{lem:leibniz}, 
      there exist words $w_1, w_2 \in \Sigma^*$ such that $w = w_1\, w_2$
      and either one of the following cases
      applies: 

      \begin{compactenum} 

      \item $Y \xRightarrow[\df{k-1}]{\delta_1} w_1$,
      $Z \xRightarrow[\df{k}]{\delta_2} w_2$ and
      $\delta'=\delta_1\,\delta_2$, or

      \item $Y \xRightarrow[\df{k}]{\delta_1} w_1$,
      $Z \xRightarrow[\df{k-1}]{\delta_2} w_2$ and
      $\delta'=\delta_2\,\delta_1$. 

      \end{compactenum} 

      Moreover, we have
      $\textsc{q}_{s}^{(r)} \xRightarrow{}_{G^{\pat}}^*
      w_1\, \textsc{q}_{y}^{(x)}$ and
      $\textsc{q}_y^{(x)} \xRightarrow{}_{G^{\pat}}^*
      w_2\, \textsc{q}_{v}^{(u)}$, for some
      $\textsc{q}_{y}^{(x)} \in \Vars^\pat$. We consider the first
      case only, the second being symmetric. Since $\len{\delta_1}
      < \len{\delta}$ and $\len{\delta_2} < \len{\delta}$ we apply the
      induction hypothesis and find two control words
      $\gamma_1 \in \zeta^{-1}(\delta_1)$ and
      $\gamma_2 \in \zeta^{-1}(\delta_2)$ such that $G^\cap$ has
      derivations $[\textsc{q}_{s}^{(r)}
      Y \textsc{q}_{y}^{(x)}] \xRightarrow[\df{k-1}]{\gamma_1} w_1$
      and $[\textsc{q}_{y}^{(x)}
      Z \textsc{q}_{v}^{(u)}] \xRightarrow[\df{k}]{\gamma_2} w_2$. By
      case (\ref{eq:2var}) of the definition of $\Delta^\cap$,
      $G^\cap$ has a production $p = ([\textsc{q}_{s}^{(r)}
      X \textsc{q}_{v}^{(u)}] \arrow{}{} [\textsc{q}_{s}^{(r)}
      Y \textsc{q}_{y}^{(x)}] [\textsc{q}_{y}^{(x)}
      Z \textsc{q}_{v}^{(u)}])$. Since $\delta' = \delta_1\,\delta_2$, we
      define $\gamma = p \, \gamma_1 \, \gamma_2$. It is immediate to
      check that $\zeta(\gamma) = \delta$ and $[\textsc{q}_{s}^{(r)}
      X \textsc{q}_{v}^{(u)}] \xRightarrow[\df{k}]{\gamma} w$.\qed
\end{compactenum}
\end{proof}

In the rest of this section, for a given bounded expression $\pat =
w_1^* \ldots w_d^*$ over \(\Sigma\), we associate the
strict \(d\)-letter-bounded expression $\patt = a_1^* \ldots a_d^*$
over an alphabet \(\mathcal{A}\), disjoint from $\Sigma$, i.e.\
$\mathcal{A} \cap \Sigma = \emptyset$, and a homomorphism
$h \colon \mathcal{A} \rightarrow \Sigma^*$ mapping as follows:
$a_i \mapsto w_i$, for all $1 \leq i \leq d$. The next step is to
define a grammar $G^\bowtie
= \tuple{\Vars^\bowtie, \mathcal{A}, \prod^\bowtie}$, such that \(\Vars^{\bowtie} = \Vars^{\cap}\) and, for all
$X \in \Vars, 1 \leq s \leq x \leq d$: 
\[h^{-1}( L_{[\textsc{q}^{(s)}_{1} X \textsc{q}^{(x)}_{1}]}(G^{\cap}) ) \cap \patt = 
L_{[\textsc{q}^{(s)}_{1} X \textsc{q}^{(x)}_{1}]}(G^\bowtie)\enspace .\] 
The grammar $G^\bowtie$ is defined from $G^\cap$, by the following
modification of the productions from $\prod^\cap$, defined by a
function \(\iota\colon \prod^{\cap}\mapsto\prod^{\bowtie}\):

\begin{compactitem}
\item \( \iota([\textsc{q}^{(r)}_{s} X \textsc{q}^{(u)}_{v}] \rightarrow w) = [\textsc{q}^{(r)}_{s} X \textsc{q}^{(u)}_{v}] \rightarrow z\) where
\begin{compactenum}
	\item if \(\len{w}=0\) then \(z=\varepsilon\). 
	\item if \(\len{w}=1\) then we have \(\textsc{q}^{(r)}_{s}\Rightarrow_{G^\pat} w\, \textsc{q}^{(u)}_{v}\) and we let \(z=a_r\) if \(v=1\) else \(z=\varepsilon\).  
	\item if \(\len{w}=2\) then we have \(\textsc{q}^{(r)}_{s}\Rightarrow_{G^\pat} (w)_{1}\, \textsc{q}^{(y)}_{x} \Rightarrow_{G^\pat} (w)_{1} (w)_{2}\,  \textsc{q}^{(u)}_{v} \) for some \(x,y\).
		Define the word \(z = z'\cdot z''\) of length at most \(2\) such that \( z'=a_r\) if \(x=1\); else \(z'=\varepsilon\) and \(z''=a_y\) if \(v=1\) else \(z''=\varepsilon\).
\end{compactenum}

\item \( \iota([\textsc{q}^{(r)}_{s} X \textsc{q}^{(u)}_v] \rightarrow b \,
	[\textsc{q}^{(x)}_{y} Y \textsc{q}^{(u)}_v]) = [\textsc{q}^{(r)}_{s} X
	\textsc{q}^{(u)}_v] \rightarrow c \, [\textsc{q}^{(x)}_{y} Y
	\textsc{q}^{(u)}_v] \) where \(c=a_r\) if \(y=1\); else \(c=\varepsilon\).

\item \( \iota([\textsc{q}^{(r)}_{s} X \textsc{q}^{(u)}_v] \rightarrow
	[\textsc{q}^{(r)}_{s} Y \textsc{q}^{(x)}_y]\, b) = [\textsc{q}^{(r)}_{s} X
	\textsc{q}^{(u)}_v] \rightarrow [\textsc{q}^{(r)}_{s} Y
	\textsc{q}^{(x)}_y]\, c\) where \(c=a_x\) if \(v=1\); else \(c=\varepsilon\).
\item \(\iota(p)=p\) otherwise.
\end{compactitem}
Let $\prod^\bowtie = \{\iota(p) \mid p \in \prod^\cap\}$.  In
addition, for every control word \(\gamma\in (\Delta^\cap)^*\) of
length \(n\), let $\iota(\gamma) = \iota( (\gamma)_1 ) \cdots \iota(
(\gamma)_n ) \in \prod^\bowtie$.  A consequence of the following
proposition is that the inverse relation $\iota^{-1} \subseteq
\prod^\bowtie \times \prod^\cap$ is a total function. 

\begin{proposition}\label{prop:singleton}
For each production $p \in \Delta^\bowtie$, the set $\iota^{-1}(p)$ is a singleton. 
\end{proposition}
\begin{proof}
By case split, based on the type of the production $p \in
\Delta^\bowtie$. Since $G^\bowtie$ is in $2$NF we have: 
\begin{compactitem}
\item if $p = ([\textsc{q}^{(r)}_{s} X \textsc{q}_{v}^{(u)}] \arrow{}{}
  a)$ then $\iota^{-1}(p) = \{[\textsc{q}^{(r)}_{s} X
      \textsc{q}_{v}^{(u)}] \arrow{}{} w\}$, where $\textsc{q}^{(r)}_{s}
  \xRightarrow{}^*_{G^\pat} w\, \textsc{q}_{v}^{(u)}$ is the shortest
	step sequence of $G^\pat$ between \(\textsc{q}^{(r)}_{s}\) and \(\textsc{q}^{(u)}_{v}\) which is unique by \(G^{\pat}\) and produces $w \in \Sigma^*$.

\item if $p = ([\textsc{q}^{(r)}_{s} X \textsc{q}_{v}^{(u)}] \arrow{}{}
  [\textsc{q}_{y}^{(x)} Y \textsc{q}_{t}^{(z)}])$, then either one of
  the cases below must hold:
  \begin{compactenum}[\upshape(\itshape i\upshape)]
	\item \(\textsc{q}_{u}^{(v)}=\textsc{q}_{z}^{(t)}\) and $\textsc{q}^{(r)}_{s} \xRightarrow{}_{G^\pat}
    b\, \textsc{q}_{y}^{(x)}$, for some \(y\neq 1\). In this case $b$ is
    uniquely determined by $\textsc{q}^{(r)}_{s}$ and 
    $\textsc{q}_{y}^{(x)}$, thus we get $\iota^{-1}(p) = \{[\textsc{q}^{(r)}_{s} X
      \textsc{q}_{v}^{(u)}] \arrow{}{} b~[\textsc{q}_{y}^{(x)} Y
      \textsc{q}_{t}^{(z)}]\}$. 
		\item \(\textsc{q}_{s}^{(r)}=\textsc{q}_{y}^{(x)}\) and $\textsc{q}_{t}^{(z)}
      \xRightarrow{}_{G^\pat} b \,\textsc{q}_{v}^{(u)}$, for some \(t\neq\ell_z\). In this case we get, symmetrically, $\iota^{-1}(p)
      = \{[\textsc{q}^{(r)}_{s} X \textsc{q}_{v}^{(u)}] \arrow{}{}
      [\textsc{q}_{y}^{(x)} Y \textsc{q}_{t}^{(z)}]~b\}$.
		\item \(\textsc{q}_{u}^{(v)}=\textsc{q}_{z}^{(t)}\) and \(\textsc{q}_{s}^{(r)}=\textsc{q}_{y}^{(x)}\). Then $\iota^{-1}(p) = \{p\}$. 
  \end{compactenum}

\item if $p = ([\textsc{q}^{(r)}_{s} X \textsc{q}_{v}^{(u)}] \arrow{}{}
  a_r\, [\textsc{q}_{y}^{(x)} Y \textsc{q}_{v}^{(u)}])$  for some $a_r \in \mathcal{A}$, hence \(y=1\) (respectively,
  $[\textsc{q}^{(r)}_{s} X \textsc{q}_{v}^{(u)}] \arrow{}{}
  [\textsc{q}^{(r)}_{s} Y \textsc{q}_{y}^{(x)}]\, a_r$ hence \(v=1\)) and then the only possibility is $\iota^{-1}(p) =
	\{[\textsc{q}^{(r)}_{s} X \textsc{q}_{v}^{(u)}] \arrow{}{} (w_r)_{\ell_{r}}\,
	[\textsc{q}_{y}^{(x)} Y \textsc{q}_{v}^{(u)}]\}$ (respectively,
	$[\textsc{q}^{(r)}_{s} X \textsc{q}_{v}^{(u)}] \arrow{}{}
	[\textsc{q}^{(r)}_{s} Y \textsc{q}_{y}^{(x)}]\, (w_r)_{\ell_{r}}$).

\item if $p = ([\textsc{q}^{(r)}_{s} X \textsc{q}^{(x)}_{y}]
	\rightarrow [\textsc{q}^{(r)}_{s} Y \textsc{q}^{(u)}_v]\,
              [\textsc{q}^{(u)}_v Z \textsc{q}^{(x)}_y])$ then
              $\iota^{-1}(p) = \set{p}$.\qed
\end{compactitem}
\end{proof}

\begin{lemma}\label{lem:interface}
Given a grammar \(G = \tuple{ \Vars, \Sigma, \prod }\) and a bounded
expression \(\pat = w_1^* \ldots w_d^*\) over \(\Sigma\), for
every \(X \in \Vars\), every \(1 \leq s \leq x \leq d\) and every \(k
> 0\), the following hold:
\begin{compactenum}
\item\label{item:g-bowtie} \(L_{[\textsc{q}^{(s)}_1
    X \textsc{q}^{(x)}_{1}]}^{(k)}(G^{\bowtie}) = h^{-1}(L_{[\textsc{q}^{(s)}_1
    X \textsc{q}^{(x)}_{1}]}^{(k)}(G^\cap)) \cap \patt\), 

\item\label{item:xi} for each control set \(\Gamma \subseteq \left(\prod^\bowtie\right)^*\), such that 
\(L_{{[\textsc{q}^{(s)}_1 X \textsc{q}^{(x)}_{1}]}}^{(k)}(G^\bowtie) \subseteq \hat{L}_{[\textsc{q}^{(s)}_1
    X \textsc{q}^{(x)}_{1}]}(\Gamma, G^\bowtie)\), we
    have \(L_{{[\textsc{q}^{(s)}_1
    X \textsc{q}^{(x)}_{1}]}}^{(k)}(G^\cap) \subseteq \hat{L}_{{[\textsc{q}^{(s)}_1
    X \textsc{q}^{(x)}_{1}]}}(\iota^{-1}(\Gamma), G^\cap)\),

\item\label{item:bowtie-xi-complexity} 
\(G^{\bowtie}\) is computable in time \( \mathcal{O}\bigl( \len{\pat}^3 \cdot \len{G} \bigr) \).
\end{compactenum}
\end{lemma}
\begin{proof}
We start by proving the following facts:

\begin{fact}\label{fact:bowtie-letter-bounded}
  For all $X \in \Vars$ and $1 \leq s \leq x \leq d$, we have
  $L_{[\textsc{q}^{(s)}_1 X \textsc{q}^{(x)}_{1}]}(G^\bowtie) \subseteq \patt$.
\end{fact}
\begin{proof}
Let $\tilde{w} \in L_{[\textsc{q}^{(s)}_1 X \textsc{q}^{(x)}_{1}]}(G^\bowtie)$.
We have \( [\textsc{q}^{(s)}_1 X \textsc{q}^{(x)}_{1}] \xRightarrow{\gamma} \tilde{w}\) is a derivation of \(G^{\bowtie}\) for some
control word \(\gamma\) over \(\prod^{\bowtie}\).
By contradiction, assume \(\tilde{w}\notin \patt\), that is there exist \(p,p'\) such that \(p < p'\) and \( (\tilde{w})_p = a_j \) and \( (\tilde{w})_{p'} = a_i \) with \(i < j\).
The definition of \(\iota\) shows that there exists \(w\in L_{[\textsc{q}^{(s)}_1 X \textsc{q}^{(x)}_{1}]}(G^{\cap})\) such that
\( [\textsc{q}^{(s)}_1 X \textsc{q}^{(x)}_{1}] \xRightarrow{\iota^{-1}(\gamma)} w\) in \(G^{\cap}\), hence that \(w\in \pat\) since \(L_{[\textsc{q}^{(s)}_1 X \textsc{q}^{(x)}_{1}]}(G^{\cap})\subseteq\pat\), and finally that \( \textsc{q}^{(s)}_1 \Rightarrow^*_{G^{\pat}} w\, \textsc{q}^{(x)}_{1} \).
Now, the mapping \(\iota\) is defined such that a production in its image produces a \(a_r\) when, in the underlying \(G^{\pat}\), either
control moves forward from \( \textsc{q}^{(r)}_s \) to \( \textsc{q}^{(u)}_1 \), e.g. 
\( [\textsc{q}^{(r)}_{s} X \textsc{q}^{(x)}_y] \rightarrow a_r \, [\textsc{q}^{(u)}_{1} Y \textsc{q}^{(x)}_y]\)
or control moves backward form \( \textsc{q}^{(u)}_1 \) to \( \textsc{q}^{(r)}_s \), e.g. 
\( [\textsc{q}^{(x)}_{y} X \textsc{q}^{(u)}_1] \rightarrow [\textsc{q}^{(x)}_{y} Y \textsc{q}^{(r)}_s]\, a_r\).
Therefore, by the previous assumption on \(\tilde{w}\) where \(a_j\) occurs before \(a_i\), we have that a production of \(\textsc{q}^{(j)}_{\ell_j} \rightarrow (w_j)_{\ell_j}\, \textsc{q}^{(u)}_1\) for some \(u\geq j\) and then a production of \(\textsc{q}^{(i)}_{\ell_i} \rightarrow (w_i)_{\ell_i}\, \textsc{q}^{(u')}_1\) for some \(u' \geq i\) necessarily occurs in that order in \(\iota^{-1}(\gamma)\). But this is a contradiction because \(j>i\) and the definition
of \(G^{\pat}\) prohibits control to move from \(\textsc{q}^{(j)}_{p_j}\) to \(\textsc{q}^{(i)}_{p_i} \) for any \(p_i, p_j\).
\qed
\end{proof}

\begin{fact}\label{fact:bowtie-cap}
  For all $X \in \Vars$, $1 \leq s \leq x \leq d$, $\gamma \in
  (\prod^\cap)^*$, $k > 0$  and $i_1, \ldots, i_d \in \nats$: 
  \[[\textsc{q}^{(s)}_1 X \textsc{q}^{(x)}_{1}]
  \xRightarrow[(k)]{\gamma} w_1^{i_1} \ldots w_d^{i_d} ~\mbox{in $G^\cap$}
	\text{ if and only if } [\textsc{q}^{(s)}_1 X \textsc{q}^{(x)}_{1}]
  \xRightarrow[(k)]{\iota(\gamma)} a_1^{i_1} \ldots a_d^{i_d} ~\mbox{in
  $G^\bowtie$}\enspace .\]
\end{fact}
\begin{proof} By induction on $\len{\gamma} > 0$, and case analysis 
on the right-hand side of $(\gamma)_1$.\qed
\end{proof}

\noindent (\ref{item:g-bowtie}) ``$\subseteq$'' Let $\tilde{w} \in
L_{[\textsc{q}^{(s)}_1 X \textsc{q}^{(x)}_{1}]}^{(k)}(G^\bowtie)$. By
Fact~\ref{fact:bowtie-letter-bounded}, we have that
$\tilde{w} \in \patt$. It remains to show that $\tilde{w} \in
h^{-1}(L_{[\textsc{q}^{(s)}_1 X \textsc{q}^{(x)}_{1}]}(G^{\cap}))$, i.e. that $h(\tilde{w}) \in
L_{[\textsc{q}^{(s)}_1 X \textsc{q}^{(x)}_{1}]}(G^{\cap})$, which follows by Fact~\ref{fact:bowtie-cap}.
``$\supseteq$'' Let $\tilde{w} \in
h^{-1}(L_{[\textsc{q}^{(s)}_1 X \textsc{q}^{(x)}_{1}]}^{(k)}(G^{\cap})) \cap \patt$ be a word, hence \(\tilde{w} = a_1^{i_1}\ldots a_d^{i_d}\) for some \(i_1,\ldots,i_d\in\nats\). Then
$h(\tilde{w}) \in L_{[\textsc{q}^{(s)}_1 X \textsc{q}^{(x)}_{1}]}^{(k)}(G^{\cap})$ by Fact~\ref{fact:bowtie-cap} and we are done.

\vspace*{\baselineskip}\noindent
(\ref{item:xi}) Let $w = w_1^{i_1} \ldots w_d^{i_d} \in
L_{[\textsc{q}^{(s)}_1 X \textsc{q}^{(x)}_{1}]}^{(k)}(G^\cap)$ be a
word.  Then $G^\cap$ has a derivation $[\textsc{q}^{(s)}_1
X \textsc{q}^{(x)}_{1}] \xRightarrow[(k)]{}^* w$. By
Fact~\ref{fact:bowtie-cap}, also $G^\bowtie$ has a derivation
$[\textsc{q}^{(s)}_1 X \textsc{q}^{(x)}_{1}] \xRightarrow[(k)]{}^*
a_1^{i_1} \ldots a_d^{i_d}$. By the hypothesis $L_{[\textsc{q}^{(s)}_1
X \textsc{q}^{(x)}_{1}]}^{(k)}(G^\bowtie) \subseteq \hat{L}_{[\textsc{q}^{(s)}_1
X \textsc{q}^{(x)}_{1}]}(\Gamma, G^\bowtie)$, there exists a control
word $\gamma \in \Gamma$ such that $[\textsc{q}^{(s)}_1
X \textsc{q}^{(x)}_{1}] \xRightarrow[]{\gamma} a_1^{i_1} \ldots
a_d^{i_d}$ in $G^\bowtie$, and by Fact~\ref{fact:bowtie-cap}, we have
$[\textsc{q}^{(s)}_1
X \textsc{q}^{(x)}_{1}] \xRightarrow[]{\iota^{-1}(\gamma)}
w_1^{i_1} \ldots w_d^{i_d}$ in
$G^\cap$. Hence \(w \in \hat{L}_{[\textsc{q}^{(s)}_1
X \textsc{q}^{(x)}_{1}]}(\iota^{-1}(\Gamma), G^\cap)\).

\vspace*{\baselineskip}\noindent (\ref{item:bowtie-xi-complexity})
Given that each production $p^\bowtie \in \Delta^\bowtie$ is the image
  of a production $p^\cap \in \Delta^\cap$ via $\iota$, we have
  $\len{p^\bowtie} = \len{\iota(p^\cap)} \leq \len{p^\cap}$. Hence
  $\len{G^\bowtie} \leq \len{G^\cap}$. Now, each production
  $p^\cap \in \prod^\cap$ corresponds to a production $p$ of $G$, such
  that the nonterminals occurring on both sides of $p$ are decorated
  with at most $3$ nonterminals from $\Vars^\pat$. Since
  $\card{\Vars^\pat} = \len{\pat}$, we obtain that, for each
  production $p$ of $G$, $G^\cap$ has at most $\len{\pat}^3$
  productions of size $\len{p}$. Hence
  $\len{G^\bowtie} \leq \len{G^\cap} \leq \len{\pat}^3 \cdot \len{G}$,
  and $G^\bowtie$ can be constructed in time
  $\len{\pat}^3 \cdot \len{G}$.\qed
\end{proof}

\begin{remark}\label{rem:letter-bounded-inclusion}
  Given $G = \tuple{\Vars,\mathcal{A},\prod}$, $X \in \Vars$, and a
  strict \(d\)-letter-bounded expression $\patt = a_1^* \ldots a_d^*$,
  the check $L_X(G) \subseteq \patt$ can be decided in time
  $\mathcal{O}(\len{\patt} \cdot \len{G})$, by building a grammar
  $\overline{G^{\pat}_1}$ such
  that \(L_{\textsc{q}^{(1)}}(\overline{G^{\pat}_1})
  = \Sigma^* \setminus \patt\) (see
  Remark~\ref{rem:letter-bounded-complement}) and
  checking \(L_X(G) \cap
  L_{\textsc{q}^{(1)}}(\overline{G^{\pat}_1}) \stackrel{?}{=} \emptyset\).
  A similar argument shows that queries $L_X(G) \cap
  (\mathcal{A}^* \cdot
  a_s \cdot \mathcal{A}^*) \stackrel{?}{=} \emptyset,~1 \leq s \leq
	d$, can be answered in time $\mathcal{O}(\len{G})$ \cite[Section~5]{BEF+ipl}.
\end{remark}

\section{Other proofs}\label{sec:hard}

\begin{lemma}\label{lem:leibniz}
Given $G = \tuple{\Vars, \Sigma, \prod}$ and a $k$-index depth-first
step sequence $X\, Y \xRightarrow[\df{k}]{\gamma} w$, for two
nonterminals $X,Y \in \Vars$, $w \in \Sigma^*$, and
$\gamma \in \prod^*$.  There exist $w_1, w_2 \in
\Sigma^*$ such that $w_1\, w_2 = w$, and $\gamma_1, \gamma_2 \in \prod^*$ 
such that either one of the following holds:
\begin{compactenum}
\item $X \xRightarrow[\df{k-1}]{\gamma_1} w_1$ and $Y \xRightarrow[\df{k}]{\gamma_2} w_2$
and $\gamma=\gamma_1\, \gamma_2$, or
\item $X \xRightarrow[\df{k}]{\gamma_1} w_1$ and $Y \xRightarrow[\df{k-1}]{\gamma_2} w_2$
and $\gamma=\gamma_2\, \gamma_1$. 
\end{compactenum}
\end{lemma}
\begin{proof}
The step sequence $X\, Y \xRightarrow[\df{k}]{\gamma} w$ has one of two
possible forms, by the definition of a depth-first sequence:
\begin{compactitem}
\item $X\, Y \xRightarrow[\df{k}]{\gamma_1} w_1\, Y \xRightarrow[\df{k}]{\gamma_2} w_1\, w_2$, or
\item $X\, Y \xRightarrow[\df{k}]{\gamma_2} X\, w_2 \xRightarrow[\df{k}]{\gamma_1} w_1\, w_2$, 
\end{compactitem}
for some words $w_1, w_2 \in \Sigma^*$ and control words
$\gamma_1, \gamma_2 \in \prod^*$. Let us consider the first case, the
second being symmetric. Since $X\, Y \xRightarrow[\df{k}]{\gamma_1}
w_1\, Y$ is a $k$-index step sequence, the sequence
$X \xRightarrow{\gamma_1} w_1$ obtained by erasing the $Y$
nonterminal from the last position in all steps of the sequence, is of index
$k-1$, i.e.\ $X \xRightarrow[\df{k-1}]{\gamma_1} w_1$. Also, since
$w_1\, Y \xRightarrow[\df{k}]{\gamma_2} w_1\, w_2$, we obtain
$Y \xRightarrow[\df{k}]{\gamma_2} w_2$, by erasing the first \(\len{w_1}\)
symbols in all steps of the sequence. Clearly, in this case we have
$\gamma=\gamma_1\, \gamma_2$.\qed
\end{proof}

\subsection{Proof of Lemma~\ref{fsa-dfk}}\label{app:fsa-dfk}

First, we formally define the notion of depth-first derivations by annotating
symbols occurring in every step with a positive integer called the \emph{rank}.
Intuitively, the rank assigns a priority between symbols in a word.
For a set \(S\) of symbols (e.g. the terminals and nonterminals) and a set \(I \subseteq \nats\),
we define \(S^{I} = \{ s^{\tuple{i}} \mid s\in S,\, i\in I\}\) and call
\(S^{I}\) a \emph{ranked alphabet}. We also sometimes write \(S^{\tuple{i}}\)
when \(I\) is a singleton. 
A \emph{ranked word} (r-word) is a word over a ranked alphabet.
Given a word \(w\) of length \(n\) and an \(n\)-dimensional
vector \(\boldsymbol{\alpha}\in\nat^n\), the \emph{ranked
word} \(w^{\boldsymbol{\alpha}}\) is the sequence
\({(w)_1}^{\tuple{(\boldsymbol{\alpha})_1}}\ldots {(w)_n}^{\tuple{(\boldsymbol{\alpha})_n}}\), 
in which the \(i\)th element of \(\boldsymbol{\alpha}\) annotates the \(i\)th symbol of \(w\). We also denote \(w^{ \rank{c} } =
{(w)_1}^{\tuple{c}} \ldots {(w)_{\len{w}}}^{\tuple{c}}\) as a
shorthand.
Let $G = \tuple{\Vars,\Sigma,\prod}$ be a grammar
and \(u
\Arrow{(Z,w)/j} v\) be a step, for a vector 
\(\boldsymbol{\alpha}\in\nat^{\len{u}}\), we define the \emph{ranked step} (r-step)
\(u^{\boldsymbol{\alpha}} \Arrow{(Z,w)/j} v^{\boldsymbol{\beta}}\) if and only if
\((u)_j = Z\) and 
\[v^{\boldsymbol{\beta}} = (u^{\boldsymbol{\alpha}})_1 \cdots
(u^{\boldsymbol{\alpha}})_{j-1} \, w^{\rank{m+1}} \,
(u^{\boldsymbol{\alpha}})_{j+1}\cdots
(u^{\boldsymbol{\alpha}})_{\len{u}} \] where each symbol in \(w\) has
rank \(m+1\) and 
\[m = \max\left(\set{
(\boldsymbol{\alpha})_i \mid \exists i\colon 1 \leq i \leq \len{u},
i\neq j, (u)_i \in \Vars} \cup \set{-1} \right)\] is the maximum among
the ranks of the nonterminals in \(u^{\boldsymbol{\alpha}}\), with
position \(j\) omitted\footnote{If $Z=(u)_j$ is the only non-terminal
in $u$, we have $m+1=-1+1=0$.}. An r-step is said to
be \emph{depth-first},
denoted \(u^{\boldsymbol{\alpha}} \xArrow{}{\textbf{df}}
v^{\boldsymbol{\beta}}\) if{}f the rank of the nonterminal at
position \(j\) where the rule applies is maximal,
i.e. $(\boldsymbol{\alpha})_j=m$.  For instance the transition
labelled \(\mathbf{p_2}\) in Fig.~\ref{fig:running-example} (d) is a
depth-first r-step.  A r-step sequence is said to be depth-first if
all of its r-steps are depth-first. Finally, an unranked step
sequence \(w_0 \Arrow{(\gamma)_1} w_1
\ldots w_{n-1} \Arrow{(\gamma)_{n}} w_n\)
is said to be depth-first,
written \(w_0 \xArrow{\gamma}{\mathbf{df}} w_n\), if{}f there exist
vectors \(\boldsymbol{\alpha}_1 \in \nats^{\len{w_1}}, \ldots,
\boldsymbol{\alpha}_n \in \nats^{\len{w_n}}\) such that 
\(w_0^{\rank{0}} \xArrow{(\gamma)_1}{\mathbf{df}}
w_1^{\boldsymbol{\alpha}_1} \ldots w_{n-1}^{\boldsymbol{\alpha}_{n-1}} 
\xArrow{(\gamma)_{n}}{\mathbf{df}} w_{n}^{\boldsymbol{\alpha}_{n}}\) holds.

Let \(\bdwords^{(k)}
= \{w^{\boldsymbol{\alpha}} \mid \exists  u^{\boldsymbol{\beta}}\colon u^{\boldsymbol{\beta}} = \proj{(w^{\boldsymbol{\alpha}})}{\Vars^{\nats}}, \len{u^{\boldsymbol{\beta}}} \leq  k, \boldsymbol{\beta} \text{ is contiguous}, \max_{i} (\boldsymbol{\beta})_i \leq k-1\} \) be the set of r-words such
that when deleting ranked terminals, the resulting word is no longer than \(k\) and has ranks between $0$ and $k-1$. 
It is routine to check that \(\bdwords^{(k)}\) is closed for the relation \(\xRightarrow[\df{k}]{}\).
For a
r-word \(w^{\boldsymbol{\alpha}}\in\bdwords^{(k)}\),
let \(\age{w^{\boldsymbol{\alpha}}}\) be the r-word
\( (\proj{w^{\boldsymbol{\alpha}}}{\Vars^{\tuple{0}}})\;
(\proj{w^{\boldsymbol{\alpha}}}{\Vars^{\tuple{1}}})
\ldots (\proj{w^{\boldsymbol{\alpha}}}{\Vars^{\tuple{k}}})\).
Intuitively, \(\age{w^{\boldsymbol{\alpha}}}\) projects out the
terminals of \(w\), and orders the remaining nonterminals in the
increasing order of their ranks. For instance, 
\(\age{a^{\tuple{1}}Y^{\tuple{1}}Z^{\tuple{0}}} = Z^{\tuple{0}} Y^{\tuple{1}}\). 
The $\age{.}$ operator is naturally lifted from r-words to sets of r-words.
Recall that we define the set \(Q\) of states of \(A^{\df{k}} = (Q,\prod,\rightarrow)\) as
\(Q = \{ w^{\boldsymbol{\alpha}} \mid w \in \Vars^{*}, \len{w}\leq k, \boldsymbol{\alpha} \text{ is contiguous},
(\boldsymbol{\alpha})_1 \leq \cdots \leq (\boldsymbol{\alpha})_{\len{w}} \} \).
It is routine to check that \(\age{\bdwords^{(k)}} = Q\) holds.
Now let us consider \(\rightarrow\) which we defined as follows.
Let \(q,q' \in Q\), \( (X,w)\in \Delta\) we have \(q \xrightarrow{(X,w)} q'\) if{}f 
\begin{compactitem}
\item \( q = u\, X^{\tuple{i}}\, v\) for some \(u,v\) and where \(i\) is the maximum rank in \(q\), and 
\item \( q'= u\, v\, (\proj{w}{\Vars})^{\rank{i'}} \)
where \(\len{u\, v\, (\proj{w}{\Vars})^{\rank{i'}}}\leq k\) and \( i' = 
\begin{cases}
  0 & \text{if } u\, v = \varepsilon\\
	i & \text{else if } \proj{(u\, v)}{\Vars^{\tuple{i}}} = \varepsilon\\
	i+1 & \text{else}
\end{cases}\)
\end{compactitem}
As \(q\in Q\), we find that \(q\in \age{\bdwords^{(k)}}\). 
Furthermore, it is an easy exercise to show that  
\(q \xrightarrow{(X,w)} q'\) if{}f there exists \(w^{\boldsymbol{\eta}} \in \bdwords^{(k)}\) 
such that \(q \xRightarrow[\df{k}]{(X,w)}  w^{\boldsymbol{\eta}}\) and \( \age{ w^{\boldsymbol{\eta}} } = q'\).
It follows that, we can equivalently write $A^{\df{k}}_G
= \tuple{ \age{\bdwords^{(k)}}, \prod, \rightarrow}$ for the labeled
graph the edge relation, is defined as:
\(u^{\boldsymbol{\alpha}} \xrightarrow{p} v^{\boldsymbol{\beta}}\)
if{}f
\(\exists w^{\boldsymbol{\eta}} \in \bdwords^{(k)}  \ldotp
u^{\boldsymbol{\alpha}} \xRightarrow[\df{k}]{p}
w^{\boldsymbol{\eta}} \land v^{\boldsymbol{\beta}}
= \age{w^{\boldsymbol{\eta}}}\). 

\begin{proof}[of Lemma~\ref{fsa-dfk}]
  \noindent ``$\Rightarrow$'' We shall prove the following more
  general
  statement. Let \(u^{\boldsymbol{\alpha}} \xRightarrow[\df{k}]{\gamma}
  w^{\boldsymbol{\beta}}\)
  where \(u^{\boldsymbol{\alpha}} \in \bdwords^{(k)}\) be
  a \(k\)-index depth-first r-step sequence. By induction
  on \(\len{\gamma}\geq 0\), we show the existence of a
  path \(\age{u^{\boldsymbol{\alpha}}}\xrightarrow{\gamma} \age{w^{\boldsymbol{\beta}}}\)
  in \(A^{\df{k}}\).  For the base case \(\len{\gamma}=0\), we
  have \(u^{\boldsymbol{\alpha}} = w^{\boldsymbol{\beta}}\) which
  yields \(\age{u^{\boldsymbol{\alpha}}}=\age{w^{\boldsymbol{\beta}}}\)
  and since \(u^{\boldsymbol{\alpha}}\in\bdwords^{(k)}\) the
  hypothesis shows
  that \(u^{\boldsymbol{\alpha}},w^{\boldsymbol{\beta}}\in\bdwords^{(k)}\),
  hence
  that \(\age{u^{\boldsymbol{\alpha}}},\age{w^{\boldsymbol{\beta}}}\in\age{\bdwords^{(k)}}\)
  and we are done. For the induction step \(\len{\gamma}>0\), let \(
  v^{\boldsymbol{\eta}} \xRightarrow[\df{k}]{p}
  w^{\boldsymbol{\beta}} \) be the last step of the sequence, for
  some \( p \in \prod\), i.e.\ \(\gamma=\sigma \cdot p\)
  with \(\sigma\in\prod^*\). By the induction
  hypothesis, \(A^{\df{k}}\) has a
  path \(\age{u^{\boldsymbol{\alpha}}} \xrightarrow{\sigma} \age{v^{\boldsymbol{\eta}}}\).
  Since \(\age{v^{\boldsymbol{\eta}}}, \age{w^{\boldsymbol{\beta}}} \in \age{\bdwords^{(k)}}\)
  and \( v^{\boldsymbol{\eta}} \xRightarrow[\df{k}]{p}
  w^{\boldsymbol{\beta}}\), we have
  that \( \age{v^{\boldsymbol{\eta}}} \xrightarrow{p} \age{w^{\boldsymbol{\beta}}} \)
  by definition of \(\rightarrow\), hence we obtain a
  path \( \age{u^{\boldsymbol{\alpha}}} \xrightarrow{\gamma} \age{w^{\boldsymbol{\beta}}}\).
	
  \noindent ``$\Leftarrow$'' We prove a more general statement. Let
  $U \xrightarrow{\gamma} W$ be a path in $A^{\df{k}}_G$, for some
  words \(U,W\in\age{\bdwords^{(k)}}\). We show by induction on
  $\len{\gamma}$ that there exist r-words $u^{\boldsymbol{\alpha}},
  w^{\boldsymbol{\beta}} \in \bdwords^{(k)}$, such
  that \(\age{u^{\boldsymbol{\alpha}}} =
  U\), \(\age{w^{\boldsymbol{\beta}}} = W\),
  and \(u^{\boldsymbol{\alpha}} \xRightarrow[\df{k}]{\gamma}
  w^{\boldsymbol{\alpha}}\). The base case $\len{\gamma}=0$ is
  trivial, because \(U=W\) and since \(U\in\age{\bdwords^{(k)}}\) then
  there exists \(u^{\boldsymbol{\alpha}}\in \bdwords^{(k)}\) such
  that \(\age{u^{\boldsymbol{\alpha}}}=U=W\) and we are done.  For the
  induction step $\len{\gamma} > 0$, let \(\gamma = \sigma \cdot p\),
  for some production \(p \in \prod\) and \(\sigma \in \prod^*\). By
  the induction hypothesis, there exist r-words
  $u^{\boldsymbol{\alpha}}, v^{\boldsymbol{\eta}} \in \bdwords^{(k)}$
  such
  that \(U=\age{u^{\boldsymbol{\alpha}}} \xrightarrow{\sigma} \age{v^{\boldsymbol{\eta}}} \xrightarrow{p}
  W\) is a path in \(A^{\df{k}}\),
  and \(u^{\boldsymbol{\alpha}} \xRightarrow[\df{k}]{\sigma}
  v^{\boldsymbol{\eta}}\) is a \(k\)-index r-step sequence. The
  definition of the edge relation in $A^{\df{k}}$
  and \(\age{v^{\boldsymbol{\eta}}} \xrightarrow{p} w\) shows
  that \(v^{\boldsymbol{\eta}} \xRightarrow[\df{k}]{p}
  w^{\boldsymbol{\beta}}\) for
  some \(w^{\boldsymbol{\beta}}\in\bdwords^{(k)}\) such
  that \(\age{w^{\boldsymbol{\beta}}} = W\). 

For the upper bound on the size of \(A^{\df{k}}\), recall that each vertex of
$A^{\df{k}}$ is a ranked word of length at most $k$, consisting of
non-terminals only, with ranks in the interval \([0,k-1]\). Moreover, the
productions of \(G\) do not produce more than \(2\) nonterminals at a time.
Hence, in every vertex of \(A^{\df{k}}\), at most \(2\) positions carry the
same rank.  Since the length of each vertex in \(Q\) is at most \(k\) and, for
each \(i\in[0,k-1]\), there are at most \(\card{\Vars}^2\) choices of
nonterminals with rank \(i\), we have $\len{A^{\df{k}}_G} \leq
\card{\Vars}^{2k} \leq \len{G}^{2k}$.  \qed
\end{proof}

\subsection{Proof of Lemma~\ref{lem:ginsbook-d}}

When $L_{X,Y}(G) \subseteq \patt$, because \(\patt=a_{1}^* \ldots
a_{s}^*\) is a strict \(s\)-letter-bounded expression with \(s\) a
fixed constant, for every step sequence $X \xRightarrow{\gamma}_G u \,
Y \, v$, we have \(u\, v
= \proj{\gamma}{a_{1}} \ldots \proj{\gamma}{a_{s}}\). Also remark
that \(u\, v = a_1^{(\boldsymbol{v})_1} \ldots a_s^{(\boldsymbol{v})_s}\) for some \(
\boldsymbol{v}\in\nats^s\), hence that \( (\boldsymbol{v})_{\ell}
= \len{\proj{\gamma}{a_\ell}}\) for each \(\ell=1,\ldots,s\).  
For convenience, given \(\gamma\in\prod^{*}\), we
denote \(\projpatt{\gamma}
= \proj{\gamma}{a_1} \ldots \proj{\gamma}{a_s}\).

We recall the definition of the labeled graph
\(A^{\df{k}}= \tuple{ Q, \prod, \rightarrow}\)
whose number of vertices we denote by \(N\). Due to the form of the
productions in \(G\), we can safely restrict \(Q\) to r-words with at
most $2$ nonterminals having the same rank,
hence \(N \leq \len{G}^{2k}\). We define $\cycles{}{q}$ is the set of
elementary cycles with $q \in Q$ as endpoints.

\begin{proposition}\label{prop:weighted-graph-decomposition}
Let $G = \tuple{\Vars, \Sigma, \prod}$ be a grammar, $X \in \Vars$ be
a nonterminal and $\patt = a_1^* \ldots a_s^*$ be a strict $s$-letter bounded
expression, for some $s \geq 0$. For any two vertices
$q,q' \in Q$ of $A^{\df{k}}$, and any path
$\pi \in \Pi(q,q')$, there exists a path $\pi' \in \Pi(q,q')$ such
that $\len{\pi} = \len{\pi'}$, $\projpatt{\omega(\pi)}
= \projpatt{\omega(\pi')}$ and $\pi'$ is of the form
$\varsigma_1 \cdot \theta_1 \cdots \varsigma_\ell \cdot \theta_\ell \cdot \varsigma_{\ell+1}$,
where \(\varsigma_1 \in \Pi(q,q_{1})\), \(\varsigma_{\ell+1} \in \Pi(q_{\ell},q')\) and 
\(\varsigma_{j} \in \Pi(q_{{j-1}},q_{{j}})\), for
each \(1< j \leq \ell\), are acyclic paths,
$\theta_1 \in \cyclestar{}{q_{1}}, \ldots, \theta_\ell \in \cyclestar{}{q_{\ell}}$
are cycles, and $\ell \leq \card{Q}$.
\end{proposition}
\begin{proof} 
The proof goes along the lines of that of Lemma~7.3.2 in Lin's PhD thesis \cite{ToThesis}. This
proof is carried on graphs labeled with integer tuples, and addition,
instead of concatenation. Since the only property of integer tuple
addition, used in the proof of \cite[Lemma~7.3.2]{ToThesis}, is
commutativity, it suffices to observe that $\projpatt{\omega(\pi)}
= \projpatt{\omega(\pi')}$, whenever $\omega(\pi)$ is a permutation of
$\omega(\pi')$.\qed
\end{proof}

\begin{proof}[of Lemma~\ref{lem:ginsbook-d}]
Given two step sequences $X \xRightarrow{\gamma}_G u \, Y \, v$,
$X \xRightarrow{\gamma'}_G u' \, Y \, v'$, the following are equivalent:
\begin{compactitem}
\item \( \len{\proj{\gamma}{a_{\ell}}} = \len{\proj{\gamma'}{a_{\ell}}} \) for all \(\ell=1,\ldots,s\),
\item \( \projpatt{\gamma} = \projpatt{\gamma'}\),
\item \( u\, v = u'\, v'\).
\end{compactitem}
Since \(L_{X,Y}(G) \subseteq \patt\) where \(\patt\) is a
strict \(s\)-letter bounded expression, for
every \(\pi \in \cycles{}{q}\) the induced word $a_1^{k_1} \ldots
a_s^{k_s} = \projpatt{\omega(\pi)}$ is such that: $\sum_{j=1}^s
k_{j} \leq 2N$, i.e.\ each production in $\prod$ issues at most $2$
symbols from $\{a_1,\ldots,a_s\}$, and each elementary cycle is of
length at most $N$. The nonnegative solutions of the inequation
$\sum_{j=1}^s k_{j} \leq 2N$ are solutions to the equation
$\sum_{j=1}^s k_{j} + y = 2N$, for a nonnegative slack variable
$y \geq 0$. Since the number of nonnegative solutions to the latter
equation\footnote{The number of nonnegative solutions of an equation
$n=x_1+\cdots+x_m$ is $\binom{m+n-1}{m-1}$.} is $\binom{s+2N}{s}$, we
have:
\begin{equation}\label{eq:binom}
	\card{ \{ \projpatt{\omega(\pi)} \mid \pi \in \cycles{}{q}\} } = \binom{s+2N}{s} = \mathcal{O}(N^s)\enspace .
\end{equation}
For each vertex \(q\), we are interested in a
set \(C_q \subseteq \cycles{}{q}\) such that \(\card{C_q}
= \mathcal{O}(N^{s})\) and,
moreover, for each \(\pi\in\cycles{}{q}\) there exists \(\pi' \in C_q\) such that
\(\projpatt{\omega(\pi)} = \projpatt{\omega(\pi')}\)
when
\(\Pi(X^{\tuple{0}},q) \neq \emptyset\) and \(\Pi(q, Y^{\tuple{0}}) \neq \emptyset\) holds.

For now we assume we have computed such sets \( \set{C_q}_{q \in
Q}\) (their effective computation will be described later).
We are now ready to define the bounded expression $\pattg$. Given a
finite set $\Gamma = \{\gamma_1,\ldots,\gamma_n\} \subseteq \prod^*$
of control words indexed following some total ordering (e.g.\ we assume a
total order $\prec$ on $\Vars \cup \mathcal{A}$, and define $(X_1,w_1)
\prec_\prod (X_2,w_2) \iff X_1\cdot w_1 \prec^{lex} X_2\cdot w_2$ in
the lexicographical extension of $\prec$, then extend $\prec_\prod$ to
a lexicographical order $\prec_\prod^{lex}$ on control words), we
define the bounded expression: 
\(\mathit{concat}(\Gamma) = \gamma_1^* \cdots \gamma_n^*\).
Let \(Q = \set{q_1, \ldots, q_N}\) be the set of
vertices of \(A^{\df{k}}\), taken in some order. We define the
set \( \set{B_i}_{i\geq 0}\) of bounded expressions as follows:
\begin{align*}
	B_0 & =  \mathit{concat}(\{ \omega(\pi) \mid \pi \in C_{q_1}\}) \cdots \mathit{concat}(\{ \omega(\pi) \mid \pi \in C_{q_N}\}) \\
B_1 & = \mathit{concat}(\prod)^{N-1} \cdot B_0 \cdot \mathit{concat}(\prod)^{N-1}\\
B_i & =  \mathit{concat}(\prod)^{N-1} \cdot B_0 \cdot B_{i-1}, \text{ for all } i \geq 2
\shortintertext{Finally, let:}
\pattg & = B_N\enspace.
\end{align*}
Let us now prove the language inclusion.

\vspace*{\baselineskip}\noindent 
It follows from Theorem~\ref{thm:luker}, that $L^{(k)}_{X,Y}(G)
= \hat{L}_{X,Y}(\Gamma_{X,Y}^{\df{k}}, G)$ for every \(X\in \Vars\),
\(Y\in\Vars\cup\set{\varepsilon}\) and \(k>0\). Hence we trivially have
$\hat{L}_{X,Y}(\pattg \cap \Gamma_{X,Y}^{\df{k}}, G) \subseteq
\hat{L}_{X,Y}(\Gamma_{X,Y}^{\df{k}}, G) = L^{(k)}_{X,Y}(G)$.  For the contrapositive
$L^{(k)}_{X,Y}(G) \subseteq \hat{L}_{X,Y}(\pattg \cap \Gamma_{X,Y}^{\df{k}},G)$, it
suffices to show the following: given a $k$-index depth first step sequence \(
X \xRightarrow[\df{k}]{\gamma} u \, Y \, v\), there exists a control word
\(\gamma' \in \pattg\) such that \(X \xRightarrow[\df{k}]{\gamma'} u'\, Y\,
v'\) and \(u\, v = u'\, v'\).

Because Lemma~\ref{fsa-dfk} shows that each
path \(\pi \in \Pi( X^{\tuple{0}}, Y^{\tuple{0}} ) \) corresponds to a
control word \(\omega(\pi)\) such
that \(X \xRightarrow[\df{k}]{\omega(\pi)} u\, Y\, v\), and because \( L_{X,Y}^{(k)}(G)\subseteq \patt\) where \(\patt\) is a
strict \(s\)-letter bounded expression, it suffices to show that
exists a path \(\rho \in \Pi( X^{\tuple{0}}, Y^{\tuple{0}} ) \) such
that \(\omega(\rho) \in \pattg\)
and \( \projpatt{\omega(\pi)} = \projpatt{\omega(\rho)}\).
We apply the result from Prop.~\ref{prop:weighted-graph-decomposition}
which shows that there exists a path
\(\rho \in \Pi( X^{\tuple{0}}, Y^{\tuple{0}} )\), such that $\len{\rho} = \len{\pi}$,
$\projpatt{\omega(\rho)} = \projpatt{\omega(\pi)}$ and $\rho$ is of
the form
$\varsigma_1 \cdot \theta_1 \cdots \varsigma_\ell \cdot \theta_\ell \cdot
\varsigma_{\ell+1}$, where 
\(\varsigma_1     \in \Pi(X^{\tuple{0}},q_{i_1})\), 
\(\varsigma_{\ell+1} \in \Pi(q_{i_\ell}, Y^{\tuple{0}})\), and 
\(\varsigma_{j}   \in \Pi(q_{i_{j-1}},q_{i_{j}})\) for each \(1< j \leq \ell\)
are acyclic paths,
$\theta_1 \in \cyclestar{}{q_{i_1}}, \ldots, \theta_\ell \in \cyclestar{}{q_{i_\ell}}$
are cycles, $q_{i_1}, \ldots, q_{i_{\ell}}$ are vertices, and
$\ell \leq \card{Q}$. Hence we conclude that
\begin{compactitem}
\item $\omega(\varsigma_j) \in \mathit{concat}(\prod)^{N-1}$, for all $1 \leq j \leq \ell+1$, 
\item for each cycle $\theta_j \in \cyclestar{}{q_{i_j}}$, consisting of a concatenation of 
several elementary cycles
$\theta^1_j, \ldots, \theta^{\ell_j}_j \in \cycles{}{q_{i_j}}$, the cycle
$\theta^{lex}_j$ obtained by a lexicographic reordering of
$\theta^1_j, \ldots, \theta^{\ell_j}_j$ (based on the lexicographic order
of their value in $\prod^*$) belongs to $B_0$, for all $1 \leq j \leq
\ell$.  Second, it is easy to see that the words produced by $\theta_j$
and $\theta^{lex}_j$ are the same, since the order of productions
labeling $\theta_j$ ($\theta^{lex}_j$) is not important.
\end{compactitem}
Let $\pi'$ be the path
$\varsigma_1 \cdot \theta^{lex}_1 \cdots \varsigma_\ell \cdot \theta^{lex}_\ell \cdot
\varsigma_{\ell+1}$. By Prop.~\ref{prop:weighted-graph-decomposition}, we have that 
$\projpatt{\omega(\pi)} = \projpatt{\omega(\pi')}$.
Moreover, $\omega(\pi') \in B_N = \pattg$.  Since $X \xRightarrow[\df{k}]{\omega(\pi)} u \,
Y \, v$ and $X
\xRightarrow[\df{k}]{\omega(\pi')} u' \, Y \, v'$ are step sequences
of $G$, the previous equality implies $u \, v = u' \, v'$.

Concerning the time needed to construct the bounded
expression \(\pattg\), the main ingredient in the previous, is the
definition of the sets of
cycles \(\set{C_q}_{q\in Q}\), such that $\card{C_q}
= \mathcal{O}(N^s)$ and, moreover, for
each \(\pi\in\cycles{}{q}\) there exists \(\pi' \in C_q\) such that
\(\projpatt{\omega(\pi)}= \projpatt{\omega(\pi')}\)
when
\(\Pi(X^{\tuple{0}},q) \neq \emptyset\) and \(\Pi(q, Y^{\tuple{0}}) \neq \emptyset\) holds.
Below we describe the construction of such sets.

Define \(\mathit{Val} = \{a_1^{\ell_1} \ldots
a_s^{\ell_s} \in \patt \mid \sum_{j=1}^s \ell_{j} \leq 2N\}\).
Using previous arguments (i.e. equation (\ref{eq:binom})), it is routine to
check that \(\card{\mathit{Val}} = \mathcal{O}(N^{s})\). Consider the
labeled graph $\mathcal{H} = \tuple{V, \prod, \arrow{}{}}$, defined upon 
\(A^{\df{k}}\), where:
\begin{compactitem}
\item $V = Q \times \mathit{Val}$, and
\item \(\tuple{q',a_1^{i_1}\ldots a_s^{i_s}} \arrow{(Z,z)}{}
	\tuple{q'', a_1^{j_1}\ldots a_s^{j_s}}\) if{}f \(q' \arrow{(Z,z)}{A^{\df{k}}} q''\) 
        and \(a^{j_{\ell}}_{\ell} = a^{i_\ell}_{\ell}\cdot \proj{z}{a_{\ell}}\) for each \(\ell\)
\end{compactitem}

First, observe that the number of vertices in this graph is \(\card{V}
\leq N^{2k} \cdot \binom{s+2N}{s} = \len{G}^{\mathcal{O}(k)}\). 
Second, it is routine to check (by induction on the length of a path)
 that given a
 path \(\pi \in \Pi_{\mathcal{H}}(\tuple{q,\varepsilon}, \tuple{q,a_1^{i_1}\ldots a_s^{i_s}})\)
 for some \(i_1,\ldots,i_s \in\nats\) we have \( \projpatt{\omega(\pi)} = a_1^{i_1}\ldots a_s^{i_s}\). Next, for
 each \(q\in Q\) define the set \(\mathcal{P}_{q}\)
 of paths of \(\mathcal{H}\) consisting for
 each \(a_1^{i_1}\ldots a_s^{i_s}\in\mathit{Val}\) of a single path (one with the least
 number of edges) from \(\tuple{q,\varepsilon}\) to \(\tuple{q,a_1^{i_1}\ldots a_s^{i_s}}\).
 By definition of \(\mathit{Val}\), we have
 that \(\card{\mathcal{P}_q} = \card{\mathit{Val}}
 = \mathcal{O}(N^s)\) and,
 moreover, for each \(\rho \in \cycles{}{q}\) (\(\rho\) is a path of \(A^{\df{k}}\)) there exists a path 
 \(\pi\in\mathcal{P}_q\) such that \( 
 \projpatt{\omega(\rho)} = \projpatt{\omega(\pi)} = a_1^{i_1} \ldots a_s^{i_s} \)
 where \( \tuple{q,\varepsilon} \) and \( \tuple{q, a_1^{i_1} \ldots a_s^{i_s}} \) are the endpoints of \(\pi\). 

Hence, we define $C_q$ to be the set of cycles in $A^{\df{k}}$ corresponding to the paths in
$\mathcal{P}_q$.  The latter can be computed applying Dijkstra's single source
shortest path algorithm on $\mathcal{H}$, with source vertex
$\tuple{q, \varepsilon}$, and assuming that the distance between adjacent
vertices is always 1. The running time of the Dijkstra's algorithm is
$\mathcal{O}(\card{V}^2) = \len{G}^{\mathcal{O}(k)}$. Upon
termination, one can reconstruct a shortest path $\pi$ from
$\tuple{q, \varepsilon}$ to each vertex $\tuple{q, a_1^{i_1}\ldots a_s^{i_s}}$, and add the
corresponding cycle of $A^{\df{k}}$ to $C_q$. Since there are at most
$\len{G}^{\mathcal{O}(k)}$ vertices $\tuple{q, a_1^{i_1}\ldots a_s^{i_s}}$ in $V$, and
building a shortest path for each such vertex takes at most
$\len{G}^{\mathcal{O}(k)}$ time, we can populate the set $C_q$ in time
$\len{G}^{\mathcal{O}(k)}$. Once the sets $C_q$ are built, it remains
to compute the bounded expressions $\mathit{concat}(\{ \omega(\pi) \mid \pi \in C_{q}\})$,
$\mathit{concat}(\Delta)^{N-1}$ and $B_0, \ldots, B_N$.  As shown
below, they are all computable in time $\len{G}^{\mathcal{O}(k)}$.

Algorithm~\ref{alg:constant-control-set} gives the construction
of \(\pattg\). An upper bound on the time needed for
building \(\pattg\) can be derived by a close analysis of the running
time of Algorithm~\ref{alg:constant-control-set}. The input to the
algorithm is a grammar $G$, a strict $s$-letter bounded expression
$\patt$ and an integer $k > 0$. First
(lines~\ref{ln:h-start}--\ref{ln:h-end}) the algorithm builds the
$\mathcal{H}$ graph, which takes time $\len{G}^{\mathcal{O}(k)}$. The
loop on (lines~\ref{ln:b0-start}--\ref{ln:b0-end}) computes, for each
vertex $q \in Q$, and each \(s\)-dimensional vector
$\boldsymbol{v} \in \mathit{Val}$, an elementary path from
$\tuple{q,\varepsilon}$ to $\tuple{q,a_1^{(\boldsymbol{v})_1}\ldots a_s^{(\boldsymbol{v})_s}}$ in
$\mathcal{H}$.  For each $q$, this set is kept in a variable
$C_q$ (line~\ref{ln:cq}). The variable $B_0$ at the end of
the loop contains the expression
$\mathit{concat}(\{ \omega(\pi) \mid \pi \in \mathcal{P}_{q_1}\}) \cdots \mathit{concat}(\{ \omega(\pi) \mid \pi \in \mathcal{P}_{q_N}\})$,
Since both $\card{Q} = \len{G}^{\mathcal{O}(k)}$
and $\card{\mathit{Val}} = \len{G}^{\mathcal{O}(k)}$, the loop at
(lines~\ref{ln:b0-start}--\ref{ln:b0-end}) takes time
$\len{G}^{\mathcal{O}(k)}$ as well.

The remaining part of the algorithm computes first an
over-approximation of $\mathit{concat}(\prod)^{N-1}$
(lines~\ref{ln:c-start}--\ref{ln:c-end}) in the variable $C$---observe
that the algorithm computes $\mathit{concat}(\prod)^{\len{G}^{2k}}-1$
instead of $\mathit{concat}(\prod)^{N-1}$. Finally, the control set
$\pattg$ with the needed property is produced by
$\len{G}^{2k} \geq N$ repeated concatenations of the bounded
expression $C \cdot B_0$, at lines
(\ref{ln:pat-start}--\ref{ln:pat-end}). Since both loops take time at
most $\len{G}^{2k}$, we conclude that Algorithm~\ref{alg:constant-control-set} 
runs in time $\len{G}^{\mathcal{O}(k)}$.\qed
\end{proof}

\subsection{Proof of Lemma~\ref{lem:ginsbook-surgery}}

A grammar \(G\) is said to be \emph{reduced}
for \(X\) if{}f $L_{X,Y}(G) \neq \emptyset$
and \(L_{Y}(G)\neq \emptyset\), for every \(Y \in \Vars\), \(X\neq
Y\). A grammar can be reduced in polynomial time, by eliminating
unreachable and unproductive
nonterminals \cite[Lemma~1.4.4]{ginsburg}.

\begin{proof}[of Lemma~\ref{lem:ginsbook-surgery}]
We start by proving a series of five facts.  
\begin{compactenum}[\upshape(\itshape i\upshape)]
\item\label{item:easyfact1} First, no production of \(G\) has the form
  \( (Y,v) \), where \(Y\in\Varsi\) and \(v\) contains a symbol of
  \(\Varse\).  By contradiction, assume such a production exists where
  \(Z \in \Varse\) is a nonterminal occurring in \(v\).  Because \(Z
  \in \Varse\), \(a_1\) occurs in some word of \(L_{Z}(G)\) and
  \(a_d\) occurs in some word of \(L_Z(G)\).  On the other hand, we
  have that either no word of \(L_{Y}(G)\) contains \(a_1\) or no
  word of \(L_{Y}(G)\) contains \(a_d\), since \(Y\in\Varsi\). Because
  \(G\) is reduced, we have \( \set{u \mid v
    \xRightarrow{}^* u} \neq \emptyset \). We reach a contradiction,
  since \( \{u \mid Y \xRightarrow{(Y,v)} v \xRightarrow{}^* u\} \)
  contains a word in which \(a_1\) occurs and a word in which \(a_d\)
  occurs, because \(Z\) occurs in \(v\).

  \item\label{item:easyfact2} Define \(Q(u,v)\) to be the following proposition:
\begin{align*}
\set{u'\in (\Vars\cup\mathcal{A})^* \mid u\Rightarrow^* u' } &\subseteq ( \set{a_1} \cup \Varsi)^*\\
\shortintertext{\textbf{and}}
\set{v'\in (\Vars\cup\mathcal{A})^* \mid v\Rightarrow^* v' } &\subseteq ( \set{a_d} \cup \Varsi)^*\enspace.
\end{align*}
We show that \( Q(u,v)\) holds if \(X_i \Rightarrow^* u\, X_j\, v\)
with \(X_i,X_j\in\Varse\). By contradiction, assume that there exists
\(u'\) such that \(u \xRightarrow{}^* u'\) and \(u' \notin ( \set{a_1}
\cup \Varsi)^*\) (a similar argument holds for \(v\)).
Then either \begin{inparaenum}[\upshape(\itshape a\upshape)]
\item \(u'\) contains a symbol \(a_\ell\), for \(\ell>1\) or 
\item \(u'\) contains a nonterminal \(Z \in \Varse\).
\end{inparaenum}
Because \(G\) is reduced, we have \(\set{u' \mid u \xRightarrow{}^* u'}
\neq \emptyset\). In either case (a) or (b), there exists a step
sequence \(u' \Rightarrow^* u_1\, a_\ell\, u_2 \in \mathcal{A}^* \) such
that \(\ell>1\). Since \(X_j\in\Varse\), we have that \(X_j\, v
\Rightarrow^* a_1\, u_3 \in \mathcal{A}^*\), hence that \(
X_i\Rightarrow^* u_1\, a_\ell\, u_2\, a_1\, u_3 \) and finally that
\(L_X(G) \nsubseteq \patt\), since \(G\) is reduced, a contradiction.

\item\label{item:easyfact3} For every step sequence \(X_j
  \Rightarrow^* x\), where \(X_j\in\Varse\), \(x\) cannot be of the
  form \(u_1\, X_d\, u_2\, X_e\, u_3\) where \(X_d,X_e\in \Varse\).
  In fact, take the decomposition \(u = u_1 \) and \( v = u_2\, X_e\,
  u_3 \) (the case \(u = u_1\, X_d\, u_2\) and \( v = u_3\) yields the
  same result). Because \( (\ref{item:easyfact2}) \) applies, we find
	that \( Q(u,v)\) holds but \(v\notin (\set{a_d}\cup \mathcal{A} \cup\Varsi)^*\), hence a contradiction.

\item\label{item:easyfact4} If \( X \xRightarrow{\gamma}_G u\, X_i\,
  v\) is a step sequence of \(G\), for some \(X_i\in\Varse\),
  \(\gamma\in\prod^*\) then \( X \xRightarrow{\gamma}_{G^\sharp} u\,
  X_i\, v\) is also a step sequence of \(G^\sharp\). The proof goes by
  induction on \(n=\len{\gamma}\). Let \( X = w_0
	\xRightarrow{(\gamma)_{1}}_G w_1 \cdots w_{n-1} \xRightarrow{(\gamma)_{n}}_G
  w_n = u\, X_i\, v\).  If \(n=0\) then \(\gamma=\varepsilon\),
  \(X=X_i\in\Varse\) and \(u=v=\varepsilon\), which trivially yields a
  step sequence of \(G^\sharp\). For the inductive case, because of \(
	(\ref{item:easyfact1}) \) we find that, necessarily, \( (w_{n-1})_{\ell} \in \Varse \) for some \(\ell\).  We thus can apply the
  induction hypothesis onto \( X \xRightarrow{(\gamma)_{1}\ldots
    (\gamma)_{n-1}}_G w_{n-1}\) and conclude that \( X
  \xRightarrow{(\gamma)_{1}\ldots (\gamma)_{n-1}}_{G^\sharp}
	w_{n-1}\). Next, since \(w_{n-1}\xRightarrow{(\gamma)_n} w_{n}\) it cannot be the case that
	\(w_{n-1}\xRightarrow{ (\gamma)_n / p} w_{n}\) where \(p\neq \ell\) and \( (\gamma)_n=(Y,t) \) with \(Y\in \Varse\) for otherwise
	\(X \Rightarrow^*_G w_{n-1}\) contradicts \((\ref{item:easyfact3})\) (recall that both \( (w_{n-1})_{\ell} \) and \( X \) belong to \(\Varse\)).
	Thus we have \((\gamma)_n\in\prod^\sharp\),
  hence \( w_{n-1} \xRightarrow{(\gamma)_n}_{G^\sharp} w_n\), and
  finally \( X \xRightarrow{\gamma}_{G^\sharp} u\, X_i\, v\).
\item\label{item:easyfact5} If \(L_1,L_2 \subseteq \patt\) and
  \(L_1\cdot L_2 \subseteq a_{\ell}^* \ldots a_r^*\), for some \(1
  \leq \ell \leq r \leq d\), then there exists \(\ell \leq q \leq r\) such that \(L_1 \subseteq a_{\ell}^* \ldots a_q^*\) and \(L_2
  \subseteq a_{q}^* \ldots a_{r}^*\). Assume, by contradiction, that
  there is no such \(q\).  Then there exist words \(w_1 =
  a_{\ell}^{i_{\ell}}\ldots a_{r}^{i_r}\in L_1\) and \( w_2 =
  a_{\ell}^{j_{\ell}}\ldots a_{r}^{j_r} \in L_2\), two positions
  \(p_1,p_2\) such that \(\ell\leq p_2 < p_1 \leq r\) such that \(
  i_{p_1}\neq 0 \), \( j_{p_2}\neq 0 \). Because all \(a_i\) are
  distinct, we conclude that \(w_1 \cdot w_2 \notin a_{\ell}^* \ldots
  a_r^*\), hence a contradiction.
\end{compactenum}
We continue with the proof of the five items of the lemma:
\begin{compactenum}[1.]
\item The derivation $X \xRightarrow[\df{k}]{\gamma} w$, where \(\len{\gamma} = n\), has a unique
  corresponding r-step sequence $X^{\tuple{0}}=w_0^{\boldsymbol{\alpha}_0}
  \xRightarrow{ (\gamma)_1 } w_1^{\boldsymbol{\alpha}_1} \ldots
  \xRightarrow{ (\gamma)_n } w_n^{\boldsymbol{\alpha}_{n}}=w^{\boldsymbol{\alpha}_{n}}$.  Now, we define a
  \emph{parent} relationship in that step sequence, denoted
  \(\triangleleft\), between r-annotated nonterminals: $Y^{\tuple{a}}
  \mathbin{\triangleleft} Z^{\tuple{b}}$
  if{}f there exists a step in the sequence that rewrites
  \(Y^{\tuple{a}}\) to \(Z^{\tuple{b}}\), that is \(u^{\boldsymbol{\alpha}} \xRightarrow{(Y,t)/j} v^{\boldsymbol{\beta}}\)
	where \( (u^{\boldsymbol{\alpha}})_{j} = Y^{\tuple{a}} \), and \( (v^{\boldsymbol{\beta}})_{\ell} = Z^{\tuple{b}} \) for some \(j\leq \ell \leq j-1+\len{t}\).

  Let $(\gamma)_{\ell_p} = (X_{i_p},a\, y\, b\, z)$ be the last occurrence, in
  \(\gamma\), of a production with head $X_{i_p} \in \Varse$.  Notice
  that such an occurrence always exists since \(X \in \Varse\) and
	moreover we have that \(a,b\in \mathcal{A} \cup \set{\varepsilon}\), \(y,z \in \Varsi \cup 
  \set{\varepsilon}\). In fact, since \(\gamma\) is a derivation, if
  \(y\in\Varse\) or \(z\in\Varse\) then \( (\gamma)_{\ell_p} \) would
  clearly not be the last such occurrence. Let $X=X_{i_0}^{\tuple{r_0}}
  \mathbin{\triangleleft} X_{i_1}^{\tuple{r_1}} \mathbin{\triangleleft}
  \cdots \mathbin{\triangleleft} X_{i_p}^{\tuple{r_p}}$ be the
  sequence of ranked ancestors of $X_{i_p}$ in the r-step
  sequence, and \((\gamma)_{\ell_j} = (X_{i_j},a\, y_{m_j}\, b\, X_{i_{j+1}}) \in \prod\) (or, symmetrically
\( (\gamma)_{\ell_j} = (X_{i_j}, a\, X_{i_{j+1}}\, b\, z_{m_j}) \in \prod \)),
	for some \(a,b\in \mathcal{A}\cup\set{\varepsilon}\), $z_{m_j}, y_{m_j} \in
	\Vars \cup \{\varepsilon\}$, be the productions introducing these
	nonterminals, for all $0 \leq j < p$.

  If $y_{m_j} \in \Vars$, let $\overline{\gamma}_j$ be the subword of
  $\gamma$ corresponding to the derivation $y_{m_j}
  \xRightarrow{\overline{\gamma}_j} w_{m_j}$, for some $w_{m_j} \in
	\mathcal{A}^* $. Notice that no \(X_{i_{\ell}}\) has $y_{m_j}$ for ancestor, and that $y_{m_j}
  \xRightarrow{\overline{\gamma}_j} w_{m_j}$ must be a depth-first
	derivation because $X \xRightarrow{\gamma} w$ is. Otherwise, if $y_{m_j} =
	\varepsilon$, let $\overline{\gamma}_j = \varepsilon$. Let $\gamma^\sharp =
	(\gamma)_{\ell_0} \cdot \overline{\gamma}_0 \cdot
  (\gamma)_{\ell_1} \cdot \overline{\gamma}_1 \cdots
  (\gamma)_{\ell_{p{-}1}} \cdot \overline{\gamma}_{p{-}1}$. Observe
  that, since each $y_{m_j} \xRightarrow{\overline{\gamma}_j} w_{m_j}$
  is a depth-first derivation, we have $X_{i_{j+1}}^{\tuple{b}}
  y_{m_j}^{\tuple{b}} \xRightarrow{\overline{\gamma}_j}
  X_{i_{j+1}}^{\tuple{b}} w_{m_j}^{\boldsymbol{\alpha}}$ (or with
  \(X_{i_{j+1}}\) and \(y_{m_j}\) swapped) is a depth-first step
  sequence because \( y_{m_j} \) and \( X_{i_{j+1}} \) have the same
  rank \(b\).  Clearly, $\gamma^\sharp$ corresponds to a valid
  step sequence of $G$ which, moreover, is depth first, since whenever
  $ (\gamma)_{\ell_j} $ fires, \(X_{i_j}\) is the only nonterminal left
  (and whose rank is therefore maximal). It follows from \(
  (\ref{item:easyfact4})\) that because \(X
  \xRightarrow{\gamma^\sharp}_G u \, X_{i_p} \, v\) holds and
  \(X, X_{i_p}\in \Varse\) then \(X
  \xRightarrow{\gamma^\sharp}_{G^\sharp} u \, X_{i_p} \, v\)
  holds (notice the use of \(G^\sharp\) instead of \(G\)).  Moreover,
  the definition of \(\gamma^\sharp\) shows that $X
  \xRightarrow{\gamma^\sharp}_{G^\sharp} u \, X_{i_p} \, v$ is a
  depth-first step sequence and $u, v \in \mathcal{A}^* $.

  Since $X \xRightarrow{\gamma}_G w$ is a $k$-index derivation, each
  step sequence \(y_{m_j} \xRightarrow{\overline{\gamma_j}} w_{m_j}\)
  are of index at most $k$.  Therefore the index of each step sequence
  \(X_{i_{j+1}} y_{m_j} \xRightarrow{\overline{\gamma_j}} X_{i_{j+1}}
  w_{m_j}\) (or in reverse order) is at most \(k+1\). Also, when each
  $ (\gamma)_{\ell_j} $ fires, \(X_{i_j}\) is the only nonterminal left
  and so the index of the step is at most $2$. Therefore we find that
  \(X \xRightarrow[(k+1)]{\gamma^\sharp} u \, X_{i_p} \, v\),
  and finally that \( X \xRightarrow[\df{k+1}]{\gamma^\sharp} u\,
  X_{i_p}\, v\) in $G^\sharp$.
  
  \item Assume that $y,z \in \Varsi$ (the cases $y = \varepsilon$ or $z = \varepsilon$
    are similar).  Since \(\gamma\) of length \(n\) induces a
    $k$-index depth first derivation, we have that
		$y\, z \xRightarrow[\df{k}]{(\gamma)_{\ell_{p}+1} \ldots (\gamma)_n} u_y\, u_z \in \mathcal{A}^*$ can be split
		into two derivations of \(G\) as follows: $y \xRightarrow[\df{k_y}]{\gamma_{y}} u_y$ and $z \xRightarrow[\df{k_z}]{\gamma_{z}} u_z$ such that
		\(\max(k_z,k_y)\leq k\) and \(\min(k_z,k_y)\leq k-1 \) (see Lem.~\ref{lem:leibniz} for a proof).
		Assume \(k_y \leq k-1\), the other case being symmetric.
		Since the only production in
		$(\gamma)_{\ell_{p}}\cdots (\gamma)_n$ whose left hand
    side is a nonterminal from $\Varse$ is $ (\gamma)_{\ell_p} =(X_{i_p},
    a\, y\, b\, z)$, which, moreover, occurs only in the first position, we have
		that $\gamma_y \in \Gamma_y^{\df{k-1}}(G_{i_p,aybz})$ and $\gamma_z \in
    \Gamma_z^{\df{k}}(G_{i_p,aybz})$, by the definition of $G_{i_p,aybz}$.

	\item It suffices to notice that \(\gamma^\sharp\cdot (\gamma)_{\ell_p}\cdots (\gamma)_n\) results from reordering the productions of \(\gamma\)
    and that reordering the productions of \(\gamma\) result into a step sequence producing the same word \(w = a_1^{i_1}\ldots
    a_d^{i_d}\) since \(L_X(G)\subseteq \patt\) where \(\patt\) is a strict \(d\)-letter bounded expression.
		That the resulting derivation has index \(k\) and is depth-first follow easily from \((\ref{item1:ginsbook-surgery})\) and  \((\ref{item2:ginsbook-surgery})\).

  \item Given that \(\prod^\sharp \subseteq \prod\) we find that \( X
    \xRightarrow{}^*_{G^\sharp} u\, X_{i_p}\, v \) implies \( X
    \xRightarrow{}^*_{G} u\, X_{i_p}\, v \), hence \( Q(u,v) \) holds by
    \( (\ref{item:easyfact2}) \) and \( X, X_{i_p} \in \Varse \). By the
    definition of \( Q(u,v) \), we have: 
    \[\begin{array}{rcl}
    \set{u' \in (\Vars \cup \mathcal{A})^* \mid u\xRightarrow{}^* u' } & \subseteq & ( \set{a_1}
    \cup \Varsi)^* ~\mbox{and} \\ 
    \set{v' \in (\Vars\cup\mathcal{A})^* \mid v \xRightarrow{}^* v' } & \subseteq & 
    ( \set{a_d} \cup \Varsi)^*
    \end{array}\] 
    Since \(G\) is reduced, $\set{u' \in \mathcal{A}^*
      \mid u \xRightarrow{}^* u' } \neq \emptyset$ and $\set{v' \in
      \mathcal{A}^* \mid v \xRightarrow{}^* v' } \neq \emptyset$.  But
    because $X_{i_p} \in \Varse$, it must be the case that $\set{u' \in
      \mathcal{A}^* \mid u \xRightarrow{}^* u' } \subseteq a_1^*$ and
    $\set{v' \in \mathcal{A}^* \mid v \xRightarrow{}^* v' } \subseteq
    a_d^*$, otherwise we would contradict the fact that
    \(L_{X}(G)\subseteq\patt\).

  \item Since $X \xRightarrow{\gamma^\sharp}_G u \, X_{i_p} \, v
    \xRightarrow{(X_{i_p}, a\, y\, b\, z)}_G u \, a\, y \, b\, z \, v$ and $G$
    is reduced, we have that $\set{u' \in \mathcal{A}^*
      \mid u \xRightarrow{}^*_G u'} \cdot a \cdot L_y(G) \cdot b \cdot L_z(G) \cdot
    \set{v' \in \mathcal{A}^* \mid v \xRightarrow{}^*_G v'} \subseteq
    L_X(G) \subseteq \patt$, and thus $L_y(G) \cdot L_z(G) \subseteq
    \patt$. We consider only the case $y,z \in \Varsi$---the cases
    $y=\varepsilon$ or $z=\varepsilon$ use similar arguments, and are left as an
    easy exercise. Hence, our proof falls into 4 cases:
    \begin{compactenum}[\upshape(\itshape a\upshape)]
    \item \(L_y(G_{i,a\, y\, b\, z})\cap (a_1\cdot\mathcal{A}^*) = \emptyset \)
      and \( L_z(G_{i,a\, y\, b\, z}) \cap (a_1\cdot\mathcal{A}^*) = \emptyset \).
      Thus \( L_y(G_{i,a\, y\, b\, z})\cdot L_z(G_{i,a\, y\, b\, z}) \subseteq a_2^* \ldots
			a_d^*\). Then fact {\upshape(\itshape\ref{item:easyfact5}\upshape)} for \(\ell=2\) and
      \(r=d\) concludes this case.
    \item \(L_y(G_{i,a\, y\, b\, z})\cap (\mathcal{A}^*\cdot a_d) = \emptyset \)
      and \( L_z(G_{i,a\, y\, b\, z}) \cap (\mathcal{A}^*\cdot a_d) = \emptyset
      \).  Thus \(L_y(G_{i,a\, y\, b\, z}) \cdot L_z(G_{i,a\, y\, b\, z}) \subseteq a_1^*
			\ldots a_{d-1}^*\). Then fact {\upshape(\itshape\ref{item:easyfact5}\upshape)} for
      \(\ell=1\) and \(r=d-1\) concludes this case.
    \item \(L_y(G_{i,a\, y\, b\, z})\cap (\mathcal{A}^*\cdot a_d) = \emptyset \)
      and \( L_z(G_{i,a\, y\, b\, z}) \cap (a_1\cdot \mathcal{A}^*) = \emptyset
      \). Thus we have \( L_y(G_{i,a\, y\, b\, z}) \subseteq a_1^* \ldots
      a_{d-1}^* \) and \( L_z(G_{i,a\, y\, b\, z}) \subseteq a_2^* \ldots
			a_{d}^*\).  By the fact {\upshape(\itshape\ref{item:easyfact5}\upshape)} (with \(\ell=1\),
      \(r=d\)) there exists \(q\), \(1\leq q\leq d\) such that
      \(L_y(G_{i,a\, y\, b\, z}) \subseteq a_1^* \ldots a_q^*\) and \(L_z(G_{i,a\, y\, b\, z})
      \subseteq a_q^* \ldots a_d^*\). Next we show \( 1 < q < d \)
      holds.  In fact, assume the inclusions hold for \(q=1\). Then
      they also hold for \(q=2\) since \(L_z(G_{i,a\, y\, b\, z}) \subseteq a_2^*
      \ldots a_d^*\).  A similar reasoning holds when \(q=d\) since
      \(L_y(G_{i,a\, y\, b\, z}) \subseteq a_1^* \ldots a_{d-1}^* \).
    \item \(L_y(G_{i,a\, y\, b\, z})\cap (a_1\cdot \mathcal{A}^*) = \emptyset \)
      and \( L_z(G_{i,a\, y\, b\, z}) \cap (\mathcal{A}^*\cdot a_d) = \emptyset
      \).  We first observe that it cannot be the case that
      \(L_y(G_{i,a\, y\, b\, z})\) contains some word where \(a_d\) occurs and
      \(L_{z}(G_{i,a\, y\, b\, z})\) contains some word where \(a_1\) occurs for
      otherwise concatenating those two words shows \(
      L_y(G_{i,a\, y\, b\, z})\cdot L_z(G_{i,a\, y\, b\, z}) \nsubseteq a_1^* \ldots
      a_d^*\). This leaves us with three cases:
      \begin{inparaenum}[\upshape(\itshape a\upshape)]
      \item If \(L_y(G_{i,a\, y\, b\, z}) \cap (\mathcal{A}^*\cdot a_d) \neq \emptyset\)
        we find that \(L_z(G_{i,a\, y\, b\, z}) \subseteq a_d^*\), hence that \(
        L_y(G_{i,a\, y\, b\, z}) \subseteq a_2^* \ldots a_{d}^* \) since
        \(L_y(G_{i,a\, y\, b\, z}) \cap (a_1\cdot\mathcal{A}^*)=\emptyset\).
      \item If \(L_z(G_{i,a\, y\, b\, z}) \cap (a_1\cdot \mathcal{A}^*) \neq
        \emptyset\) we find that \(L_y(G_{i,a\, y\, b\, z}) \subseteq a_1^*\),
        hence that \( L_z(G_{i,a\, y\, b\, z}) \subseteq a_1^* \ldots a_{d-1}^*\)
        since \( L_z(G_{i,a\, y\, b\, z}) \cap (\mathcal{A}^*\cdot a_d) =
        \emptyset\).
      \item Then \(L_y(G_{i,a\, y\, b\, z}) \cap (\mathcal{A}^* \cdot a_d ) =
        \emptyset\) and \( L_z(G_{i,a\, y\, b\, z}) \cap (a_1\cdot \mathcal{A}^*)
        = \emptyset\). Hence \( L_y(G_{i,a\, y\, b\, z})\cdot L_z(G_{i,a\, y\, b\, z})
        \subseteq a_2^* \ldots a_{d-1}^*\) and by the fact
				{\upshape(\itshape\ref{item:easyfact5}\upshape)} for \(\ell=2\) and \(r=d-1\) there
        exists \(1 < q < d\) such that \(L_y(G_{i,a\, y\, b\, z}) \subseteq a_2^*
        \ldots a_q^* \) and \(L_z(G_{i,a\, y\, b\, z}) \subseteq a_q^* \ldots
        a_{d-1}^*\).
	\end{inparaenum}\qed
\end{compactenum}
\end{compactenum} 
\end{proof} 

\subsection{Proof of Theorem~\ref{thm:letter-bounded-control-set}}
\begin{proof}[of Theorem~\ref{thm:letter-bounded-control-set}]
We prove the theorem by induction on $d > 0$. If $d=1,2$, we obtain
$\pattg$ from Lemma~\ref{lem:ginsbook-d}, and time needed to compute
$\pattg$, using Algorithm \ref{alg:constant-control-set}, is
$\len{G}^{\mathcal{O}(k)}$. Moreover, we have $L_X^{(k)}(G) =
\hat{L}_X(\pattg \cap \Gamma_{X}^{\df{k}}, G) \subseteq
\hat{L}_X(\pattg \cap \Gamma_{X}^{\df{k+1}}, G)$.

For the induction step, assume $d \geq 3$. W.l.o.g. we assume that $G$
is reduced for $X$, and that $a_1^* \ldots a_d^*$ is the minimal
bounded expression such that $L_X(G) \subseteq a_1^* \ldots
a_d^*$. Consider the partition $\Varse \cup \Varsi = \Vars$ and
$\Varse \cap \Varsi = \emptyset$, defined in the previous. Since $G$
is reduced for $X$, then $X \in \Varse$. Define 
\[\prod_{\mathit{pivot}} =
\set{(X_i,a\, y\, b\, z)\in\prod \mid X_i \in \Varse ~\mbox{and}~a,b \in \mathcal{A}\cup\set{\varepsilon},\ y,z \in \Varsi \cup \set{\varepsilon}}\enspace .\]

By Lemma~\ref{lem:ginsbook-d}, for each $X_i \in \Vars$, such that
$L_{X,X_i}(G) \subseteq a_1^*a_d^*$, there exists a bounded expression
$\Gamma_{1,d}^{X,X_i}$ such that
$L_{X,X_i}^{(k+1)} = \hat{L}_{X,X_i}(\Gamma_{1,d}^{X,X_i} \cap \Gamma_{X,X_i}^{\df{k+1}},
G)$. Moreover, by the induction hypothesis, for each $\ell, m, r$ such
that $1 \leq \ell \leq m \leq r \leq d$, $m - \ell < d - 1$ and $r - m
< d - 1$, and for each $Y,Z \in \Vars$ such that $L_Y(G) \subseteq
a_\ell^* \ldots a_m^*$ and $L_Z(G) \subseteq a_m^* \ldots a_r^*$,
there exist two sets $\mathcal{S}^Y_{\ell \ldots m}, \mathcal{S}^Z_{m
\ldots r}$ of bounded expressions over \(\Delta_{i,aybz}\) such that
$L_Y^{(k)}(G) \subseteq
\hat{L}_Y(\bigcup\mathcal{S}^Y_{\ell \ldots m} \cap \Gamma_{Y}^{\df{k+1}}, G)$ and
$L_Z^{(k)}(G) \subseteq \hat{L}_Z(\bigcup\mathcal{S}^Z_{m \ldots r} \cap
\Gamma_{Z}^{\df{k+1}}, G)$. We extend this notation to \(\varepsilon\), and assume that 
\(\mathcal{S}^{\varepsilon}_{i \ldots j} = \set{\varepsilon}\). We define: 
\begin{align*}
\mbox{IH} & = \{(\ell,m,r) \mid 1 \leq \ell \leq m \leq r \leq d,~m - \ell < d - 1 \land r - m < d - 1\} \\
\mathcal{S}_\patt & = \{ \Gamma_{1,d}^{X,X_i} \cdot (X_i,a\, y\, b\, z)^* \cdot \Gamma' \cdot \Gamma'' 
\mid (X_i,a\, y\, b\, z) \in \prod_{\mathit{pivot}} \land \\
&\quad\quad L_{X,X_i}(G) \subseteq a_1^* a_d^* \land \Gamma'\in\mathcal{S}^y_{\ell \ldots m} \land \Gamma''\in\mathcal{S}^z_{m \ldots r} \land (\ell,m,r) \in IH\}
\end{align*}

First, let us prove that
$L_X^{(k)}(G) \subseteq \hat{L}_X(\bigcup\mathcal{S}_\patt \cap
\Gamma_{X}^{\df{k+1}}, G)$. Let $w \in L_X^{(k)}(G)$ be a word, and $X
\xArrow{\gamma}{\df{k}} w$ be a $k$-index depth first derivation
of $w$ in $G$. Since $w \in L_X^{(k)}(G)$, such a derivation is
guaranteed to exist. By Lemma~\ref{lem:ginsbook-surgery}, there
exists $(X_i, a\, y\, b\, z) \in \prod_{pivot}$, and $\gamma^\sharp \in
(\prod^\sharp)^*$, $\gamma_{y},\gamma_{z} \in (\Delta_{i,aybz})^*$, such
that \(\gamma^\sharp \cdot (X_i,a\, y\, b\, z) \cdot \gamma_y \cdot \gamma_z\) is
a permutation of \(\gamma\), and:
\begin{compactitem}
	\item \(X \xArrow{\gamma^\sharp}{\df{k+1}} u\, X_i\, v\) is a
		step sequence of \({G^\sharp}\) with \(u, v\in \mathcal{A}^*\);

	\item \(y \xArrow{\gamma_{y}}{\df{k_y}} u_y\) and \(z
		\xArrow{\gamma_{z}}{\df{k_z}} u_z\) are derivations of
		\(G_{i,aybz}\) (hence  \(u_y\, u_z \in \mathcal{A}^* \)), 
                \(\max(k_y,k_z)\leq k\) and \(\min(k_y,k_z)\leq k-1\);

	\item \(X \xArrow{\gamma^\sharp \cdot (X_i,aybz) \cdot \gamma_y \cdot \gamma_z}{\df{k+1}} w\) 
              is a derivation of \({G^\sharp}\) if \(y \xArrow{\gamma_y}{\df{k-1}} u_y\) is a derivation of \({G_{i,aybz}}\); 
        
        \item \(X \xArrow{\gamma^\sharp \cdot (X_i,aybz) \cdot \gamma_z \cdot \gamma_y}{\df{k+1}} w\) 
              is a derivation of \({G^\sharp}\) if \(z \xArrow{\gamma_z}{\df{k-1}} u_z\) is a derivation of \({G_{i,aybz}}\); 

	\item \(L_{X,X_i}(G^\sharp) \subseteq a_1^*a_d^*\);

	\item \(L_{y}(G_{i,aybz}) \subseteq a_{\ell}^* \ldots a_{m}^*\) if \(y \in \Varsi\); 
              \(L_{z}(G_{i,aybz}) \subseteq a_{m}^* \ldots a_{r}^*\) if \(z \in \Varsi\), with
		\(1\leq \ell \leq m \leq r \leq d\), such that \( m-\ell < d - 1\)
		and \( r-m < d - 1\).
\end{compactitem}
Let us consider the case where $y,z \in \Vars$ (the other cases of
$y=\varepsilon$ or $z=\varepsilon$ being similar, are left to the reader). 
We also assume \(k_y \leq k-1\) the other case being symmetric.

Therefore, by the induction hypothesis there exist bounded expressions
\(\Gamma' \in \mathcal{S}^y_{\ell \ldots m}\) and \(\Gamma'' \in \mathcal{S}^z_{m \ldots r}\) 
such that
$y \xArrow{\gamma'}{\df{k_y+1}} u_y$ and $z
\xArrow{\gamma''}{\df{k_z+1}} u_z$, for some control words 
\(\gamma'\in\Gamma'\) and \(\gamma''\in\Gamma''\). 
If $L_{X,X_i}(G^\sharp) \subseteq a_1^*a_d^*$, by Lemma~\ref{lem:ginsbook-d}, 
there exists a control word $\gamma^\sharp \in \Gamma_{1,d}^{X,X_i}$ such that 
\(X \xArrow{\gamma^\sharp}{\df{k+1}} u \, X_i \, v\) is a
\((k+1)\)-index depth first step sequence in \(G^\sharp\). 
It follows that:
\[
X \xArrow{\gamma^\sharp}{\df{k+1}} u \, X_i \, v 
\xArrow{(X_i,aybz)}{} u\, a\, y\, b\, z\, v
\xArrow{\gamma'}{\df{k_y+2}} u\, a\, u_y\, b\, z\, v
\xArrow{\gamma''}{\df{k_z+1}} u\, a\, u_y\, b\, u_z\, v = w\enspace .
\]
Observe that \( u\, a\, y\, b\, z\, v \xArrow{\gamma'}{\df{k_y+2}}
u\, a\, u_y\, b\, z\, v\)
because \(a,b,u,v \in \mathcal{A}^*\), \(z \in \Vars\) and \(
y \xArrow{\gamma'}{\df{k_y+1}} u_y\).  Since \(k_y\leq k-1\)
and \(k_z\leq k\), we find that \(k_y+2 \leq k+1\) and \(k_z+1\leq
k+1\), respectively.  Hence the overall index of the foregoing
derivation with control word \( (\gamma^\sharp\,
(X_i,aybz)\, \gamma'\, \gamma'')\) is at most \(k+1\).
Since it is also a depth-first derivation, we finally find
that \(w\in \hat{L}_{X}(\bigcup\mathcal{S}_\patt\cap \Gamma_X^{\df{k+1}}, G)\), i.e.\ 
\(L^{(k)}_X(G) \subseteq \hat{L}_{X}(\bigcup\mathcal{S}_\patt\cap \Gamma_X^{\df{k+1}}, G)\).

In the following, we address the time complexity of the construction
of $\mathcal{S}_\patt$, and of each bounded expression
$\Gamma\in\mathcal{S}_\patt$. We refer to
Algorithm~\ref{alg:bounded-control-set} in the following. Notice first
that both the \Call{minimizeExpression}{}
and \Call{partitionNonterminals}{} functions take time
$\mathcal{O}(\len{G})$, because emptiness of the intersection between
a context-free grammar and a finite automaton of constant size is linear in
the size of the grammar~\cite[Section~5]{BEF+ipl}. Moreover, the inclusion check on
(line~\ref{line:inclusion}) is possible also in time
$\mathcal{O}(\len{G})$ (see Remark\ref{rem:letter-bounded-inclusion}). By
Lemma~\ref{lem:ginsbook-d}, a call to
$\Call{ConstantBoundedControlSet}{G,\pat,k}$ will take time
$\len{G}^{\mathcal{O}(k)}$. Lemma~\ref{lem:ginsbook-surgery} shows
that the sizes of the bounded expression considered at
lines~\ref{line:rec-y} and \ref{line:rec-z}, in a recursive call, sum
up to the size of the bounded expression for the current
call. Thus the total number of recursive calls is at most \(d\).  We
thus let $T(d)$ denote the time needed for the top-level call of the
function $\Call{LetterBoundedControlSet}{G,X,a_1^* \ldots a_d^*,k}$ to
complete. Since the loop on
(lines~\ref{line:for-begin}--\ref{line:for-end}) will be taken at most
$\card{\prod}\leq\len{G}$ times, we obtain:
\[ T(d) = \len{G}^{\mathcal{O}(k)} + \len{G}( \mathcal{O}(\len{G}) + 2\, T(d{-}1))\]
where $2\, T(d{-}1)$ is the time needed for the two recursive calls at
lines \ref{line:rec-y} and \ref{line:rec-z} to complete.
Because \(T(0)=\mathcal{O}(\len{G}) + \len{G}^{\mathcal{O}(k)} \), we
find that $T(d) = \len{G}^{\mathcal{O}(k)+d}$ .

Finally, the time needed to build each bounded
expression \(\Gamma\in\mathcal{S}_\patt\) can be evaluated by
observing that each such expression is uniquely determined by a
sequence \(\sigma\in\prod^*\) of productions of \(G\) that are
successively chosen at line \ref{line:for-begin}. Let us consider now
a slightly modified version of Algorithm \ref{alg:bounded-control-set}
that is guided by a sequence \(\sigma\in\prod^*\) received in input
--- the function \(\Call{LetterBoundedControlSet}{G,X,a_s^* \ldots
a_t^*,k,\sigma}\) receives an extra parameter and returns also the
suffix of \(\sigma\) that remains after processing the first
production on \(\sigma\), i.e.\ the recursive calls at
lines \ref{line:rec-y} and \ref{line:rec-z} have returned.  Since the
sum of sizes of the bounded expressions for these recursive
calls is at most \(t-s\), by Lemma~\ref{lem:ginsbook-surgery}, we
obtain that, in total, Algorithm \ref{alg:bounded-control-set}
initiates at most \(d\) calls to
\(\Call{LetterBoundedControlSet}{}\). We recall also that the
prefix of each call (before making recursive calls) takes
time \(\mathcal{O}(\len{G})
+ \len{G}^{\mathcal{O}(k)}\). Since \(L_X(G) \subseteq \patt\),
assuming that \(\patt\) is minimal, we
have \(\len{\patt} \leq \len{G}\). Hence, the time needed to compute a
bounded expression \(\Gamma\in\mathcal{S}_\patt\) is bounded
by: \[d \cdot (\mathcal{O}(\len{G})
+ \len{G}^{\mathcal{O}(k)}) \leq \len{G}\cdot (\mathcal{O}(\len{G})
+ \len{G}^{\mathcal{O}(k)}) = \len{G}^{\mathcal{O}(k)}\enspace.\]\qed
\end{proof}

\subsection{Proof of Lemma~\ref{lem:optimality}}
\begin{proof}[of Lemma~\ref{lem:optimality}]
Given \(k > 0\), consider the
following grammar: 
\[G = \tuple{ \{ X_i \mid 0\leq i\leq k \}, \{ a\}, \{ X_i \rightarrow
X_{i-1}\, X_{i-1} \mid 1\leq i \leq k \}\cup \{ X_0 \rightarrow a \}
}\enspace .\] Notice that \(L_{X_k}(G) = \{ a^{2^k} \} \subseteq a^*\)
and \(\len{G} = \mathcal{O}(k)\). Moreover, every depth-first
derivation of \(G\) has index \(k+1\).

For each \(i \in \set{1, \ldots, n}\), let \(p_i\) be the
production \(X_i \rightarrow X_{i-1}\, X_{i-1}\) of \(G_n\), and
let \(p_0\) be \(X_0 \rightarrow a\). It is easy to see that, because
the derivation is depth-first, the control word \(\gamma\)
generating \(a^{2^k}\) from \(X_k\) is unique. Now suppose that there
exists \(\Gamma = w_1^* \ldots w_d^*\) such that \(\gamma =
w_1^{i_1} \ldots w_d^{i_d}\), for some \(i_1, \ldots, i_d \geq
0\). Next we show that, for all \(j=1,\ldots,d\) we must
have \(i_j \leq 2\).

We first make this crucial observation, since the derivation tree is
binary and its traversal is depth-first, we have that for
every \(p_i\), every three consecutive occurrences
\(\ell_1<\ell_2<\ell_3\) of \(p_i\)---\( (\gamma)_{\ell_1}=(\gamma)_{\ell_2}=(\gamma)_{\ell_3}=p_i\)---implies that there exists a position
\(\ell\) between \(\ell_1\) and \(\ell_3\) such that \( (\gamma)_{\ell} = p_{i+1} \).
Otherwise that would imply that the derivation tree has a node \(X_{i+1}\) with three \(X_{i}\) children;
or that the tree was not traversed in depth-first.

Take an arbitrary \(w_j\) in \(\Gamma\) and let \(g\) be the greatest
index of a production occurring in \(w_j\).  The number \(i_j\) of
repetitions of \(w_j\) cannot be greater than two for
otherwise \(p_g\) contradicts the previous fact.  So this concludes
that no \(i_j\) can be larger than \(2\).

Now, since the only string of \(L_{X^k}(G)\) has length \(2^k\) and that
no rule produces more than one terminal then
necessarily \(\len{\gamma} \geq 2^k\). So we show that
\(\len{\Gamma}\) has to be at least \(2^{k-1}\). By contradiction, suppose
\(\len{\Gamma}\leq (2^{k-1}-1)\), then since in order to capture \(\gamma\) no
word of \(\Gamma\) can occur more than twice, the longest control word that
\(\Gamma\) can capture is \(2 \cdot (2^{k-1}-1) = 2^{k}-2\) which is shorter
than \(2^k=\len{\gamma}\), hence a contradiction.\qed
\end{proof}

\subsection{Proof of Theorem \ref{thm:fo-reachability}}\label{app:fo-reachability}

\begin{proof}[of Theorem \ref{thm:fo-reachability}]
The \textsc{Np}-hard lower bound is by reduction from the Positive
Integer Linear Programming (PILP) problem, which is known to
be \textsc{Np}-complete \cite[Corollary 18.1a]{schrijver}.  Consider
the following instance of PILP, with variables $k_1, \ldots, k_m$
ranging over positive integers:

  \[\left\{\begin{array}{lcl}
  a_{11} \cdot k_1 + \ldots + a_{m1} \cdot k_m + c_1 & \leq & 0 \\
  & \cdots & \\
  a_{1n} \cdot k_1 + \ldots + a_{mn} \cdot k_m + c_n & \leq & 0 
  \end{array}\right.\]
  and denote $\vec{a}_i = \langle a_{i1}, \ldots, a_{in} \rangle \in \zed^n$,
  for all $i = 1, \ldots, m$, and $\vec{c} = \langle c_1, \ldots,
  c_n \rangle \in \zed^n$. Let $\vec{x}=\set{x_1,\ldots,x_n}$ be a set of integer variables.
  Consider the program $\mathcal{P}_{\mathrm{PILP}}=\tuple{G,X_0,\sem{.}}$, where
  $G = \tuple{\Vars,\Sigma,\prod}$:
  \begin{compactitem}
  \item $\Vars=\set{X_0,\ldots,X_{m+1}}$,
  \item $\Sigma=\set{\tau_i \mid i=0,\ldots,m+1} \cup \set{\lambda_i \mid i=0,\ldots,m}$,
  \item $\prod = \set{X_i \rightarrow \tau_i\, X_{i+1} \mid i=0,\ldots,m} \cup 
  \set{X_i \rightarrow \lambda_i\, X_i \mid i=1,\ldots,m} \cup 
  \set{X_{m+1} \rightarrow \tau_{m+1}}$, 
  \item the semantics of the words \(w\in L_{X_0}(G)\) is defined by the following relations:
  \[\begin{array}{lcll}
  \rho_{\tau_0} & \equiv & \vec{x}'=0 \\
  \rho_{\tau_i} & \equiv & \vec{x}'=\vec{x} & \quad \mbox{for all $i=1,\ldots,m-1$} \\
  \rho_{\lambda_i} & \equiv & \vec{x}'=\vec{x}+\vec{a}_i &\quad \mbox{for all $i=1,\ldots,m$} \\ 
  \rho_{\tau_m} & \equiv & \vec{x}'=\vec{x}+\vec{c} \\ 
  \rho_{\tau_{m+1}} & \equiv & \vec{x} \leq \vec{0} 
  \end{array}\]
  \end{compactitem}
  Let
  $\patt_{\mathrm{PILP}}=\tau_0^* \lambda_1^* \tau_1^* \ldots \lambda_m^* \tau_m^* \tau_{m+1}^*$
  be a bounded expression. It is immediate to check that the PILP
  problem has a solution if and only
  if \(\foreach(\mathcal{P}_{\mathrm{PILP}}, \patt_{\mathrm{PILP}})\)
  holds.  This settles the \textsc{Np}-hard lower bound for the class
  of fo-reachability problems.

  We show next that the class of fo-reachability
  problems \(\foreach(\mathcal{P}, \pat)\) is included
  in \textsc{Nexptime}. Let \(\mathcal{P} = \tuple{G,I,\sem{.}}\) be a
  given program, where \(G=\tuple{\Vars,\Sigma,\prod}\) is its
  underlying grammar, and \(\pat = w_1^* \ldots w_d^*\) a bounded
  expression. By Lemma \ref{lem:intersection}, there exists a
  grammar \(G^\cap = \tuple{\Vars^\cap,\Sigma,\prod^\cap}\) such
  that: \[\bigcup_{1\leq s\leq x\leq d} L_{[\textsc{q}^{(s)}_1
  I \textsc{q}^{(x)}_{1}]}(G^{\cap}) = L_I(G) \cap \pat\enspace.\]
  Moreover, we have that $\len{G^\cap}
  = \mathcal{O}(\len{\pat}^3 \cdot \len{G})$.
  Let \(\mathcal{P}_{s,x}=\tuple{G^\cap, [\textsc{q}^{(s)}_1
  I \textsc{q}^{(x)}_{1}], \sem{.}}\) be a program, for each \(1 \leq
  s \leq x \leq d\). Since the alphabets of \(G\) and \(G^\cap\)
  coincide, the mapping of symbols to octagonal relations is the same
  for \(G\) and \(G^\cap\), hence: \[\bigcup_{1\leq s\leq x\leq
  d} \sem{\mathcal{P}_{s,x}} = \sem{\mathcal{P}}_\pat\enspace.\]
  Then \(\sem{P}_\pat \neq \emptyset\) if and only
  if \(\sem{\mathcal{P}_{s,x}} \neq \emptyset\), for some \(1 \leq
  s \leq x \leq d\). We have reduced the original
  problem \(\foreach(\mathcal{P},\pat)\)
  to \(\mathcal{O}(\len{\pat}^2)\) reachability problems, of
  size \(\mathcal{O}(\len{\pat}^3 \cdot \len{G})\) each. In the
  following we fix \(1 \leq s \leq x \leq d\), focus w.l.o.g on the
  problem \(\foreach(\mathcal{P}_{s,x},\pat)\) and we denote by \(X =
  [\textsc{q}^{(s)}_1 I \textsc{q}^{(x)}_{1}]\) in the rest of this
  proof.

  Let \(\mathcal{A} = \set{a_1,\ldots,a_d}\) be an alphabet disjoint
  from \(\Sigma\) and \(\patt=a_1^* \ldots a_d^*\) be a strict
  letter-bounded expression, such that \(\pat=h(\patt)\),
  where \(h: \mathcal{A} \rightarrow \Sigma^*\) is the
  homomorphism \(h(a_i) = w_i\), for all \(i=1,\ldots,d\). By
  Lemma \ref{lem:interface} there exists a
  grammar \(G^\bowtie=\tuple{\Vars^\cap, \mathcal{A}, \prod^\bowtie}\)
  such that, for every \(k >
  0\): \begin{compactenum} 

  \item \(L_X^{(k)}(G^{\bowtie}) =
        h^{-1}(L_X^{(k)}(G^\cap)) \cap \patt\),

  \item for each \(\Gamma \subseteq \left(\prod^\bowtie\right)^*\),
  such that \(L_X^{(k)}(G^\bowtie) \subseteq \hat{L}_X(\Gamma,
  G^\bowtie)\), we
  have \(L_X^{(k)}(G^\cap) \subseteq \hat{L}_X(\iota^{-1}(\Gamma),
  G^\cap)\).  \end{compactenum} Moreover, we have \(\len{G^\bowtie}
  = \mathcal{O}(\len{\pat}^3 \cdot \len{G})\).
  Since \(L_X^{(k)}(G^{\bowtie}) \subseteq \patt\), by
  Theorem \ref{thm:letter-bounded-control-set}, there exists a
  set \(\mathcal{S}_\patt\) of bounded expressions
  over \(\prod^\bowtie\) such
  that: \[L_X^{(k)}(G^\bowtie) \subseteq \hat{L}_X\left(\bigcup\mathcal{S}_\patt \cap \Gamma_X^{\df{k+1}}(G^\bowtie),
  G^\bowtie\right)\enspace.\] Hence, by Lemma \ref{lem:interface}, we
  obtain: \[L_X^{(k)}(G^\cap) \subseteq \hat{L}_X\left(\iota^{-1}\left(\bigcup\mathcal{S}_\patt\right) \cap \Gamma_X^{\df{k+1}}(G^\cap),
  G^\cap\right)\enspace.\] We used the fact
  that \(\iota^{-1}(\Gamma_X^{\df{k+1}}(G^\bowtie))
  = \Gamma_X^{\df{k+1}}(G^\cap)\).
  Because \(L_X(G^\cap) \subseteq \pat\), there
  exists \(K=\mathcal{O}(\len{G^\cap})\) such that \(L_X(G^\cap) =
  L_X^{(K)}(G^\cap)\) as Theorem~\ref{thm:luker} shows.
  Hence \(K=\mathcal{O}(\len{\pat}^3 \cdot \len{G})\) as well. We
  obtain the following: \[L_X(G^\cap) \subseteq
  L_X^{(K)}(G^\cap) \subseteq \hat{L}_X\left(\iota^{-1}\left(\bigcup\mathcal{S}_\patt\right) \cap \Gamma_X^{\df{K+1}}(G^\cap),
  G^\cap\right) \subseteq L_X(G^\cap)\] thus, \(L_X(G^\cap)
  = \hat{L}_X\left(\iota^{-1}\left(\bigcup\mathcal{S}_\patt\right) \cap \Gamma_X^{\df{K+1}}(G^\cap),
  G^\cap\right)\). Assume that \(\mathcal{S}_\patt
  = \set{\Gamma_1, \ldots, \Gamma_m}\), for some \(m>0\), and
  denote \(\iota^{-1}(\Gamma_i)\) by \(\widetilde{\Gamma}_i\). We have
  that, for each derivation \(X \xRightarrow[\df{k+1}]{\gamma} w\)
  of \(G^\cap\), \(\sem{w}=\emptyset\)
  if{}f \(\sem{\gamma}=\emptyset\) \cite[Lemma 2]{gik13}.  As a
  result, \(\sem{\mathcal{P}_{s,x}} \neq \emptyset\) if{}f there
  exists \(i=1,\ldots,m\)
  and \(\gamma \in \widetilde{\Gamma_i} \cap \Gamma_X^{\df{k+1}}(G^\cap)\),
  such that \(\sem{\gamma}\neq\emptyset\). By
  Theorem \ref{thm:letter-bounded-control-set}, each set \(\Gamma_i\)
  can be constructed in time: \[\len{G^\bowtie}^{\mathcal{O}(K)} =
  (\len{\pat}^3 \cdot \len{G})^{\mathcal{O}(K)} =
  (\len{\pat}^3 \cdot \len{G})^{\mathcal{O}(\len{\pat}^3 \cdot \len{G})}
  = 2^{\mathcal{O}(\len{\pat}^3 \cdot \len{G} \cdot (\log \len{\pat}
  + \log \len{G}))}\enspace.\] We have used the
  facts \(\len{G^\bowtie} = \mathcal{O}(\len{\pat}^3 \cdot \len{G})\)
  and \(K = \mathcal{O}(\len{\pat}^3 \cdot \len{G})\).

  By Lemma \ref{fsa-dfk}, there exists a finite
  automaton \(A^{\df{K+1}}_{G^\cap}\) that recognizes the
  language \(\Gamma_X^{\df{K+1}}(G^\cap)\). Equivalently, we consider
  a grammar \(\mathcal{G}^{\df{K+1}}\), such
  that \(L_{X^{\tuple{0}}}(\mathcal{G}^{\df{K+1}})
  = \Gamma_X^{\df{K+1}}(G^\cap)\), where \(X^{\tuple{0}}\) is the
  ranked nonterminal corresponding to the initial state
  of \(A^{\df{K+1}}_{G^\cap}\) in
  Lemma \ref{fsa-dfk}. Let \(\mathcal{Q}
  = \tuple{\mathcal{G}^{\df{K+1}},X^{\tuple{0}},\sem{.}}\) be the
  program associated
  with \(\mathcal{G}^{\df{K+1}}\). If \(\mathcal{P}\) was assumed to
  be an octagonal program, then so is \(\mathcal{Q}\).

  The problem \(\foreach(\mathcal{P}_{s,x},\pat)\) is thus equivalent
  to the finite set of
  problems \(\foreach(\mathcal{Q},\widetilde{\Gamma}_i)\),
  for \(i=1,\ldots,m\). The size of \(\mathcal{G}^{\df{K+1}}\)
  is \[\len{\mathcal{G}^{\df{K+1}}} = \len{G^\cap}^{\mathcal{O}(K)} =
  (\len{\pat}^3 \cdot \len{G})^{\mathcal{O}(K)} =
  2^{\mathcal{O}(\len{\pat}^3 \cdot \len{G} \cdot (\log \len{\pat}
  + \log \len{G}))}\enspace.\] Hence the size of the input to each
  problem \(\foreach(\mathcal{Q},\widetilde{\Gamma}_i)\)
  is \(2^{\mathcal{O}(\len{\pat}^3 \cdot \len{G} \cdot
  (\log \len{\pat} + \log \len{G}))}\). Since \(\mathcal{Q}\) is a
  procedure-less octagonal program, and each such problem can be
  solved in \textsc{Nptime} \cite[Theorem 10]{bik14}, this provides
  a \textsc{Nexptime} decision procedure for the
  problem \(\foreach(\mathcal{P}_{s,x},\pat)\).

  We are left with proving that the \(\foreach(\mathcal{P},\pat)\)
  problem is in \textsc{Np}, when \(\sem{P}=\sem{P}^{(k)}\), for a
  constant \(k>0\). To this end, we define a
  grammar \(G_k=\tuple{\Vars \times \set{0,\bar{0},\ldots,k,\bar{k}}, \Sigma, \prod_k}\)
  such that \(L_X(G)^{(k)} = L_{(X,k)}(G_k)\) \cite[Definition
  3.1]{ls13}. Using the fact that, for each production
  $(Z,w)\in\prod$, there are at most two nonterminals in $w$, we
  establish that $\len{G_k}\leq 3k\len{G} + k(k+1)$, hence
  $\len{G_k}=\mathcal{O}(k^2 \cdot\len{G})$. 

  The corresponding program is \(\mathcal{P}_k
  = \tuple{G_k,(I,k),\sem{.}}\). By applying the reduction above, we
  obtain a set of
  problems \(\foreach(\mathcal{Q}_k,\widetilde{\Gamma}_i)\), each of
  which of size \((\len{\pat}^3 \cdot \len{G_k})^{\mathcal{O}(k)} =
  (\len{\pat}^3 \cdot (k^2 \cdot\len{G}))^{\mathcal{O}(k)}\). Since $k$ is
  constant, we can solve this problem in \textsc{Nptime}, using
  an \textsc{Np} procedure \cite[Theorem 10]{bik14}. Since
  the \textsc{Np}-hard lower bound was proved above, the problem
  is \textsc{Np}-complete. \qed
\end{proof}

\end{document}